\documentclass[10pt,a4paper,abstracton,toc=bibliography,DIV=12]{scrartcl} 
\usepackage{amsmath,amssymb,amsthm}
\usepackage{mathtools}
\usepackage{interval} 
\usepackage{stmaryrd} 
\usepackage{bbm} 
\usepackage[scr=boondox]{mathalfa} 
\usepackage{graphicx}
\usepackage{hyperref}

\newtheorem{theorem}{Theorem}[section]   
\newtheorem{proposition}[theorem]{Proposition}
\newtheorem{lemma}[theorem]{Lemma}

\newtheorem{hypothesis}[theorem]{Hypothesis}
                                         
\newcommand{\N}{\mathbbm{N}}
\newcommand{\Z}{\mathbbm{Z}}

\newcommand{\R}{\mathbbm{R}}
\newcommand{\C}{\mathbbm{C}}

\newcommand{\id}{\mathbbm{1}}            

\newcommand{\hilbert}{{\mathcal H}}      

\DeclareMathOperator{\im}{Im}            
\DeclareMathOperator{\re}{Re}            
\DeclareMathOperator{\sign}{sign}        

\DeclareMathOperator{\res}{res}          

\DeclareMathOperator{\dom}{dom}          
\DeclareMathOperator{\ran}{ran}          
\DeclareMathOperator{\mathspan}{span}    
\DeclareMathOperator{\rank}{rank}        
\DeclareMathOperator{\tr}{tr}            

\DeclareMathOperator{\arcosh}{arcosh}

\numberwithin{equation}{section}

\begin{document}
\title{Ground-state energy of one-dimensional free Fermi gases in the thermodynamic limit}
\author{
Peter Otte\thanks{peter.otte@rub.de} 
\ and\
Wolfgang Spitzer\thanks{wolfgang.spitzer@fernuni-hagen.de}\\ 
Fakult\"at f\"ur Mathematik und Informatik\\ 
FernUniversit\"at in Hagen\\ 
Germany}
\date{\today}
\maketitle
\begin{abstract}
We study the ground-state energy of one-dimensional, non-interacting fermions 
subject to an external potential in the thermodynamic limit. 
To this end, we fix some (Fermi) energy $\nu>0$, 
confine fermions with total energy below $\nu$ inside the interval $[-L,L]$ 
and study the shift of the ground-state energy due to the
potential $V$ in the thermodynamic limit $L\to\infty$.
We show that the difference $\mathcal{E}_L(\nu)$ of the two ground-state energies with and without
potential can be decomposed into a term
of order one (leading to the \emph{Fumi-term}) and a term of order
$1/L$, which yields the so-called \emph{finite size energy}.
We compute both terms for all possible boundary conditions explicitly and express them
through the scattering data of the one-particle Schr\"odinger operator
$-\Delta +V$ on $L^2(\mathbbm R)$.

\end{abstract}
\tableofcontents
%
\section{Introduction\label{introduction}}
In the seminal papers~\cite{Anderson1967a,Anderson1967} from 1967, Ph.~Anderson studied the
overlap between the ground states of $N$ non-interacting fermions with and without an external potential.
He showed that this overlap vanishes in the thermodynamic limit $N\to\infty$ at the rate $N^{-\gamma}$ and
expressed the orthogonality exponent $\gamma$ in terms of scattering data. 
This effect became known as Anderson's orthogonality catastrophe (AOC).

Beside the ground-state overlap an appropriate quantity to describe the change caused by the potential 
is the difference of the respective ground-state energies.
In 1994, I.~Affleck~\cite{Affleck1997,AffleckLudwig1994} proposed that the coefficient of the 
next-to-leading term of that energy difference in the thermodynamic limit can be identified with $\gamma$. 
His arguments were based upon ideas from conformal field theory.
In this paper we will thoroughly describe the asymptotics of the energy difference for one-dimensional (spin-less) fermions 
and thereby test Affleck's remarkable relation.

The infinite system is approached either by intervals $\Lambda_L\coloneqq [-L,L]$ or by $\Lambda_L \coloneqq [0,L]$. 
On such an interval, the free single-particle Hamiltonian $H_L\coloneqq -\Delta$ acting on the Hilbert space $L^2(\Lambda_L)$ 
(and depending on some boundary conditions, which we suppress in this notation) has eigenvalues 
$\lambda_{1,L}\leq\lambda_{2,L}\leq \ldots$ and (normalized) eigenfunctions $\varphi_{j,L}$. 
The ground state $\Phi_{N,L}$ of $N$ free fermions on $\Lambda_L$ is 
the anti-symmetric tensor product of the single particle (normalized) eigenfunctions $\varphi_{j,L}$. 
Hence, the energy of the ground state $\Phi_{N,L} = \varphi_{1,L} \wedge\cdots\wedge\varphi_{N,L}$ equals
\begin{equation}\label{intro01}
  \lambda^{(N)}_L = \lambda_{1,L}+\cdots+\lambda_{N,L} .
\end{equation}
Similarly, if $H_{V,L} \coloneqq H_L+V$ is another (or perturbed) single-particle Hamiltonian on the same Hilbert space $L^2(\Lambda_L)$ 
with eigenvalues $\mu_{j,L}$ (in ascending order) and eigenfunctions $\psi_{j,L}$ then the new ground state of 
$N$ fermions is $\Psi_{N,L} =\psi_{1,L}\wedge\cdots\wedge\psi_{N,L}$ with energy
\begin{equation}\label{intro02}
  \mu^{(N)}_L = \mu_{1,L}+\cdots+\mu_{N,L}  .
\end{equation}
The overlap mentioned in AOC is the square of the modulus of the scalar product, $|(\Phi_{N,L},\Psi_{N,L})|^2$. 
With growing $N$ and $L$ but with particle density $N/|\Lambda_L|$ converging to some fixed $\rho>0$, this
overlap equals to leading order $N^{-\gamma}$ with $\gamma>0$ as $N\to\infty$. 
It is conjectured (see \cite{Anderson1967a,GebertKuettlerMuellerOtte2016}) that
\begin{equation*}
  \gamma(\nu) = \frac{1}{\pi^2}\|\arcsin|T(\nu)/2|\|_2^2
\end{equation*}
where $T(\nu)\coloneqq S(\nu)-\id$ is the T-matrix and $S(\nu)$ the S-matrix at energy $\nu$. 
AOC has not been studied from a mathematical point until recently. In \cite{GebertKuettlerMuellerOtte2016}, it was proved
that $N^{-\gamma}$ with the above $\gamma$ is indeed an upper bound to the decay of the overlap. A smaller upper bound was proved earlier
in \cite{GebertKuettlerMueller2014} and in spatial dimension one in \cite{KuettlerOtteSpitzer2014} together with a lower bound of the
order $N^{-\gamma'}$ for some $\gamma'>\gamma$.

The difference of the ground-state energies is $\mu^{(N)}_L- \lambda^{(N)}_L$.
We perform the thermodynamic limit in a slightly different manner by fixing some (Fermi) energy $\nu>0$
and instead analyzing the \emph{canonical energy difference}
\begin{equation}\label{intro03}
  \mathcal{E}_L(\nu) \coloneqq \sum_{\mu_{j,L}\leq\nu} f(\mu_{j,L}) - \sum_{\lambda_{j,L}\leq\nu}f(\lambda_{j,L})
\end{equation}
in the limit $L\to\infty$. Here, we allow for a (holomorphic) weight function $f$, the most important case being $f(z)=z$. 
The particle number $N$ is recovered through the largest $\lambda_{N,L}\leq \nu$, see \ref{ssc}. It is easy to see
that $N/L\to\rho>0$ as $N,L\to\infty$.
It turns out that expanding $\mathcal{E}_L(\nu)$ in powers of $1/L$ does not quite cover the actual asymptotic behaviour 
which seems to contradict Affleck's relation when taken too strictly. Rather, leaving aside some technical details
(see \eqref{asymptotics01} for the precise statement) we decompose the energy difference into
\begin{equation*}
  \mathcal{E}_L(\nu) = -f(\nu)\xi_L(\nu) + \mathcal{E}_L^{\text{Fumi}}(\nu) + \mathcal{E}_L^{\text{FSE}}(\nu) .
\end{equation*}
Here, $\xi_L(\nu)$ is Kre\u\i{}n's spectral shift function, see \ref{ssf} and \ref{ssc}.
The so-called \emph{Fumi-term} $\mathcal{E}_L^{\text{Fumi}}(\nu)$ is of order one though with possible lower order correction terms.
The next-to-leading term is $\mathcal{E}_L^{\text{FSE}}(\nu)$, called \emph{finite size energy}, which is of order $1/L$
\begin{equation*}
  \mathcal{E}_L^{\text{FSE}}(\nu) = \mathcal{E}^{\text{FSE}}(\nu)\frac{1}{L} + o(\frac{1}{L}) .
\end{equation*}
According to Affleck, the coefficient $\mathcal{E}^{\text{FSE}}(\nu)$ is supposed to equal the orthogonality exponent $\gamma$.

The paper's main result is a rigorous analysis of the asymptotic behaviour of the Fumi-term and the finite size energy 
as well as explicit expressions for the limit terms. 

\begin{theorem}\label{intro01t}
Let $V\in L^1(\R)$ satisfy the Birman--Solomyak conditions
$V\in\ell^{\frac{1}{2}}(L^1(\R))$ and $X^2V\in\ell^{\frac{1}{2}}(L^1(\R))$, see \eqref{birman_solomyak}.
Then, $\mathcal{E}_L^{\text{Fumi}}(\nu)\to\mathcal{E}^{\text{Fumi}}(\nu)$ as $L\to\infty$ with the Fumi term
\begin{equation*}
  \mathcal{E}^{\text{Fumi}}(\nu)
      = \int_{-\infty}^\nu f'(\lambda)\xi(\lambda)\, d\lambda .
\end{equation*}
Here, $\xi$ is the spectral shift function for the infinite volume system (cf. \eqref{krein_xi}).
\end{theorem}

\begin{theorem}\label{intro02t}
Let $V\in L^1(\R)$ satisfy $X^3V\in L^1(\R)$. 
Assume that $e^{2\pi iL\sqrt{\nu}}\to e^{i\pi\eta}$, $\eta\in\interval[open right]{-1}{1}$, for $L\to\infty$.
Then, in this limit $L\mathcal{E}_L^{\text{FSE}}(\nu)\to\mathcal{E}^{\text{FSE}}(\nu)$ where
\begin{equation*}
  \mathcal{E}^{\text{FSE}}(\nu)
      = \frac{\sqrt{\nu}}{4\pi}f'(\nu) \tr\big[ \arccos^2(\re(e^{i\eta}\sigma_xU(\nu)^*S(\nu))) - \arccos^2(\re(e^{i\eta}\sigma_xU(\nu)^*)) \big]
\end{equation*}
with the boundary condition scattering matrix $U(\nu)$, see \eqref{fr_bcs}, and the Pauli matrix $\sigma_x$.
$S(\nu)$ is the scattering matrix at energy $\nu$, see Section \ref{scattering}).
\end{theorem}

The Fumi term $\mathcal{E}^{\text{Fumi}}(\nu)$ does not depend on the special sequence of system lengths $L$ used for the limit. 
In contrast, the finite size energy does. This subtlety was first noted by M.~Gebert \cite{Gebert2015} who analyzed
the corresponding system on the half-line with $\Lambda_L = [0, L]$ and Dirichlet boundary conditions at both endpoints. 
We return to the half-line model in Section \ref{hl} and recover Gebert's result 
although under slightly different conditions on the potential $V$.

Moreover, while the Fumi term $\mathcal{E}^{\text{Fumi}}(\nu)$ is the same for all boundary conditions
the finite size energy $\mathcal{E}^{\text{FSE}}(\nu)$ is not.
Through the boundary condition scattering matrix $U(\nu)$, see \eqref{fr_bcs},
it depends on the boundary conditions and, therefore, cannot
equal the orthogonality exponent $\gamma$ in general. However, AOC was essentially studied with Dirichlet boundary
conditions and its independence thereof has not yet been established. That leaves Affleck's relation still open.

The first appearance of the relation between the energy difference and the (integral of the) scattering 
data seems to be in Fumi's work~\cite{Fumi1955}, see also \cite{LangerAmbegaokar1961} and \cite[Sec. 4.1F]{Mahan1990}.
Shortly thereafter Fukuda and Newton~\cite{FukudaNewton1956} studied the difference of eigenvalues 
$\mu_{N,L}-\lambda_{N,L}$ and related the limit to scattering data.

Avoiding the thermodynamic limit, Frank, Lewin, Lieb, and Seiringer~\cite{FrankLewinLiebSeiringer2011}
studied the energy difference directly in the infinite-volume system
(cf. \eqref{energy_difference}) and proved quantitative semiclassical bounds reminiscent of Lieb--Thirring bounds.

The proofs of Theorems \ref{intro01t} and \ref{intro02t} are presented in Section \ref{asymptotics}.
We start off by using a generalization of Riesz's integral formula due to Dunford 
to express the energy difference $\mathcal{E}_L(\nu)$ in terms of an integral of the perturbation determinant 
which involves the Birman--Schwinger operator $K_L(z)\coloneqq \sqrt{|V|}R_L(z)\sqrt{|V|}J$, see Proposition \ref{energy01t}.
Here the potential is written in the Rollnik form, $V=\sqrt{|V|}J\sqrt{|V|}$ and 
$R_L(z)$ is the resolvent of the free Schr\"odinger operator on the bounded domain $\Lambda_L$.
The latter can be written such as to separate the infinite volume part from the boundary conditions, 
$R_L(z) = R_{\infty,L}(z) -D_L(z)$.
$R_{\infty,L}(z)$ is the resolvent of the free Schr\"odinger operator on the full domain 
(either $\R$ or $\interval[open right]{0}{\infty}$) restricted to $\Lambda_L$. 
We allow all possible boundary conditions for the Schr\"odinger operator on $\Lambda_L$. 
The operator $D_L(z)$ is trace class and, in dimension one, even a rank two operator, respectively a rank one operator.
An equally important fact is that this representation of $R_L(z)$ eventually yields a natural decomposition of $\mathcal E_L(\nu)$ 
into the Fumi term $\mathcal{E}^{\text{Fumi}}(\nu)$ and the finite size energy $\mathcal{E}^{\text{FSE}}(\nu)$.

Besides the Birman--Schwinger operator $K_L(z)$ corresponding to the (free) Schr\"odinger operator on $\Lambda_L$ 
there is a locally reduced Birman--Schwinger operator $K_{\infty,L}(z)$, 
which stems from the resolvent $R_\infty(z)$ of the (free) Schr\"odinger operator on $\R$, 
respectively $\interval[open right]{0}{\infty}$. 
In Lemma \ref{wo03t} we find an expression of $K_L(z)$ in terms of $K_{\infty,L}(z)$ 
and a factorization of the corresponding perturbation determinants
which allows us to derive the limiting behaviour of the energy difference, see Section \eqref{asymptotics}.

Our mathematical approach via the resolvent is related to the work of Gesztesy and Nichols~\cite[Lemma 3.2]{GesztesyNichols2012},
who proved that
\begin{equation*}
  \det(\id-\sqrt{V}R_L(z)\sqrt{V}) \to \det(\id-\sqrt{V}R_\infty(z)\sqrt{V}) \ \text{as}\ L\to\infty,\
    z\in\C\setminus\interval[open right]{0}{\infty} .
\end{equation*}
Since we are interested in the $1/L$-correction to the leading Fumi-term 
we need more information (uniformly in $z$) than just the above limit of determinants. 
In fact, we need the rate of approach to this limit. 

Although we treat here only one-dimensional systems, our approach via
Riesz's integral formula allows us to treat higher dimensional systems
with non-radial potentials. We plan to pursue this in the future.

\section{Representation of the energy difference \texorpdfstring{$\mathcal{E}(\nu)$}{}\label{energy}}
Let $\hilbert$ be a separable complex Hilbert space with scalar product $(\cdot,\cdot)$ which is linear in the second argument. 
We denote by $B_p(\hilbert)$, $1\leq p<\infty$, the Schatten--von Neumann classes with norm $\|\cdot\|_p$.
Furthermore, $B(\hilbert)$ are the bounded operators with operator norm $\|\cdot\|$ occasionally referred to as the $p=\infty$-norm.
These norms satisfy Jensen's inequality
\begin{equation}\label{jensen_inequality}
   \| A\|_q \leq \|A\|_p,\ 1\leq p\leq q\leq \infty
\end{equation}
and H\"older's inequality
\begin{equation}\label{hoelder_inequality}
   \|A B\|_r \leq \|A\|_p\|B\|_q,\ 1\leq r,p,q\leq \infty,\
   \frac{1}{r} = \frac{1}{p}+\frac{1}{q} .
\end{equation}
The only cases needed herein are $p=1,2$, the trace class and Hilbert--Schmidt operators, respectively, 
and $p=\infty$.

Let $H:\dom(H)\to\hilbert$, $\dom(H)\subset\hilbert$, be a self-adjoint operator
and $V:\hilbert\to\hilbert$ be (for simplicity) a bounded symmetric, hence self-adjoint, operator. Its polar decomposition
can be written in a form first used by Rollnik \cite{Rollnik1956} when studying the Lippmann--Schwinger equation (cf. Section \ref{scattering})
\begin{equation}\label{V_polar_decomposition}
  V = \sqrt{|V|}J\sqrt{|V|},\ J^*=J,\ J^2=\id,\ \|J\|=1 ,
\end{equation}
where $\id$ denotes the unity operator. Then, $H_V\coloneqq H+V$ is self-adjoint as well with $\dom(H_V)=\dom(H)$. 
We denote by $\sigma(H)$ and $\sigma(H_V)$ the spectrum of $H$ and $H_V$, respectively, by
\begin{equation}\label{resolvents}
  R(z) \coloneqq (z\id-H)^{-1},\ z\in\C\setminus\sigma(H),\ 
  R_V(z) \coloneqq (z\id- H_V)^{-1},\ z\in\C\setminus\sigma(H_V)
\end{equation}
their resolvents and their spectral families by
\begin{equation}\label{spectral_projections}
  P_\nu \coloneqq \chi_{\interval[open left]{-\infty}{\nu}}(H),\   
  \Pi_\nu \coloneqq \chi_{\interval[open left]{-\infty}{\nu}}(H_V),\ \nu\in\R ,
\end{equation}
where $\chi_I$ is the indicator function of the interval $I\subset\R$.
The paper's central quantity is the energy difference
\begin{equation}\label{energy_difference}
  \mathcal{E}(\nu) \coloneqq \tr \big[f(H_V)\Pi_\nu - f(H)P_\nu\big]
\end{equation}
for some fixed holomorphic function $f:\C\to\C$.
It is well-defined under appropriate conditions on the operators $H$ and $H_V$.

\subsection{Perturbation determinant \texorpdfstring{$\Delta(z)$}{}\label{pd}}
There is a handy formula for $\mathcal{E}(\nu)$ based upon Riesz's integral formula
for spectral projections (or the Dunford integral) and the perturbation determinant.

The resolvents in \eqref{resolvents} are related by Kre\u\i{}n's resolvent formula for additive perturbations
\begin{equation}\label{krein_Omega}
  R_V(z) - R(z) = R(z) \sqrt{|V|} J\Omega(z)\sqrt{|V|} R(z),\ z\in\C\setminus(\sigma(H)\cup\sigma(H_V)) ,
\end{equation}
which holds true whenever the following inverse
\begin{equation}\label{Omega_K}
      \Omega(z) \coloneqq (\id-K(z))^{-1},\ K(z) \coloneqq \sqrt{|V|}R(z)\sqrt{|V|}J ,
\end{equation}
exists. Motivated by stationary scattering theory, see e.g. \cite[3.6.1]{Thirring2002}, we call $\Omega(z)$ the \emph{wave operator}.
The sandwiched resolvent $K(z)$ is called \emph{Birman--Schwinger operator}. We used Rollnik's factorization \eqref{V_polar_decomposition}
since in our applications the operator $R(z)V$ generally will not have the necessary trace class properties.
The wave operator is closely related to the spectrum of the perturbed operator, which is known as Birman--Schwinger principle.

\begin{lemma}\label{energy_omega}
Let $z\in\C\setminus\sigma(H)$. Then, $\Omega(z)$ exists and is bounded if and only if
$z\notin\sigma(H_V)$ in which case the map $z\mapsto\Omega(z)$ is holomorphic with respect 
to the operator norm and we have the formula
\begin{equation}\label{energy_omega01}
  \Omega(z) = \id + \sqrt{|V|} R_V(z) \sqrt{|V|}J .
\end{equation}
\end{lemma}
\begin{proof}
(i) 
Let $z\notin\sigma(H_V)$ which means that $R_V(z)$ exists as a bounded operator. Then,
\begin{equation*}
\begin{split}
\lefteqn{ (\id - \sqrt{|V|}R(z)\sqrt{|V|}J)(\id + \sqrt{|V|}R_V(z)\sqrt{|V|}J) }\\
   & = \id - \sqrt{|V|}R(z)\sqrt{|V|}J + \sqrt{|V|} ( \id - R(z)V ) R_V(z)\sqrt{|V|}J\\
   & = \id - \sqrt{|V|}R(z)\sqrt{|V|}J + \sqrt{|V|}R(z) (z\id - H - V ) R_V(z)\sqrt{|V|}\\
   & = \id .
\end{split}
\end{equation*}
If the factors are interchanged the calcutions are analogous. We conclude that $\Omega(z)$ exists and is bounded.
Formula \eqref{energy_omega01} is obvious.
 
It is well-known that the map $z\mapsto R(z)$ is holomorphic for $z\notin\sigma(H)$.
Since $V$ is bounded a Neumann series argument shows that the map $z\mapsto\Omega(z)$ is holomorphic 
on $\C\setminus( \sigma(H)\cup\sigma(H_V) )$.

(ii)
Let $z\in\sigma(H_V)\subset\R$ and $z_n\in\C\setminus\R$ with $z_n\to z$ for $n\to\infty$.
From (i) we know $\|\Omega(z_n)\|<\infty$.
If $\Omega(z)$ existed as a bounded operator it would follow from (i) that $\Omega(z_n)\to \Omega(z)$. In particular,
$\sup_n \|\Omega(z_n)\| < \infty$.
On the other hand, we have $\|R_V(z_n)\|<\infty$ and $\|R_V(z_n)\|\to\infty$ as $z_n\to z$, which yields a contradiction
via \eqref{energy_omega01}.
\end{proof}

The spectrum of the Birman--Schwinger operator can be described a tad more detailed.

\begin{lemma}\label{spectrum_K}
Let $z\in\C\setminus\R$ and $\varkappa\in\C\setminus\{0\}$ be an eigenvalue of $K(z)$. 
Then $\im(\varkappa)\neq 0$.
\end{lemma}
\begin{proof}
Since $\im(z)\neq 0$ the resolvent $R(z)$ and, hence, the Birman--Schwinger operator is bounded.
Let $\varkappa\in\C\setminus\{0\}$ and let $\varphi\in\hilbert$, $\varphi\neq 0$, such that
\begin{equation*}
  \varkappa\varphi = K(z)\varphi = \sqrt{|V|}R(z)\sqrt{|V|}J \varphi .
\end{equation*}
We multiply by $J$, take scalar products
\begin{equation*}
  \varkappa (\varphi,J\varphi) = (\varphi,J\sqrt{|V|}R(z)\sqrt{|V|}J\varphi) ,
\end{equation*}
and look at the imaginary part
\begin{equation*}
  \im(\varkappa) (\varphi,J\varphi) = -\im(z) \|R(z)\sqrt{|V|}J\varphi\|^2 .
\end{equation*}
The norm on the right-hand side does not vanish since that would imply
$\varkappa\varphi=0$ by the eigenvalue equation which contradicts $\varphi\neq 0$ and 
$\varkappa\neq 0$. Hence, $\im(z)\neq 0$ implies $\im(\varkappa)\neq 0$.
\end{proof}

We note the relevant trace class properties. Because detailed accounts are scattered
throughout the literature we sketch the proofs for the sake of completeness.

\begin{lemma}\label{energy00t}
Let $\sqrt{|V|}R(z_0)\in B_2(\hilbert)$ for some $z_0\in\C\setminus\sigma(H)$.
Then, the following hold true
\begin{enumerate}
\item
$\sqrt{|V|}R(z)\in B_2(\hilbert)$ and $R(z)\sqrt{|V|}\in B_2(\hilbert)$ for all $z\in\C\setminus\sigma(H)$
and the maps $z\mapsto\sqrt{|V|}R(z)$ and $z\mapsto R(z)\sqrt{|V|}$ are continuous with respect to the 
Hilbert--Schmidt norm.
\item
If in addition $\sqrt{|V|}R(z_0)\sqrt{|V|}\in B_1(\hilbert)$ then $\sqrt{|V|}R(z)\sqrt{|V|}\in B_1(\hilbert)$
for all $z\in\C\setminus\sigma(H)$ and the map $z\mapsto\sqrt{|V|}R(z)\sqrt{|V|}$ is holomorphic
with respect to the trace norm.
\item
Furthermore, $R_V(z)-R(z)\in B_1(\hilbert)$ for $z\in\C\setminus(\sigma(H)\cup\sigma(H_V))$ and
the map $z\mapsto R_V(z)-R(z)$ is continuous with respect to the trace norm.
\end{enumerate}
\end{lemma}
\begin{proof}
1.
The first resolvent identity multiplied by $\sqrt{|V|}$ yields
\begin{equation*}
  \sqrt{|V|}R(z) = \sqrt{|V|}R(z_0)\big[\id  - (z-z_0) R(z)\big],\ z_0,z\in\C\setminus\sigma(H) .
\end{equation*}
The right-hand side is the product of the Hilbert--Schmidt operator $\sqrt{|V|}R(z_0)$ 
and a bounded operator. Hence, $\sqrt{|V|}R(z)\in B_2(\hilbert)$. 
Similarly, by the first resolvent identity and H\"older's inequality \eqref{hoelder_inequality}
\begin{equation*}
  \|\sqrt{|V|}R(w) - \sqrt{|V|}R(z)\|_2 \leq |w-z| \|\sqrt{|V|}R(z)\|_2 \|R(w)\| ,
\end{equation*}
which shows Lipschitz-continuity with respect to  the Hilbert--Schmidt norm since the map $w\mapsto\|R(w\|$ is continuous.
Finally, use $R(z)\sqrt{|V|}=(\sqrt{|V|}R(\bar z))^*$ to show the corresponding statements.

2.
As in 1. we obtain via the first resolvent identity
\begin{equation*}
  \sqrt{|V|}R(z)\sqrt{|V|} = \sqrt{|V|}R(z_0)\sqrt{|V|} - (z-z_0) \sqrt{|V|}R(z_0)R(z)\sqrt{|V|},\ 
     z_0,z\in\C\setminus\sigma(H) .
\end{equation*}
The right-hand side is trace class. Furthermore, by H\"older's inequality \eqref{hoelder_inequality}
\begin{equation*}
  \|\sqrt{|V|}R(w)\sqrt{|V|} - \sqrt{|V|}R(z)\sqrt{|V|}\|_1 \leq  |w-z| \| \sqrt{|V|}R(z)\|_2 \|R(w)\sqrt{|V|}\|_2 ,\
     z,w\in\C\setminus\sigma(H),
\end{equation*}
which shows Lipschitz-continuity with respect to the trace norm.  Moreover,
\begin{equation*}
  \frac{1}{z-w}\big[ \sqrt{|V|}R(w)\sqrt{|V|} - \sqrt{|V|}R(z)\sqrt{|V|}\big] 
    = \sqrt{|V|}R(z)R(w)\sqrt{|V|} .
\end{equation*}
Since the right-hand side converges in trace norm to $\sqrt{|V|}R(z)R(z)\sqrt{|V|}$ so does the difference
quotient on the left-hand side.

3.
In Kre\u\i{}n's resolvent formula \eqref{krein_Omega} the right-hand side is the
product of the Hilbert--Schmidt operators $R(z)\sqrt{|V|}J$, $\sqrt{|V|}R(z)$, see part 1.,
and the bounded operator $\Omega(z)$, Lemma \ref{energy_omega},
and thus a trace-class operator. Furthermore, the operators depend continuously on $z$ in the respective
norms and so does their product.
\end{proof}

The applications we have in mind require that the trace-class properties are valid up to the real axis
or in other words for $z\in\sigma(H)$ in Lemma \ref{energy00t}.

\begin{hypothesis}[Limiting absorption principle]\label{h1}
Let $\nu\in\R$.
\begin{enumerate}
\item 
For $s>0$ the Birman--Schwinger operators $\sqrt{|V|}R(\nu+is)\sqrt{|V|}J\in B_1(\hilbert)$ and the limit
$s\to+0$ exists with respect to $B_1(\hilbert)$.
\item 
For $s>0$ the operators $R_V(\nu+is)-R(\nu+is)\in B_1(\hilbert)$ and the limit $s\to+0$ exists 
with respect to $B_1(\hilbert)$.
\end{enumerate}
\end{hypothesis}

Note that by the resolvent property $R(z)^* = R(\bar z)$ Hypothesis \ref{h1} implies the analogous statements
for the lower half-plane, $s<0$, albeit the respective limits need not be the same. 
Now, with Lemma \ref{energy00t} and Hypothesis \ref{h1} in mind, we define
the (modified) perturbation determinant (cf. \cite[(0.9.35)]{Yafaev2010})
\begin{equation}\label{perturbation_determinant}
  \Delta(z) \coloneqq \det(\id - K(z)) .
\end{equation}
It is closely related to Kre\u\i{}n's spectral shift function $\xi$ (see Section \ref{ssf}).
Its behaviour at the non-essential spectrum is analogous to the finite dimensional case.

\begin{lemma}\label{energy00at}
Assume that $R(z)^{\frac{1}{2}}\sqrt{|V|}\in B_2(\hilbert)$ for all $z\in\C\setminus\sigma(H)$.
Let $\lambda\in\C$ be an isolated eigenvalue of $H$ and $H_V$  with finite multiplicities $m,n\in\N_0$, respectively.
Here multiplicity $0$ means that $\lambda$ is a regular value.
Then, for $z\neq 0$ the perturbation determinant can be factorized
\begin{equation*}
  \Delta(z) = \Big(1-\frac{\lambda}{z}\Big)^{n-m} \Delta_\lambda(z)
\end{equation*}
where $\Delta_\lambda(z)$ is continuous and non-vanishing in a neighbourhood of $\lambda$.
\end{lemma}
\begin{proof}
Cf. \cite[IV.3.4]{GohbergKrein1969}. 
Since $R(z)^{\frac{1}{2}}\sqrt{|V|}\in B_2(\hilbert)$ we know that $R(z)^{\frac{1}{2}}VR(z)^{\frac{1}{2}}\in B_1(\hilbert)$ and thus
\begin{equation*}
  \Delta(z) = \det( \id -R(z)^{\frac{1}{2}}VR(z)^{\frac{1}{2}} )
\end{equation*}
where the square root $R(z)^{\frac{1}{2}}$ is defined via the functional calculus.
Determinants of this type were also studied in \cite{GesztesyZinchenko2012}.
Furthermore,
\begin{equation*}
  \id -R(z)^{\frac{1}{2}}VR(z)^{\frac{1}{2}}  = R(z)^{\frac{1}{2}}(z\id - H_V)R(z)^{\frac{1}{2}} .
\end{equation*}
Let $\Pi_\lambda$ be the spectral projection of $H_V$ corresponding to $\lambda$ and $\Pi_\lambda^\perp\coloneqq\id-\Pi_\lambda$. 
From the simple relation
\begin{equation*}\
  z\id-H_V =  (\id- \frac{1}{z}H_V\Pi_\lambda)(z\id-H_V\Pi_\lambda^\perp),\ z\neq 0,
\end{equation*}
we obtain
\begin{equation*}
\begin{split}
  \id -R(z)^{\frac{1}{2}}VR(z)^{\frac{1}{2}}
     & = R(z)^{\frac{1}{2}}(\id-\frac{1}{z}H_V\Pi_\lambda)(z\id-H_V\Pi_\lambda^\perp)R(z)^{\frac{1}{2}}\\
     & = R(z)^{\frac{1}{2}}(\id-\frac{1}{z}H_V\Pi_\lambda)(z\id-H)^{\frac{1}{2}}\times
            R(z)^{\frac{1}{2}}(z\id-H_V\Pi_\lambda^\perp)R(z)^{\frac{1}{2}}\\
     & = ( \id - \frac{\lambda}{z}R(z)^{\frac{1}{2}}\Pi_\lambda(z\id-H)^{\frac{1}{2}} )\times
            R(z)^{\frac{1}{2}}(z\id-H_V\Pi_\lambda^\perp)R(z)^{\frac{1}{2}} .
\end{split}
\end{equation*}
The determinant of the second factor exists. For,
\begin{equation*}
 R(z)^{\frac{1}{2}}(z\id-H_V\Pi_\lambda^\perp)R(z)^{\frac{1}{2}}
    = \id - R(z)^{\frac{1}{2}}VR(z)^{\frac{1}{2}} + \lambda R(z)^{\frac{1}{2}} \Pi_\lambda R(z)^{\frac{1}{2}} .
\end{equation*}
Here, the term with $V$ is trace class by assumption and the term with $\Pi_\lambda$ is a finite rank operator.
Hence, we can write our determinant as
\begin{equation*}
  \Delta(z) = \det\Big(\id-\frac{\lambda}{z}\Pi_\lambda\Big) \det\big[ (z\id-H)^{-\frac{1}{2}}(z\id-H_V\Pi_\lambda^\perp)(z\id-H)^{-\frac{1}{2}} \big] .
\end{equation*}
Let likewise $P_\lambda$ be the spectral projection of $H$ corresponding to $\lambda$ and $P_\lambda^\perp\coloneqq\id-P_\lambda$. Write
\begin{equation*}
  z\id-H = (\id-\frac{1}{z}HP_\lambda)(z\id-HP_\lambda^\perp),\ z\neq 0  .
\end{equation*}
Then,
\begin{equation*}
  R(z) = (\id-\frac{1}{z}HP_\lambda)^{-1} \tilde R(z),\  \tilde R(z) \coloneqq (z\id-HP_\lambda^\perp)^{-1}  ,
\end{equation*}
which implies that
\begin{equation*}
  \tilde R(z)^{\frac{1}{2}}\sqrt{|V|} = (\id-\frac{\lambda}{z}P_\lambda)^{\frac{1}{2}} R(z)^{\frac{1}{2}}\sqrt{|V|}
     \in B_2(\hilbert),\ z\in\C\setminus(\sigma(H)\cup\{0\}) .
\end{equation*}
Furthermore, from
\begin{equation*}
  H_V\Pi_\lambda^\perp = HP_\lambda^\perp + V + \lambda P_\lambda-\lambda\Pi_\lambda
\end{equation*}
we obtain
\begin{equation*}
  \tilde R(z)^{\frac{1}{2}}(z\id-H_V\Pi_\lambda^\perp)\tilde R(z)^{\frac{1}{2}}
    = \id - \tilde R(z)^{\frac{1}{2}}(V+\lambda P_\lambda - \lambda\Pi_\lambda)\tilde R(z)^{\frac{1}{2}} .
\end{equation*}
Since $P_\lambda$ and $\Pi_\lambda$ are finite rank operators we infer that the determinant
\begin{equation*}
   \Delta_\lambda(z) 
      \coloneqq \det [ \tilde R(z)^{\frac{1}{2}}(z\id-H_V\Pi_\lambda^\perp)\tilde R(z)^{\frac{1}{2}} ]
      = \det [ \id - \tilde R(z)^{\frac{1}{2}}(V+\lambda P_\lambda - \lambda\Pi_\lambda)\tilde R(z)^{\frac{1}{2}}]
\end{equation*}
is well-defined. Therefore,
\begin{equation*}
  \Delta(z) = \det[\id-\frac{\lambda}{z}\Pi_\lambda ] \det[\id-\frac{\lambda}{z}P_\lambda]^{-1} \Delta_\lambda(z)  .
\end{equation*}
We study the behaviour for $z\to\lambda$. To begin with,
\begin{equation*}
  \tilde R(z)^{\frac{1}{2}}\sqrt{|V|}
     =   \tilde R(w)^{\frac{1}{2}}\sqrt{|V|}
           + \frac{1}{2} (w-z) \tilde R(z)^{\frac{1}{2}} \int_0^1 \tilde R(tw+(1-t)z)^{\frac{1}{2}}\, dt
                          \times \tilde R(w)^{\frac{1}{2}}\sqrt{|V|}
\end{equation*}
where $w\in\C\setminus\sigma(H)$ such that $tw+(1-t)z\in\C\setminus\sigma(H)$.
Using H\"older's inequality \ref{hoelder_inequality} we conclude that the map 
$z\mapsto\tilde R(z)^{\frac{1}{2}}\sqrt{|V|}\in B_2(\hilbert)$ is continuous in a neighbourhood of $\lambda$. 
Therefore, the map $z\mapsto\Delta_\lambda(z)$ is continuous in that neighbourhood as well.
Finally, one can easily check that $1$ is not an eigenvalue of 
$\tilde R(\lambda)^{\frac{1}{2}}(V+\lambda P_\lambda - \lambda\Pi_\lambda)\tilde R(\lambda)^{\frac{1}{2}}$ and 
hence $\Delta_\lambda(\lambda)\neq 0$. Along with the continuity of $\Delta_\lambda(z)$ that
proves the statement.
\end{proof}

In order to allow the Fermi energy $\nu$ to be in the essential spectrum as well we extend 
the Dunford integral formula for holomorphic functions of self-adjoint operators.

\begin{lemma}\label{dunford01t}
Let the self-adjoint operator $H$ be bounded from below. Let $\nu\in\R$ and assume that 
the spectral family $E_\lambda$ is continuous in a neighbourhood of $\nu$. 
Let $P$ be the spectral projection to the set $\interval[open left]{-\infty}{\nu}$.
Let $\Gamma$ be a closed contour crossing the real axis perpendicularly at $\nu$. Then,
\begin{equation}\label{dunford01t01}
  f(H)P = \int_\Gamma f(z) R(z)\, dz
\end{equation}
where the integral at $\nu$ is to be understood as a Cauchy principal value.
\end{lemma}
\begin{proof}
By the spectral theorem
\begin{equation}\label{dunford01t02}
  f(H)P = \int_{-\infty}^\nu f(\lambda)\, dE_\lambda .
\end{equation}
Since $H$ is bounded from below the interval of integration is actually finite.
We want to replace $f$ via Cauchy's integral formula. For $\lambda\neq\nu$,
\begin{equation}\label{dunford01t03}
  \frac{1}{2\pi i} \int_{\Gamma\setminus\Gamma_\delta} \frac{f(z)}{z-\lambda} \, dz =
\begin{cases}
  f(\lambda) - F_\delta(\lambda) & \lambda \ \text{inside}\ \Gamma , \\
             - F_\delta(\lambda) & \lambda \ \text{outside}\ \Gamma
\end{cases}
\end{equation}
where
\begin{equation*}
   F_\delta(\lambda) =  \frac{1}{2\pi i} \int_{\Gamma_\delta} \frac{f(z)}{z-\lambda} \, dz,\
   \Gamma_\delta \coloneqq \{ z\in\Gamma \mid |z-\nu| \leq \delta \} .
\end{equation*}
By our assumption on $\Gamma$ and with the aid of Cauchy's integral theorem, $\Gamma_\delta$ can be replaced
by a straight line. For the sake of simplicity we assume that it already is one. Hence,
\begin{equation*}
  z\in\Gamma_\delta,\ z(s) = \nu + is,\ -\delta \leq s \leq\delta .
\end{equation*}
By Taylor's formula,
\begin{equation*}
\begin{split}
  F_\delta(\lambda)
    & = \frac{1}{2\pi}\int_{-\delta}^\delta f(\nu+is)\frac{1}{\nu+is-\lambda}\, ds \\
    & = \frac{1}{2\pi} f(\nu) \int_{-\delta}^\delta \frac{1}{\nu+is-\lambda}\, ds +
       \frac{1}{2\pi}  \int_{-\delta}^\delta s f'(\xi(s))\frac{1}{\nu+is-\lambda}\, ds .
\end{split}
\end{equation*}
The first integral can be rewritten
\begin{equation*}
  \int_{-\delta}^\delta \frac{1}{\nu+is-\lambda}\, ds
     =  \int_0^\delta  \frac{\nu-\lambda}{(\nu-\lambda)^2+s^2}\, ds
     =  \int_0^{\frac{\delta}{|\nu-\lambda|}} \frac{1}{1+s^2}\, ds
\end{equation*}
and the second integral can be estimated
\begin{equation*}
 \big|  \int_{-\delta}^\delta s f'(\xi(s))\frac{1}{\nu+is-\lambda}\, ds \big|
   \leq \sup_{z\in\Gamma_\delta}|f'(z)| \int_{-\delta}^\delta \frac{|s|}{((\nu-\lambda)^2+s^2)^{\frac{1}{2}}}\, ds
   \leq 2\delta \sup_{z\in\Gamma_\delta}|f'(z)| .
\end{equation*}
Hence, $|F_\delta(\lambda)|\leq C$ for all $\lambda\in\R$ with some constant $C\geq 0$
and $F_\delta(\lambda)\to 0$ as $\delta\to 0$ for all $\lambda\neq\nu$. We insert \eqref{dunford01t03}
into \eqref{dunford01t02} and obtain
\begin{equation*}
\begin{split}
  f(H)P & = \frac{1}{2\pi i} \int_\R \int_{\Gamma\setminus\Gamma_\delta} \frac{f(z)}{z-\lambda}\, dz\, dE_\lambda +
          \int_\R F_\delta(\lambda)\, dE_\lambda \\
        & = \frac{1}{2\pi i} \int_{\Gamma\setminus\Gamma_\delta} f(z) \int_\R \frac{1}{z-\lambda}\, dE_\lambda\, dz
             +  \int_\R F_\delta(\lambda)\, dE_\lambda\\
        & = \frac{1}{2\pi i} \int_{\Gamma\setminus\Gamma_\delta} f(z) R(z)\, dz
             +  \int_\R F_\delta(\lambda)\, dE_\lambda .
\end{split}
\end{equation*}
By Lebesgue's Theorem, the integral over $F_\delta$ tends to zero as $\delta\to 0$. 
Thus, the first one yields \eqref{dunford01t01}.
\end{proof}

In what follows we will apply Lemmas \ref{energy00at} and \ref{dunford01t} only to operators
with special types of spectra which makes it reasonable to
single out the corresponding assumptions on the Fermi energy $\nu$.

\begin{hypothesis}\label{h2}
Let $\nu\in\R$ and let $P_\nu$ and $\Pi_\nu$ the spectral projections of $H$ and $H_V$ as in \eqref{spectral_projections}.
Then, at least one of the following holds true.
\begin{enumerate}
\item $\dim\ran P_\nu<\infty$ and $\dim\ran\Pi_\nu<\infty$.
\item $P_\lambda$ and $\Pi_\lambda$ are continuous in a neighbourhood of $\nu$.
\end{enumerate}
\end{hypothesis}

The conditions in Hypothesis \ref{h2} are not independent since $\dim\ran P_\nu=0=\dim\ran\Pi_\nu$ in 1.
obviously implies 2. Nor are they exhaustive in that embedded eigenvalues are not included.
Now, everything is at hand to derive the aforementioned formula for the energy difference.

\begin{proposition}\label{energy01t}
Let $H$ and $H_V$ be semi-bounded from below and let $\Gamma_\nu$ be a closed contour in $\C$ 
intersecting the real axis perpendicularly at $\nu\in\R$ and below $\sigma(H)$ and $\sigma(H_V)$ 
(cf. Figure \ref{f_fermi_parabola}, p. \pageref{f_fermi_parabola}).
Assume that Hypothesis \ref{h1} holds true.
Then, the energy difference $\mathcal{E}(\nu)$ from \eqref{energy_difference} satisfies
\begin{equation}\label{energy01t01}
  \mathcal{E}(\nu)
      =   - f(\nu)\xi(\nu) - \frac{1}{2\pi i} \int_{\Gamma_\nu} f'(z) \ln[\Delta(z)]\, dz
\end{equation}
where $f:\C\to\C$ is a holomorphic function. 
\end{proposition}
\begin{proof}
(i)
To begin with, let $\varkappa\in\R$ such that the conditions of Lemma \ref{dunford01t} are satisfied.
Then, we can express the energy difference through a Dunford integral
\begin{equation*}
    f(H_V)\Pi_\varkappa - f(H)P_\varkappa = \frac{1}{2\pi i} \int_{\Gamma_\varkappa} f(z)[ R_V(z) - R(z) ] \, dz  
\end{equation*}
which is a Riemann integral with respect to the operator norm.
By Hypothesis \ref{h1} the integrand is piecewise continuous with respect to the trace norm.
Therefore, the right-hand side is a trace class operator. We, thus, may take the trace and interchange it
with the integration
\begin{equation}\label{energy01t02}
   \mathcal{E}(\varkappa)
    = \tr [f(H_V)\Pi_\varkappa - f(H)P_\varkappa ]
    = \frac{1}{2\pi i} \int_{\Gamma_\varkappa} f(z) \tr [ R_V(z) - R(z) ] \, dz .
\end{equation}
This trace is the logarithmic derivative of the perturbation determinant (see e.g. \cite[(8.1.4)]{Yafaev1992})
\begin{equation*}
  \tr [ R_V(z) - R(z) ] =  \frac{d}{dz} \ln[\Delta(z)] .
\end{equation*}
Thereby, via an integration by parts \eqref{energy01t02} becomes 
\begin{equation}\label{energy01t03}
\begin{split}
  \mathcal{E}(\varkappa)
     & = \frac{1}{2\pi i} \int_{\Gamma_\varkappa} f(z) \frac{d}{dz} \ln[\Delta(z)] \, dz\\
     & = \frac{1}{2\pi i}\lim_{\varepsilon\to+0} [ f(\varkappa-i\varepsilon)\ln[\Delta(\varkappa-i\varepsilon)] 
                                                -f(\nu+i\varepsilon)\ln[\Delta(\varkappa+i\varepsilon)] ]
       - \frac{1}{2\pi i} \int_{\Gamma_\varkappa} f'(z)\ln[\Delta(z)] \, dz\\
     & = - f(\varkappa) \xi(\varkappa) - \frac{1}{2\pi i} \int_{\Gamma_\varkappa} f'(z)\ln[\Delta(z)] \, dz .
\end{split}
\end{equation}
Here we used the continuity of $f$ and Kre\u\i{}n's formula for the spectral shift function \eqref{krein_xi}.

(ii)
If $\nu\in\R$ satisfies the conditions of Lemma \ref{dunford01t} we may choose $\varkappa\coloneqq\nu$ in \eqref{energy01t03}
to prove \eqref{energy01t01}.

(iii)
If $\nu\in\R$ is an isolated eigenvalue of $H$ or $H_V$ we take $\varkappa\coloneqq\nu+\delta$
where $\delta>0$ is so small that $\interval[open left]{\nu}{\nu+\delta}\cap(\sigma(H)\cup\sigma(H_V))=\emptyset$.
In the following considerations we assume that near $\nu$ the contour $\Gamma_\nu$ is a straight line 
\begin{equation*}
   \Gamma_\nu = \tilde\Gamma_\nu \cup \{ z=\nu+is \mid -b\leq s\leq b\},\ b>0,
\end{equation*}
which can be justified by means of standard arguments. Let
\begin{equation*}
  \Gamma_\varkappa \coloneqq \Gamma_{\nu,\delta} 
                 \coloneqq \tilde\Gamma_\nu \cup \Gamma_{\nu,\delta}^{(1)} \cup \Gamma_{\nu,\delta}^{(2)},\
    \Gamma_{\nu,\delta}^{(1)}\coloneqq \{ z = \nu+is \mid \delta \leq |s| \leq b\} ,\ 
    \Gamma_{\nu,\delta}^{(2)}\coloneqq \{ z = \nu+\delta e^{i\vartheta} \mid -\frac{\pi}{2} \leq \vartheta \leq \frac{\pi}{2}\} .
\end{equation*}
The choice of $\delta$ implies $\mathcal{E}(\nu+\delta)=\mathcal{E}(\nu)$ as well as $\xi(\nu+\delta)=\xi(\nu)$ (see \eqref{lifshitz}).
Then, \eqref{energy01t03} and Lemma \ref{energy00at} yield
\begin{equation}\label{energy01t04}
\begin{split}
  \mathcal{E}(\nu) 
    & = -f(\nu+\delta)\xi(\nu) - \frac{1}{2\pi i} \int_{\Gamma_{\nu,\delta}} f'(z)\ln[\Delta(z)]\, dz\\
    & = -f(\nu+\delta)\xi(\nu)
           - \frac{1}{2\pi i}\Big[(n-m)\int_{\Gamma_{\nu,\delta}} f'(z) \ln(1-\frac{\nu}{z})\, dz 
                                    + \int_{\Gamma_{\nu,\delta}} f'(z) \ln[\Delta_\nu(z)]\, dz\Big] .
\end{split}
\end{equation}
In the last integral we may perform the limit $\delta\to 0$ which simply means to replace $\Gamma_{\nu,\delta}$ by $\Gamma_\nu$.
For the remaining integral only the parts over $\Gamma_{\nu,\delta}^{(1)}$ and $\Gamma_{\nu,\delta}^{(2)}$ need a closer look
\begin{equation*}
  \int_{\Gamma_{\nu,\delta}^{(1)}} f'(z) \ln(1- \frac{\nu}{z})\, dz
     = i \int_{\delta\leq |s|\leq b} f'(\nu+is) \ln(1-\frac{\nu}{\nu+is})\, ds
     = i \int_{\frac{1}{b}\leq |s|\leq \frac{1}{\delta}} f'(\nu+\frac{i}{s}) \frac{1}{s^2}\ln(1-\frac{s\nu}{{s\nu+i}})\, ds .
\end{equation*}
The limit $\delta\to 0$ exists since the logarithm is dominated by any power of $s$. Furthermore,
\begin{equation*}
  \int_{\Gamma_{\nu,\delta}^{(2)}} f'(z) \ln(1- \frac{\nu}{z})\, dz
     = i \delta\int_{-\frac{\pi}{2}}^{\frac{\pi}{2}} f'(\nu+\delta e^{i\vartheta}) 
                         \ln(1- \frac{\nu}{\nu+\delta e^{i\vartheta}})e^{i\vartheta}\, d\vartheta \to 0,\ \delta\to 0 .
\end{equation*}
We conclude that
\begin{equation*}
  \lim_{\delta\to 0} \int_{\Gamma_{\nu,\delta}} f'(z) \ln[\Delta(z)]\, dz
     = \int_{\Gamma_\nu} f'(z) \ln[\Delta(z)]\, dz
\end{equation*}
in \eqref{energy01t04}. This and the continuity of $f$ prove \eqref{energy01t01}.
\end{proof}

\subsection{Abstract boundary conditions\label{abc}}
In order to define a self-adjoint differential operator on an interval one has to ensure
that the boundary terms obtained via integration by parts vanish. That is to say one has to impose
boundary conditions. We reformulate that in an abstract framework. Let 
\begin{equation}\label{abc_H_tilde}
  \tilde H:\dom(\tilde H)\to\hilbert,\   \dom(\tilde H)\subset\hilbert
\end{equation}
be a densely defined linear operator and let $\Gamma_{1,2}:\dom(\tilde H)\to\C^n$, $n\in\N$, be surjective linear maps
such that
\begin{equation*}
  (\varphi,\tilde H\psi) - (\tilde H\varphi,\psi)
    = (\Gamma_1\varphi,\Gamma_2\psi) - (\Gamma_2\varphi,\Gamma_1\psi),\ \varphi,\psi\in\dom(\tilde H) .
\end{equation*}
We study restrictions $H$ of $\tilde H$ given by boundary conditions
\begin{equation}\label{abc_H}
  \dom(H) \coloneqq \{\varphi\in\dom(\tilde H)\mid (A\Gamma_1-B\Gamma_2)\varphi=0 \},\
    H \coloneqq \tilde H|_{\dom(H)}
\end{equation}
with $n\times n$ matrices $A,B:\C^n\to\C^n$. 
For $H$ to be self-adjoint it is necessary that $A$ and $B$ satisfy
\begin{equation}\label{abc_condition}
  AB^* = BA^*,\ \rank(A\mid B)=n
\end{equation}
where $(A\mid B)$ is the $n\times 2n$-matrix formed by the columns of $A$ and $B$.
The matrices $A$ and $B$ are not uniquely determined since we can multiply them on the left by an invertible matrix
without altering \eqref{abc_H} and \eqref{abc_condition}.
It is shown in \cite{BehrndtLanger2010}, that all self-adjoint restrictions of $\tilde H$ are given in this way
if at least one operator $H$ is self-adjoint (for ordinary differential operators \eqref{abc_condition} already
appeared in \cite[9.4]{Ince1926}, and \cite[II.2.2]{Hellwig1964}).
From this one can deduce that
\begin{gather}\label{abc_deficiency_index}
  \dim\ker(z\id-\tilde H) = n\in\N, z\in\C\setminus\R \\
  \label{abc_deficiency_subspace}
  \mathcal{N}_z \coloneqq \ker(z\id-\tilde H) = \mathspan\{ \varepsilon_1(z),\ldots, \varepsilon_n(z)\} .
\end{gather}
Actually, these are the deficiency subspaces of a certain symmetric operator which, however, is not
needed herein and therefore omitted. We note some simple properties.

\begin{lemma}\label{abc01t}
We have $\ker(A^*)\cap\ker(B^*)=\{0\}$. Furthermore, the operators
$\Gamma_{1,2}|_{\mathcal{N}_z}$ and $(A\Gamma_1-B\Gamma_2)|_{\mathcal{N}_z}$ are
injective for all $z\in\C\setminus\R$.
\end{lemma}
\begin{proof}
First note the general equality
\begin{equation*}
  \ker(A^*)\cap\ker(B^*) = (\ran(A))^\perp \cap (\ran(B))^\perp = ( \mathspan(\ran(A)\cup\ran(B)) )^\perp .
\end{equation*}
The rank condition implies that $\mathspan(\ran(A)\cup\ran(B))=\C^n$ and therefore
\begin{equation*}
  \ker(A^*)\cap\ker(B^*) = \{ 0 \} .
\end{equation*}
Let $\varphi\in\mathcal{N}_z$. Then,
\begin{equation*}
 (\Gamma_1\varphi,\Gamma_2\varphi)-(\Gamma_2\varphi,\Gamma_1\varphi)
    = (\varphi,\tilde H\varphi) - (\tilde H\varphi,\varphi)
    = (z-\bar z) \|\varphi\|^2
\end{equation*}
and thus
\begin{equation*}
  \im( (\Gamma_1\varphi,\Gamma_2\varphi) ) = \im(z) \|\varphi\|^2 .
\end{equation*}
Therefore, $\Gamma_{1,2}\varphi\neq 0$ for $\varphi\neq 0$.

By the above we may write
\begin{equation*}
  A\Gamma_1-B\Gamma_2 = (A-B\Gamma) \Gamma_1,\ \Gamma\coloneqq\Gamma_2\Gamma_1^{-1} .
\end{equation*}
It is therefore enough to show that $A-B\Gamma$ is injective which is equivalent to
$A^*-\Gamma^*B^*$ being injective. We have 
\begin{equation*}
  \im( (B^*\varphi,(A^*-\Gamma^* B^*)\varphi) ) = - \im( B^*\varphi, \Gamma^* B^*\varphi) .
\end{equation*}
Note that
\begin{equation*}
  2i \im \Gamma^*
   = \Gamma^*- \Gamma
   = (\Gamma_2\Gamma_1^{-1})^* -\Gamma_2\Gamma_1^{-1}
   = (\Gamma_1^{-1})^* ( \Gamma_2^*\Gamma_1 - \Gamma_1^*\Gamma_2 ) \Gamma_1^{-1}
   = - 2i \im(z) (\Gamma_1^{-1})^* \Gamma_1^{-1}
\end{equation*}
and thus
\begin{equation*}
  \im( (B^*\varphi,(A^*-\Gamma^* B^*)\varphi) ) = \im(z) (B^*\varphi, (\Gamma_1^{-1})^* \Gamma_1^{-1} B^*\varphi) .
\end{equation*}
Thereby, $(A^*-\Gamma^*B^*)\varphi=0$ implies that $B^*\varphi=0$ and this in turn implies that $A^*\varphi=0$.
From the above we conclude $\varphi=0$. That finishes the proof.
\end{proof}

The resolvent $R(z)\coloneqq(z\id-H)^{-1}$ of $H$ (cf. \eqref{abc_H}) is also a right inverse to $z\id-\tilde H$ (cf. \eqref{abc_H_tilde})
\begin{equation}\label{abc_right_inverse}
  (z\id-\tilde H)R(z) = (z\id - H)R(z) = \id .
\end{equation}
Note that generally it will be not a left inverse. We derive a formula that relates different right inverses.
This includes the Albeverio--Pankrashkin formula \cite{AlbeverioPankrashkin2005}, a specialization of
Kre\u\i{}n's resolvent formula, which relates resolvents corresponding to
different boundary conditions.

\begin{proposition}\label{abc02t}
For all $z\in\C\setminus\R$ let $\tilde R(z):\hilbert\to\dom(\tilde H)$ (with $\tilde H$ as in \eqref{abc_H_tilde}) 
be a bounded operator that has the resolvent properties
\begin{equation}\label{abc02t01}
  (z\id-\tilde H)\tilde R(z)=\id\ \text{and}\ \tilde R(z)^* = \tilde R(\bar z) .
\end{equation}
Then, the resolvent of $H$ in \eqref{abc_H} can be written (cf. \eqref{abc_deficiency_subspace})
\begin{equation}\label{abc02t02}
  R(z) = \tilde R(z) - D(z),\
  D(z) = \sum_{j,k=1}^n d_{jk}(z) (\varepsilon_k(\bar z),\cdot)\varepsilon_j(z)
\end{equation}
where the coefficient matrix $\hat D(z)\coloneqq(d_{jk}(z))_{j,k=1,\ldots,n}$ satisfies
\begin{equation}\label{abc02t03}
  (A\Gamma_1-B\Gamma_2)\tilde R(z)|_{\mathcal{N}_z} = (A\Gamma_1-B\Gamma_2)|_{\mathcal{N}_z} \hat D(z) E(z)
\end{equation}
with the generalized Gram matrix
\begin{equation}\label{abc02t04}
  E(z) \coloneqq ( (\varepsilon_j(\bar z), \varepsilon_k(z)) )_{j,k=1,\ldots,n} .
\end{equation}
\end{proposition}
\begin{proof}
We study the properties of the difference
\begin{equation*}
  D(z) = \tilde R(z) - R(z) .
\end{equation*}
From \eqref{abc_right_inverse} and the first equality in \eqref{abc02t01} we obtain
\begin{equation*}
  (z\id-\tilde H) D(z) = (z\id-\tilde H)\tilde R(z) - (z\id-\tilde H)R(z) = \id -\id = 0
\end{equation*}
which implies that
\begin{equation*}
  \ran D(z)\subset\mathcal{N}_z,\ z\in\C\setminus\R ,
\end{equation*}
and therefore
\begin{equation*}
  D(z) = \sum_{k=1}^n (\tilde\varepsilon_k(z),\cdot)\varepsilon_k(z),\ \tilde\varepsilon_k(z)\in\hilbert .
\end{equation*}
Furthermore, the second equality in \eqref{abc02t01} yields
\begin{equation*}
  D(z) = \tilde R(z) - R(z)
       = ( \tilde R(\bar z) - R(\bar z) )^*
       = D(\bar z)^*
       = \sum_{k=1}^n (\varepsilon_k(\bar z),\cdot)\tilde\varepsilon_k(\bar z)
\end{equation*}
which implies that $\tilde\varepsilon_k(\bar z)\in\mathcal{N}_z$ and thereby
\begin{equation*}
  \tilde\varepsilon_k(\bar z) = \sum_{l=1}^n \overline{d_{kl}(\bar z)} \varepsilon_l(z) .
\end{equation*}
The special form of the coefficients has been chosen for convenience.
To determine the coefficient matrix $\hat D(z)$ we use the boundary conditions, i.e. $(A\Gamma_1-B\Gamma_2)R(z)=0$,
\begin{equation*}
  (A\Gamma_1-B\Gamma_2) D(z)
    = (A\Gamma_1-B\Gamma_2)\tilde R(z)
    = \sum_{k,l=1}^n d_{kl}(z)(\varepsilon_l(\bar z),\cdot)(A\Gamma_1-B\Gamma_2)\varepsilon_k(z) .
\end{equation*}
We evaluate this on the vector $\varepsilon_j(z)$ thereby obtaining
\begin{equation*}
  (A\Gamma_1-B\Gamma_2)\tilde R(z)|_{\mathcal{N}_z} = (A\Gamma_1-B\Gamma_2)|_{\mathcal{N}_z} \hat D(z) E(z)
\end{equation*}
where $E(z)$ has the entries $(\varepsilon_l(\bar z), \varepsilon_j(z))$. This finishes the proof.
\end{proof}

We have not shown that the operator $D(z)$ is uniquely determined which would require to solve
for $\hat D(z)$ in \eqref{abc02t03}. Although that sounds reasonable since $(A\Gamma_1-B\Gamma_2)|_{\mathcal{N}_z}$ is injective
we would have to study the matrix $E(z)$. It is easier to do that in the concrete case 
considered in Section \ref{fr}.

\section{Schr\"odinger operators\label{schroedinger}}
From now on we consider concrete Schr\"odinger operators in dimension one.
Let $\hilbert_\infty\coloneqq L^2(\R)$ and let 
\begin{equation}\label{H_infinity}
  H_\infty:\dom(H_\infty)\to\hilbert_\infty,\ H_\infty \coloneqq -\frac{d^2}{dx^2}
\end{equation}
be the free Schr\"odinger operator defined on the whole line. The domain $\dom(H_\infty)$ can be described
with the aid of Sobolev spaces which, however, is not needed herein.
The spectrum is $\sigma(H_\infty)=[0,\infty[$. The resolvent
\begin{equation}\label{R_infinity}
  R_\infty(z)\coloneqq(z\id - H_\infty)^{-1}:\hilbert_\infty\to\hilbert_\infty,\
  z\in\C\setminus\sigma(H_\infty) ,
\end{equation}
has a cut singularity along $\sigma(H_\infty)$.
The corresponding Green function (the kernel of $R_\infty(z)$) reflects this fact. It reads
\begin{equation}\label{green_infinity}
  R_\infty(z;x,y) =
   \mp \frac{i}{2\sqrt{z}} e^{\pm i\sqrt{z}|x-y|},\ \im(\sqrt{z})\gtrless 0,\ x,y\in\R .
\end{equation}
Here, $\sqrt{z}$ is the principal branch of the square root function. This means in particular that
for $\re(z)\geq 0$ the condition $\im(\sqrt{z})\gtrless0$ is equivalent to $\im(z)\gtrless0$.
The Green function can be extended to $\im(\sqrt{z})=0$, $z\neq 0$. We denote the respective boundary values by
\begin{equation}\label{R_boundary_values}
  R_\infty^\pm(\nu;x,y) \coloneqq \lim_{\im(\sqrt{z})\to \pm 0} R_\infty(z;x,y)
                       = \mp \frac{i}{2\sqrt{\nu}} e^{\pm i\sqrt{\nu}|x-y|}
\end{equation}
and write formally
\begin{equation*}
  R_\infty^\pm(\nu) \coloneqq \lim_{\im(\sqrt{z})\to \pm 0} R_\infty(z)
\end{equation*}
for the corresponding unbounded operators. The limit will be studied more carefully below.

We restrict $H_\infty$ to $\hilbert_L\coloneqq L^2(\Lambda_L)$, $\Lambda_L\coloneqq[-L,L]$,
thereby obtaining the maximal operator 
\begin{equation}\label{H_maximal}
  \tilde H:\dom(\tilde H)\to\hilbert_L,\ \tilde H\coloneqq -\frac{d^2}{dx^2},\
  \dom(\tilde H) \coloneqq \{\varphi|_{\Lambda_L} \mid \varphi\in\dom(H_\infty)\} . 
\end{equation}
Integration by parts shows that
\begin{equation*}
  (\varphi,\tilde H\psi) - (\tilde H\varphi,\psi)
  = (\Gamma_1\varphi,\Gamma_2\psi) - (\Gamma_2\varphi,\Gamma_1\psi)
\end{equation*}
where
\begin{equation*}
  \Gamma_1\varphi =
\begin{pmatrix}
  \varphi(L)\\
  \varphi(-L)
\end{pmatrix}
  ,\
  \Gamma_2\varphi =
\begin{pmatrix}
  -\varphi'(L)\\
  \varphi'(-L)
\end{pmatrix} .
\end{equation*}
The deficiency subspace \eqref{abc_deficiency_subspace} is
\begin{equation}\label{deficiency_subspace}
  \mathcal{N}_z = \mathspan\{ \varepsilon_1(z),\varepsilon_2(z)\},\
  \varepsilon_1(z;x) = e^{i\sqrt{z}x},\ \varepsilon_2(z;x) = e^{-i\sqrt{z}x} .
\end{equation}
The boundary conditions are parametrized by the $2\times 2$ matrices $A$ and $B$, see \eqref{abc_condition}.
We define
\begin{equation}\label{H_L}
  H_L\coloneqq\tilde H|_{\dom(H_L)},\
  \dom(H_L) \coloneqq \{ \varphi\in\dom(\tilde H) \mid (A\Gamma_1-B\Gamma_2)\varphi =0 \}
\end{equation}
and note some fundamental properties.

\begin{lemma}\label{H_L_spectrum}
The operator $H_L$ is self-adjoint. Its spectrum $\sigma(H_L)$ consists of eigenvalues
\begin{equation}\label{H_L_spectrum01}
  \sigma(H_L) = \{ \lambda_{j,L} \mid j\in\N\},\ \lambda_{j,L} \leq \lambda_{j+1,L}
\end{equation}
counted with multiplicity, which can be at most two. Furthermore,
\begin{equation}\label{H_L_spectrum02}
  \lambda_{j,L} \geq -c\Big(\frac{1}{L} + \frac{4}{c+1}\Big) + \frac{1}{c+1}\Big(\frac{\pi(j-1)}{2L}\Big)^2,\ j\in\N,
\end{equation}
with some constant $c\geq 0$ depending only on $A$ and $B$.
\end{lemma}
\begin{proof}
The choice $B=0$, i.e. Dirichlet boundary conditions, yields a self-adjoint operator.
We know from Section \ref{abc} that then all self-adjoint realizations are given via matrices $A$ and $B$ that
satisfy \eqref{abc_condition}.

To begin with, we show that $H_L$ is semi-bounded from below and look at the quadratic form
\begin{equation*}
  (\varphi,H_L\varphi)
     = - \int_{-L}^L \bar \varphi(x)\varphi''(x)\, dx
     = (\Gamma_1\varphi,\Gamma_2\varphi) + \int_{-L}^L |\varphi'(x)|^2\, dx,\ \varphi\in\dom(H_L) .
\end{equation*}
We show that
\begin{equation*}
  | (\Gamma_1\varphi,\Gamma_2\varphi) | \leq c ( |\varphi(L)|^2 + |\varphi(-L)|^2 ) ,\ c\geq 0 .
\end{equation*}
For $\rank(B)=0$ this is trivial with $c=0$ and for $\rank(B)=2$ this follows from $\Gamma_2\varphi = B^{-1}A\Gamma_1\varphi$.
The case $\rank(B)=1$ is slightly more difficult. We may choose $A$ and $B$ as
\begin{equation*}
  A =
\begin{pmatrix}
  a_{11} & a_{12} \\
  a_{21} & a_{22}
\end{pmatrix}
  ,\
  B = 
\begin{pmatrix}
  b_{11} & b_{12} \\
  0     & 0 
\end{pmatrix} 
  , \   |a_{21}|^2 + |a_{22}|^2 \neq 0 \neq |b_{11}|^2 + |b_{12}|^2,\   a_{21}\bar b_{11} + a_{22}\bar b_{12} = 0
\end{equation*}
where we used \eqref{abc_condition}. We write out the boundary condition and conclude
\begin{equation*}
\begin{pmatrix}
  \bar\varphi(L) \\ \bar\varphi(-L)
\end{pmatrix}
  = \beta
\begin{pmatrix}
  b_{11} \\ b_{12}
\end{pmatrix}
  ,\ |\beta|^2 = \frac{|\varphi(L)|^2 + |\varphi(-L)|^2}{|b_{11}|^2 + |b_{12}|^2} .
\end{equation*}
Using the boundary condition we obtain
\begin{equation*}
  (\Gamma_1\varphi,\Gamma_2\varphi) 
    = \beta ( - b_{11}\varphi'(L) + b_{12}\varphi'(-L) )
    = \beta ( a_{11}\varphi(L) + a_{12}\varphi(-L) )
\end{equation*}
which proves this case. All in all,
\begin{equation*}
  (\varphi,H_L\varphi) \geq -c ( |\varphi(L)|^2 + |\varphi(-L)|^2 ) + \|\varphi'\|^2 .
\end{equation*}
We estimate the boundary values. To this end,
\begin{equation*}
  |\varphi(L)|^2 = |\varphi(x)|^2 + \int_x^L \frac{d}{dy} |\varphi(y)|^2\, dy
                 = |\varphi(x)|^2 + 2\int_x^L \re(\bar\varphi(y)\varphi'(y))\, dy
                 \leq |\varphi(x)|^2 + 2\|\varphi\| \|\varphi'\| .
\end{equation*}
The other boundary value can be trated likewise. Integrating then yields
\begin{equation*}
  |\varphi(\pm L)|^2 \leq \frac{1}{2L}\|\varphi\|^2 + 2\|\varphi\|\|\varphi'\| 
                     \leq \frac{1}{2L}\|\varphi\|^2 + \frac{1}{\delta}\|\varphi\|^2 + \delta\|\varphi'\|^2
\end{equation*}
with some $\delta>0$. Thus,
\begin{equation*}
  (\varphi,H_L\varphi) 
    \geq -c ( \frac{1}{L}\|\varphi\|^2 + \frac{2}{\delta}\|\varphi\|^2 + 2\delta\|\varphi'\|^2) + \|\varphi'\|^2
     = -c(\frac{1}{L} + \frac{2}{\delta})\|\varphi\|^2 + (1-2c\delta)\|\varphi'\|^2 .
\end{equation*}
A convenient choice is $\delta\coloneqq\frac{1}{2(c+1)}$.
This shows the semi-boundedness. By the variational principle (see e.g. \cite[Thm. XIII.1]{ReedSimon1978})
\begin{equation*}
  \lambda_k \geq -c\Big(\frac{1}{L} + \frac{4}{c+1}\Big) + \frac{1}{c+1}\tilde\lambda_k,\ 
  \tilde\lambda_k=\Big(\frac{\pi(k-1)}{2L}\Big)^2,\ k\in\N,
\end{equation*}
where we used the variational characterization of the $\tilde\lambda_k$, the eigenvalues for Neumann boundary conditions, $A=0$. 
This proves \eqref{H_L_spectrum02} which in turn shows \eqref{H_L_spectrum01}. 
Since $H_L$ is a second order differential operator the eigenvalues' multiplicity can at most be two.
\end{proof}

We will thoroughly study the behaviour of the resolvent
\begin{equation}\label{R_L}
  R_L(z)\coloneqq(z\id- H_L)^{-1}:\hilbert_L\to\hilbert_L,\
  z\in\C\setminus\sigma(H_L) ,
\end{equation}
in the complex plane. 
We investigate the energy difference \eqref{energy_difference} for some fixed $\nu>0$ 
which we call Fermi energy. Because of $\sqrt{z}$ appearing in \eqref{green_infinity},
we choose in \eqref{energy01t01} an integration contour that is composed of parabolas 
(see Fig. \ref{f_fermi_parabola}). 
\begin{figure}[bth]
\includegraphics[width=0.75\textwidth]{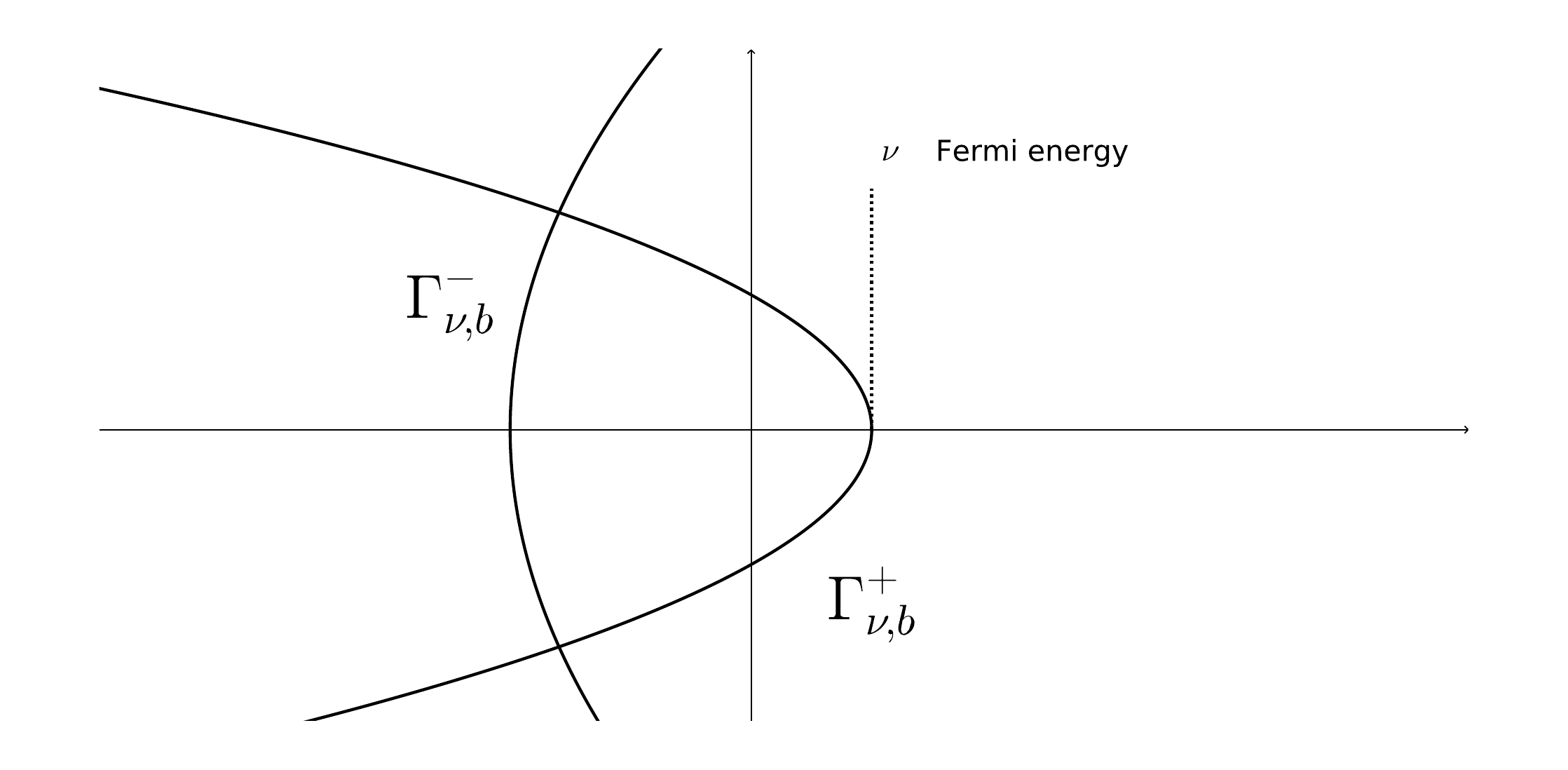}
\caption{The Fermi parabola in the complex plane.}
\label{f_fermi_parabola}
\end{figure}
More precisely, for $b>0$ we define $\Gamma_{\nu,b} \coloneqq \Gamma_{\nu,b}^- \cup \Gamma_{\nu,b}^+$ with
\begin{equation}\label{fermi_parabola_b}
    \Gamma_{\nu,b}^- \coloneqq \{ z = (t+ib)^2\mid \sqrt{\nu}\geq t\geq -\sqrt{\nu}\},\
    \Gamma_{\nu,b}^+ \coloneqq \{ z = (\sqrt{\nu}+is)^2\mid -b\leq s\leq b\} .
\end{equation}
Furthermore, we define the Fermi parabola
\begin{equation}\label{fermi_parabola}
    \Gamma_\nu \coloneqq \Gamma_{\nu,\infty}^+ \coloneqq \{ z = (\sqrt{\nu}+is)^2\mid s\in\R \} .
\end{equation}
Let $V:\hilbert_L\to\hilbert_L$ be the multiplication operator
$(V\varphi)(x)\coloneqq V(x)\varphi(x)$, $\varphi\in\hilbert_L$. For simplicity, we assume $V$ to be
bounded, $V\in L^\infty(\R)$. Then, the perturbed operator
\begin{equation}\label{H_VL}
  H_{V,L}:\dom(H_{V,L})\to\hilbert_L,\ H_{V,L} \coloneqq H_L + V,\
\end{equation} 
is self-adjoint and $\dom(H_{V,L})=\dom(H_L)$. Its spectrum $\sigma(H_{V,L})$ consists of eigenvalues,
\begin{equation}\label{H_VL_spectrum}
  \sigma(H_{V,L}) = \{ \mu_{j,L} \mid j\in\N\},\ \mu_{j,L}\leq \mu_{j+1,L},\ \mu_{j,L}\to\infty\ \text{as} \ j\to\infty,
\end{equation}
counted with multiplicity, which can be at most two. This can be shown with the aid of Lemma \ref{H_L_spectrum}.

At times, the potential needs to fall off at infinity sufficiently fast which is expressed via
\begin{equation}\label{V_decay1}
  X^n V\in L^1(\R),\ (X^nV)(x) \coloneqq x^n V(x),\
  \|V\|_{1,n} \coloneqq \max_{k=0,\ldots,n} \|X^k V\|_1 ,\ n\in\N_0 .
\end{equation}
The limiting absorption principle in Lemma \ref{bso03t} requires a slightly more regular behaviour
expressed with the aid of the Birman--Solomyak condition (see e.g. \cite[pp. 38]{Simon2005})
\begin{equation}\label{birman_solomyak}
  \ell^q(L^1(\R)) \coloneqq \{ V \in L^1(\R)\mid \llbracket V\rrbracket_{1,q} < \infty\},\
  \llbracket V\rrbracket_{1,q} \coloneqq \sum_{j\in\Z} \| V\chi_{I_j} \|_1^q,\
  I_j\coloneqq [j,j+1],\ q>0 .
\end{equation}
Note that compactly supported potentials satisfy the Birman--Solomyak condition, i.e. $L^1_0(\R)\subset\ell^q(L^1(\R))$.
Due to the complex integration contour in the integral representation \eqref{energy01t01} 
of the energy difference we need a weighted $L^1$-norm (cf. \cite[(3.25)]{KuettlerOtteSpitzer2014})
\begin{equation}\label{V_L}
  \mathcal{V}_L(s) \coloneqq \int_{-L}^L |V(x)| e^{s|x|} \, dx,\ s\in\R .
\end{equation}
Note that $\mathcal{V}_L$ is differentiable with respect to $s$.

\subsection{Decomposition of the free resolvent\label{fr}}
We apply the results of Section \ref{abc} to the resolvents of $H_\infty$
and $H_L$, see \eqref{R_infinity} and \eqref{R_L}. First of all, we restrict $R_\infty(z)$ to $\hilbert_L$
which defines both operators on the same Hilbert space.
To this end let $\chi_L$ be the indicator function of $\Lambda_L$. We define
\begin{equation}\label{R_infinity_L}
  R_{\infty,L}(z) \coloneqq \chi_L R_\infty(z)\chi_L,\ R_{\infty,L}^\pm (\nu) \coloneqq \chi_L R_\infty^\pm(\nu)\chi_L ,
\end{equation}
which can be considered both an operator on $L^2(\Lambda_L)$ and on $L^2(\R)$. We use
the same symbol for these operators as the meaning should be clear from the context.
Note that $R_{\infty,L}(z)$ and particularly $R_{\infty,L}^\pm(\nu)$ are bounded operators 
while $R_\infty^\pm(\nu)$ are not since $\nu\in\sigma(H_\infty)$, cf. Lemma \ref{bso01t}.

The operator $R_{\infty,L}(z)$ and the resolvent $R_L(z)$ differ by a rank two operator as we show in Lemma \ref{fr01t}. 
A crucial role is played by the matrix
\begin{equation}\label{fr_bcs}
  U(z) \coloneqq (iA-\sqrt{z}B)^{-1}(iA+\sqrt{z}B)
\end{equation}
with the matrices $A$ and $B$ describing the boundary conditions, see \eqref{H_L} and \eqref{abc_condition}.
In the context of quantum wires, this matrix $U(z)$ is the scattering matrix if $\sqrt{z}$ is real, see \cite[Thm. 2.1]{KostrykinSchrader1999}. 
Since it enters through the boundary conditions we name it \emph{boundary condition scattering matrix}. 
It is studied in more detail in Section \ref{bcs}.

\begin{lemma}\label{fr01t}
For $z\in\C\setminus\sigma(H_L)$ the resolvent $R_L(z)$ of $H_L$ can be decomposed into (cf. \eqref{deficiency_subspace})
\begin{equation}\label{fr01t01}
  R_L(z) = R_{\infty,L}(z) - D_L(z),\ 
  D_L(z) = \frac{1}{2i\sqrt{z}}d_L(z)\sum_{j,k=1}^2 g_{jk}(z) ( \varepsilon_k(\bar z), \cdot )\varepsilon_j(z) .
\end{equation}
Note that $D_L(\bar{z})=(D_L(z))^*$.
The coefficient matrix $G_L(z)\coloneqq (g_{jk}(z))_{j,k=1,2}$ and the scalar prefactor are given through
\begin{equation}\label{fr01t02}
  G_L(z) = ( e^{2i\sqrt{z}L}\id + U(z)\sigma_x )^{-1},\
  d_L(z) \coloneqq e^{2i\sqrt{z}L} ,\
  \im(\sqrt{z})\geq 0,
\end{equation}
with the Pauli matrix $\sigma_x$.
\end{lemma}
\begin{proof}
Since $R_{\infty,L}(z)$ satisfies the resolvent properties \eqref{abc02t01}
(with $\tilde H$ the maximal operator from \eqref{H_maximal}) we may use Proposition \ref{abc02t}.
Computing $\Gamma_{1,2}\varepsilon_{1,2}(z)$ yields the matrix representations
\begin{equation*}
  \Gamma_1|_{\mathcal{N}_z} = e^{i\sqrt{z}L}\id + e^{-i\sqrt{z}L}\sigma_x,\
  \Gamma_2|_{\mathcal{N}_z} = -i\sqrt{z} ( e^{i\sqrt{z}L}\id - e^{-i\sqrt{z}L}\sigma_x )
\end{equation*}
and thus
\begin{equation*}
\begin{split}
  (A\Gamma_1-B\Gamma_2)|_{\mathcal{N}_z}
    & = A(e^{i\sqrt{z}L}\id+e^{-i\sqrt{z}L}\sigma_x) +i\sqrt{z} B(e^{i\sqrt{z}L}\id-e^{-i\sqrt{z}L}\sigma_x) \\
    & = e^{i\sqrt{z}L}(A+i\sqrt{z}B) + e^{-i\sqrt{z}L}(A-i\sqrt{z}B)\sigma_x .
\end{split}
\end{equation*}
Using the formula for the Green function in \eqref{green_infinity} we obtain
\begin{equation*}
  \big[ \Gamma_1R_{\infty,L}(z)\varphi \big]_j = - \frac{i}{2\sqrt{z}}e^{i\sqrt{z}L} (\varepsilon_j(\bar z), \varphi)
  ,\
  \big[ \Gamma_2R_{\infty,L}(z)\varphi \big]_j = -\frac{1}{2} e^{i\sqrt{z}L} (\varepsilon_j(\bar z), \varphi)
  ,\ j=1,2 .
\end{equation*}
Recall, $E(z)_{jk} = (\varepsilon_j(\bar z),\varepsilon_k(z))$. Then, taking $\varphi=\varepsilon_{1,2}(z)$ yields
\begin{equation*}
  \Gamma_1R_{\infty,L}(z)|_{\mathcal{N}_z} = -\frac{i}{2\sqrt{z}}e^{i\sqrt{z}L} E(z),\
  \Gamma_2R_{\infty,L}(z)|_{\mathcal{N}_z} = -\frac{1}{2}e^{i\sqrt{z}L} E(z) .
\end{equation*}
Finally, the equation \eqref{abc02t03} for $\hat D_L(z)$ becomes
\begin{equation*}
  -\frac{1}{2}e^{i\sqrt{z}L} ( \frac{i}{\sqrt{z}}A-B)E(z)
    = ( e^{i\sqrt{z}L}(A+i\sqrt{z}B) + e^{-i\sqrt{z}L}(A-i\sqrt{z}B)\sigma_x ) \hat D_L(z)E(z) .
\end{equation*}
Computing the scalar products $(\varepsilon_j(\bar z),\varepsilon_k(z))$, $j,k=1,2$, we obtain
\begin{equation*}
  E(z) = 2L\id + \frac{1}{\sqrt{z}}\sin(2L\sqrt{z}) \sigma_x
       = 2L ( \id + \frac{1}{w}\sin(w)\sigma_x ),\ w\coloneqq 2L\sqrt{z} .
\end{equation*}
Obviously, $E(z)$ is not invertible if
\begin{equation*}
  0 = \det( \id + \frac{1}{w}\sin(w)\sigma_x ) = 1 - \frac{1}{w^2}\sin^2(w) .
\end{equation*}
Since the right-hand side is an entire holomorphic function of $w$ there are at most countably many solutions 
without any point of accumulation (for more details see \cite{BurnistonSiewert1973a}).
Except for those points we may solve for $\hat D_L(z)$
\begin{equation*}
\begin{split}
  \hat D_L(z) & = -\frac{i}{2\sqrt{z}}( iA-\sqrt{z}B + e^{-2i\sqrt{z}L}(iA+\sqrt{z}B)\sigma_x)^{-1}(iA-\sqrt{z}B)\\
         & = -\frac{i}{2\sqrt{z}} e^{2i\sqrt{z}L} ( e^{2i\sqrt{z}L} \id + U(z)\sigma_x )^{-1}
\end{split}
\end{equation*}
which yields \eqref{fr01t01} and \eqref{fr01t02}. The matrix $G_L(z)$ exists for all $z\in\C\setminus\sigma(H_L)$, 
see Lemma \ref{bcs_special01t}.

Finally, since $R_L(\bar{z}) = (R_L(z))^*$ and $R_{\infty,L}(\bar{z}) = (R_{\infty,L}(z))^*$ we see that $D_L(\bar{z})=(D_L(z))^*$.
\end{proof}

\subsection{Birman--Schwinger operators \texorpdfstring{$K_L$, $K_\infty$, and $K_{\infty,L}$}{}\label{bso}}
We will need some analytic properties of the Birman--Schwinger operators
\begin{gather}
  K_L(z)\coloneqq\sqrt{|V|}R_L(z)\sqrt{|V|}J,\label{K_L}\\
  K_\infty(z)\coloneqq\sqrt{|V|}R_\infty(z)\sqrt{|V|}J,\ K_{\infty,L}(z)\coloneqq\sqrt{|V|}R_{\infty,L}(z)\sqrt{|V|}J\label{K_infinity} ,
\end{gather}
see \eqref{V_polar_decomposition} and \eqref{R_L}, \eqref{R_infinity}, \eqref{R_infinity_L},
as well as their boundary values (cf. \eqref{R_boundary_values})
\begin{equation}\label{K_boundary_values}
  K_{\infty,L}^\pm(\nu) \coloneqq \sqrt{|V|}R_{\infty,L}^\pm(\nu)\sqrt{|V|}J,\
  K_\infty^\pm(\nu) \coloneqq \sqrt{|V|}R_\infty^\pm(\nu)\sqrt{|V|}J .
\end{equation}
We will use the simple facts
\begin{equation}\label{exp_properties}
  |e^{\pm i\sqrt{z}u}| = e^{\mp\im(\sqrt{z})u} \ \text{and}\
  |e^{i\sqrt{z}u} - e^{i\sqrt{w}u}| \leq |u| |\sqrt{z}-\sqrt{w}|\int_0^1 e^{\mp u (t \im(\sqrt{z}) + (1-t)\im(\sqrt{w}))}\, dt ,\ u\in\R .
\end{equation}

\begin{lemma}\label{bso01t}
Let $V\in L^1(\R)$ and $z\in\C\setminus\{0\}$. Then the following hold true.
\begin{enumerate}
\item The Birman--Schwinger operators $K_{\infty,L}(z)$ and $K_\infty(z)$ (see \eqref{K_infinity}) are in $B_2(\hilbert)$ with
\begin{equation}\label{bso01t01}
  \| K_{\infty,L}(z)\|_2 \leq \|K_\infty(z)\|_2 \leq \frac{1}{2|\sqrt{z}|}\|V\|_1 .
\end{equation}
In particular, $\sqrt{|V|}R_{\infty,L}(z)\in B_2(\hilbert)$.
\item
If, in addition, $X^nV\in L^1(\R)$ for $n=1,2$ (cf. \eqref{V_decay1}), then for $w\neq 0$
\begin{equation}\label{bso01t01a}
  \|K_{\infty,L}(z) - K_{\infty,L}(w)\|_2
       \leq  \|V\|_{1,2} \frac{1}{|\sqrt{z}|}\Big( 1+\frac{1}{2|w|}\Big)^{\frac{1}{2}} |\sqrt{z}-\sqrt{w}|,\
  \im(\sqrt{z})\cdot\im(\sqrt{w})\geq 0 .
\end{equation}
\item
Let $\chi_L^\perp\coloneqq 1-\chi_L$. Then the operator $K_{\infty,L}(z):\hilbert_\infty\to\hilbert_\infty$ satisfies
\begin{equation}\label{bso01t02}
  \|K_{\infty,L}(z) - K_\infty(z)\|_2 \leq \frac{1}{\sqrt{|z|}} \|V\|_1^{\frac{1}{2}} \|\chi_L^\perp V\|_1^{\frac{1}{2}} .
\end{equation}
\end{enumerate}
\end{lemma}
\begin{proof}
We prove the estimates, essentially, by bounding the respective kernel functions, cf. \eqref{green_infinity}.
Throughout the proof let $\im(\sqrt{z})\geq 0$ and $\im(\sqrt{w})\geq 0$. The other case can be treated in like manner.

1.
The bound \eqref{bso01t01} follows from
\begin{equation}\label{bso01t03}
  |R_\infty(z;x,y)|^2 = \frac{1}{4|z|} e^{-2\im(\sqrt{z})|x-y|}
                 \leq \frac{1}{4|z|} .
\end{equation}
This also implies that $\sqrt{|V|}R_{\infty,L}(z) = \sqrt{|V|}\chi_LR_\infty(z)\chi_L\in B_2(\hilbert)$, cf. \eqref{R_infinity_L}.

2.
We use \eqref{exp_properties} with $u\coloneqq |x-y|\geq 0$
\begin{multline*}
   |R_\infty(z;x,y) - R_\infty(w;x,y)|
     \leq \frac{1}{2}\frac{1}{|\sqrt{z}|}|e^{i\sqrt{z}u}-e^{i\sqrt{w}u}| + \frac{1}{2}|\frac{1}{\sqrt{z}}-\frac{1}{\sqrt{w}}| |e^{i\sqrt{w}u}| \\
     \leq \frac{1}{2|\sqrt{z}|} |\sqrt{z}-\sqrt{w}| ( u + \frac{1}{|\sqrt{w}|} ) 
     \leq \frac{1}{\sqrt{2}}\frac{1}{|\sqrt{z}|} |\sqrt{z}-\sqrt{w}| ( u^2 + \frac{1}{|w|} )^{\frac{1}{2}} .
\end{multline*}
The term with $u^2$ leads to
\begin{equation*}
  \int_{-L}^L\int_{-L}^L |x-y|^2 |V(x)| |V(y)|\, dy \, dx
     = 2\int_{-L}^L\int_{-L}^L x^2 |V(x)| |V(y)|\, dy \, dx 
         - 2\Big[ \int_{-L}^L  x |V(x)| \, dx\Big]^2 .
\end{equation*}
Dropping the rightmost term yields
\begin{equation*}
  \| K_{\infty,L}(z) - K_{\infty,L}(w)\|_2^2
    \leq  \frac{1}{2}\frac{1}{|z|} |\sqrt{z}-\sqrt{w}|^2 \big[ 2 \|X^2V\|_1 \|V\|_1 + \frac{1}{|w|} \|V\|_1^2 \big] .
\end{equation*}
Finally, we simplify the constant via \eqref{V_decay1} and obtain \eqref{bso01t01a}.

3. Furthermore, we have from \eqref{bso01t03}
\begin{equation*}
\begin{split}
  |R_{\infty,L}(z;x,y) - R_\infty(z;x,y)|
    & \leq |(\chi_L(x)-1)R_\infty(z;x,y)\chi_L(y)| + |R_\infty(z;x,y)(\chi_L(y)-1)|\\
    & \leq \frac{1}{2\sqrt{|z|}} \big[  |\chi_L(x)-1| + |\chi_L(y)-1| \big]
\end{split}
\end{equation*}
since $0\leq\chi_L\leq\id$. That implies \eqref{bso01t02}.
\end{proof}

In order to make the Fredholm determinants such as in \eqref{energy01t01} 
well-defined we establish the relevant trace class properties of the Birman--Schwinger operators 
starting with the finite volume operator. 

\begin{lemma}\label{bso02at}
Let $V\in L^1(\R)$. Then, $\sqrt{|V|}R_L(z)\in B_2(\hilbert_L)$ and
$K_L(z)\in B_1(\hilbert_L)$ for all $z\notin\sigma(H_L)$. 
The map $z\mapsto K_L(z)$ is holomorphic with respect to $B_1(\hilbert_L)$ on $\C\setminus\sigma(H_L)$. 
\end{lemma}
\begin{proof}
The spectral representation of the resolvent yields
\begin{equation}\label{bso02at01}
  \sqrt{|V|}R_L(z)\sqrt{|V|} = 
     \sum_{j=1}^\infty \frac{1}{z-\lambda_{j,L}} (\sqrt{|V|}\varphi_{j,L},\cdot)\sqrt{|V|}\varphi_{j,L}.
\end{equation}
From \eqref{H_L_spectrum02} we infer that
\begin{equation}\label{bso02at02}
  \sum_{j=1}^\infty \frac{1}{|z-\lambda_{j,L}|} < \infty \ \text{for}\ z\notin\sigma(H_L) .
\end{equation}
The eigenvectors $\varphi_{j,L}$ have the general form
\begin{equation*}
  \varphi(x) = C_1 e^{i\sqrt{\lambda}x} + C_2 e^{-i\sqrt{\lambda}x}
\end{equation*}
with appropriate constants $C_{1,2}$. For $\lambda>0$ the normalization implies 
\begin{equation*}
  1  = \int_{-L}^L |\varphi(x)|^2 \, dx
     = 2L ( |C_1|^2 + |C_2|^2 ) + 2\re (C_1\bar C_2) \frac{\sin(2L\sqrt{\lambda})}{\sqrt{\lambda}} 
     \geq \frac{2L\sqrt{\lambda}-1}{\sqrt{\lambda}} ( |C_1|^2 + |C_2|^2 ).
\end{equation*}
For $2L\sqrt{\lambda}>1$ this yields an upper bound for $|C_1|^2+|C_2^2|$ and thereby
\begin{equation*}
  |\varphi_{j,L}(x)|^2 \leq 2 (|C_1|^2+|C_2|^2) \leq \frac{2\sqrt{\lambda_{j,L}}}{2L\sqrt{\lambda_{j,L}}-1}\ 
     \text{for}\ 2L\sqrt{\lambda_{j,L}} > 1.
\end{equation*}
By \eqref{H_L_spectrum02} there are only finitely many $\lambda_{j,L}$ with $2L\sqrt{\lambda_{j,L}} \leq 1$
whence $\|\varphi_{j,L}\|_\infty\leq C$, $j\in\N$, with some constant $0\leq C<\infty$, which may depend on $L$, however.
Using $V\in L^1(\R)$ we conclude
\begin{equation}\label{bso02at03}
  \sup_{j\in\N} \|\sqrt{|V|}\varphi_{j,L}\| < \infty .
\end{equation}
With the aid of \eqref{bso02at02} and \eqref{bso02at03} we deduce from \eqref{bso02at01}
that $\sqrt{|V|}R_L(z)\sqrt{|V|}$ is compact since it is the operator norm limit of finite rank operators. 
Let $\varkappa_j\geq 0$ be its singular values. The singular value decomposition gives
\begin{equation*}
  0 \leq \sum_{j=1}^n \varkappa_j
       = \sum_{j=1}^n (f_j, \sqrt{|V|}R_L(z)\sqrt{|V|} g_j)
       = \sum_{k=1}^\infty \frac{1}{z-\lambda_{k,L}}
           \sum_{j=1}^n (f_j,\sqrt{|V|}\varphi_{k,L})(\sqrt{|V|}\varphi_{k,L},g_j)
\end{equation*}
with orthonormal systems $\{f_j\}$ and $\{g_j\}$. By the Cauchy--Schwarz inequality and Bessel's inequality 
\begin{equation*}
\begin{split}
  \sum_{j=1}^n \varkappa_j
    & \leq \sum_{k=1}^\infty \frac{1}{|z-\lambda_{k,L}|}
           \Big[\sum_{j=1}^n |(f_j,\sqrt{|V|}\varphi_{k,L})|^2\Big]^{\frac{1}{2}}
           \Big[\sum_{j=1}^n |(\sqrt{|V|}\varphi_{k,L},g_j)|^2\Big]^{\frac{1}{2}} \\
    & \leq \sum_{k=1}^\infty \frac{1}{|z-\lambda_{k,L}|} \|\sqrt{|V|}\varphi_{k,L}\|^2 \\
    & < \infty
\end{split}
\end{equation*}
which shows the trace class property. For the holomorphy note that
\begin{equation*}
  \sqrt{|V|}R'(z)\sqrt{|V|} = -\sum_{j=1}^\infty \frac{1}{(z-\lambda_{j,L})^2} (\sqrt{|V|}\varphi_{j,L},\cdot)\sqrt{|V|}\varphi_{j,L}
\end{equation*}
is trace class by the same arguments. Then, standard analysis yields the statement.
That $\sqrt{|V|}R_L(z)$ is a Hilbert--Schmidt operator can be shown in like manner.
\end{proof}

The infinite volume Birman--Schwinger operators are treated in two steps (cf. \cite{Froese1997}).

\begin{lemma}\label{bso02t}
Let $V\in L^1(\R)$ and $\im(\sqrt{z})\neq 0$. Then, $K_\infty(z)$ and $K_{\infty,L}(z)$
are trace class and
\begin{equation}\label{bso02t01}
  \|K_{\infty,L}(z)\|_1 \leq \| K_\infty(z)\|_1 \leq \frac{1}{2|\im(\sqrt{z})|} \|V\|_1 .
\end{equation}
\end{lemma}
\begin{proof}
The first inequality in \eqref{bso02t01} follows via H\"older's inequality \eqref{hoelder_inequality}.
In order to prove the second we factor the resolvent into
\begin{equation*}
   R_\infty(z) = ( \sqrt{z}\id - i\nabla)^{-1}(\sqrt{z}\id + i\nabla)^{-1}
\end{equation*}
and start with the first factor, which has the integral kernel
\begin{equation*}
  (\sqrt{z}\id - i\nabla)^{-1}(x,y) = - e^{-i\sqrt{z}x}e^{i\sqrt{z}y} \Theta(y-x),
   \ \im(\sqrt{z})>0 .
\end{equation*}
We compute the Hilbert--Schmidt norm
\begin{equation*}
\begin{split}
\|\sqrt{|V|}(\sqrt{z}\id-i\nabla)^{-1}\|_2^2
     & = \int_\R\int_\R |V(x)| | e^{-i\sqrt{z}x}e^{i\sqrt{z}y} |^2
          \Theta(y-x)\, dy \, dx  \\
     & = \int_\R |V(x)| e^{2\im(\sqrt{z})x} \int_x^\infty e^{-2\im(\sqrt{z})y}\, dy\, dx\\
     & = \frac{1}{2\im(\sqrt{z})}\|V\|_1 .
\end{split}
\end{equation*}
The second factor can be treated likewise. 
For $\im(\sqrt{z})<0$ we obtain the same result since $K_{\infty,L}(\bar{z}) = J K_{\infty,L}(z)^*J$ and $K_{\infty}(\bar{z}) = J K_{\infty}(z)^*J$.

Thus, $K_\infty(z)=\sqrt{|V|}R_\infty(z)\sqrt{|V|}J$ is the product of two Hilbert--Schmidt operators and thereby trace class.
\end{proof}

We prove a trace class limiting absorption principle which, in particular, extends Lemma \ref{bso02t}
to $\im(\sqrt{z})=0$. Our proof combines the use of indicator functions as in \cite[Prop. 5.6]{Simon2005} with a
local resolvent formula from \cite{Froese1997}. Thereby, the Birman--Schwinger operators can, essentially, be written as an infinite 
sum of rank one operators whose trace norm can be computed explicitly
\begin{equation}\label{trace_norm_rank_one}
  \|(g,\cdot)f\|_1 = \|f\| \|g\|,\ f,g\in\hilbert .
\end{equation}
For an alternative approach based upon Mourre-type estimates see \cite[Thm. 6.1]{Sobolev1993}.

\begin{lemma}\label{bso03t} 
Let $V\in L^1(\R)$ satisfy $V\in\ell^{\frac{1}{2}}(L^1(\R))$ and $X^2V\in\ell^{\frac{1}{2}}(L^1(\R))$, see \eqref{birman_solomyak}.
Let $z,w\in\C\setminus\{0\}$ with $\im(\sqrt{z})\cdot\im(\sqrt{w})\geq 0$.
Then, $K_{\infty,L}(z)-K_{\infty,L}(w)\in B_1(\hilbert)$ and
$K_\infty(z)-K_\infty(w)\in B_1(\hilbert)$. Moreover,
\begin{equation}\label{bso03t01}
\begin{split}
    \| K_{\infty,L}(z)-K_{\infty,L}(w)\|_1
     & \leq \| K_\infty(z)-K_\infty(w)\|_1\\
     & \leq \frac{|\sqrt{z}-\sqrt{w}|}{|\sqrt{z}\sqrt{w}|} \Big\{ 
              |\sqrt{z}+\sqrt{w}|\Big[ \llbracket V\rrbracket_{1,\frac{1}{2}} \llbracket X^2V\rrbracket_{1,\frac{1}{2}} + \frac14 \|V\|_1 \Big]
                    + \llbracket V\rrbracket_{1,\frac{1}{2}}^2 \Big\} .
\end{split}
\end{equation}
In particular, the limits
\begin{equation*}
  \lim\limits_{\im(\sqrt{z})\to\pm 0} K_{\infty,L}(z) = K_{\infty,L}^\pm(\nu)\ \text{and}\
  \lim\limits_{\im(\sqrt{z})\to\pm 0} K_\infty(z) = K_\infty^\pm(\nu),\ z=(\sqrt{\nu}+is)^2,
\end{equation*}
exist in trace class. 
\end{lemma}
\begin{proof}
Throughout the proof we assume $\im(\sqrt{z})\geq 0$ and $\im(\sqrt{w})\geq 0$. 
The case $\im(\sqrt{z})\leq 0$ and $\im(\sqrt{w})\leq 0$ follows from the relation $K_{\infty,L}(\bar{z}) = J K_{\infty,L}(z)^*J$ 
and $K_{\infty}(\bar{z}) = J K_{\infty}(z)^*J$.

The first inequality in \eqref{bso03t01} follows via H\"older's inequality \eqref{hoelder_inequality}.
In order to prove the second we write the difference as an infinite sum of trace class operators.
To this end, put $I_j\coloneqq[j,j+1]$, $j\in\Z$, and write
\begin{equation*}
  K_\infty(z) - K_\infty(w) = \sum_{j,k\in\Z} \chi_{I_j}(K_\infty(z) - K_\infty(w))\chi_{I_k} .
\end{equation*}
We distinguish two cases, namely $j\neq k$ and $j=k$.

(a)
Let $j\neq k$. We start with the case $j<k$.
Let $x\in I_j$ and $y\in I_k$. Hence $x\leq y$ and the Green function \eqref{green_infinity} factorizes
which allows us to write it as the sum of four rank one operators
\begin{equation*}
\begin{split}
  \lefteqn{ R_\infty(z;x,y) - R_\infty(w;x,y) }\\
     & = \frac{1}{8i} \big[\frac{1}{\sqrt{z}} + \frac{1}{\sqrt{w}}\big]
          \big[ f(z,w;c_{jk}-x)g(z,w;y-c_{jk}) + g(z,w;c_{jk}-x)f(z,w;y-c_{jk}) \big]  \\
     & \quad + \frac{1}{4i} \big[\frac{1}{\sqrt{z}}-\frac{1}{\sqrt{w}}\big] 
          \big[h(z;c_{jk}-x)h(z;y-c_{jk})+h(w;c_{jk}-x)h(w;y-c_{jk})\big] .
\end{split}
\end{equation*}
The kernel functions are
\begin{equation*}
    f(z,w;u) \coloneqq e^{i\sqrt{z}u} - e^{i\sqrt{w}u},\
    g(z,w;u) \coloneqq e^{i\sqrt{z}u} + e^{i\sqrt{w}u},\
    h(z;u) \coloneqq e^{i\sqrt{z}u} .
\end{equation*}
The auxiliary quantity is to be chosen such that
\begin{equation*}
  c_{jk} \coloneqq 0,\ \text{for}\ j+1\leq 0\leq k\ \text{and}\
     j+1 \leq c_{jk} \leq k \ \text{otherwise} .
\end{equation*}
We conclude that $\chi_{I_j}(K_\infty(z) - K_\infty(w))\chi_{I_k}$ is of finite rank and, thereby, trace class.
We estimate its trace norm via \eqref{trace_norm_rank_one}. To this end, we use \eqref{exp_properties} and obtain
\begin{gather*}
  \|\chi_{I_j}\sqrt{|V|}f(z,w)\|\cdot \|\chi_{I_k}\sqrt{|V|}g(z,w)\|
     \leq 2|\sqrt{z}-\sqrt{w}| \Big[\int_j^{j+1} |V(x)|(c_{jk}-x)^2\, dx \int_k^{k+1} |V(y)|\, dy\Big]^{\frac{1}{2}} ,\\
  \|\chi_{I_j}\sqrt{|V|}g(z,w)\|\cdot \|\chi_{I_k}\sqrt{|V|}f(z,w)\|
     \leq 2|\sqrt{z}-\sqrt{w}| \Big[\int_j^{j+1} |V(x)|\, dx \int_k^{k+1} |V(y)|(y-c_{jk})^2\, dy\Big]^{\frac{1}{2}} .
\end{gather*}
We get rid of $c_{jk}$ by bounding the products of the integrals by
\begin{equation*}
  \alpha_j \coloneqq \int_j^{j+1} |V(x)|\, dx,\ 
  \beta_k \coloneqq \int_k^{k+1}  y^2 |V(y)|\, dy
\end{equation*}
which yields the estimates
\begin{gather*}
  \int_j^{j+1}|V(x)|\, dx \int_k^{k+1} |V(y)| (y-c_{jk})^2\, dy
     \leq \alpha_j\beta_k + \beta_j\alpha_k, \\
  \int_j^{j+1} |V(x)|(c_{jk}-x)^2\, dx \int_k^{k+1}|V(y)|\, dy
     \leq \beta_j\alpha_k  + \alpha_j\beta_k .
\end{gather*}
Here we used
\begin{align*}
  (y-c_{jk})^2 & \leq y^2     & (x-c_{jk})^2 & \leq x^2+y^2 && \text{for}\ 0<j+1\leq k ,\\
  (y-c_{jk})^2 & = y^2        & (x-c_{jk})^2 & = x^2        && \text{for}\ j+1\leq 0 \leq k ,\\
  (y-c_{jk})^2 & \leq y^2+x^2 & (x-c_{jk})^2 & \leq x^2     && \text{for}\ j+1\leq k<0 .
\end{align*}
Since the estimates are symmetric in $j$ and $k$ we obtain the same result for $j>k$.
Obviously, the norms involving $h$ are bounded by $\alpha_j$. Therefore,
\begin{equation*}
  \sum_{j\neq k} \|\chi_{I_j}(K_\infty(z) - K_\infty(w))\chi_{I_k}\|_1
     \leq \frac{1}{2}|\sqrt{z}-\sqrt{w}| \Big|\frac{1}{\sqrt{z}}+\frac{1}{\sqrt{w}}\Big| 
           \sum_{j\neq k} (\alpha_j\beta_k+\beta_j\alpha_k)^{\frac{1}{2}} 
      + \frac{1}{2}\Big|\frac{1}{\sqrt{z}}-\frac{1}{\sqrt{w}}\Big| \sum_{j\neq k} \alpha_j^{\frac{1}{2}}\alpha_k^{\frac{1}{2}} .
\end{equation*}

(b)
Let $j=k$. We use a local resolvent formula which slightly generalizes the one found in \cite[(7.1),(7.2)]{Froese1997}.
For an arbitrary interval $I\coloneqq [a,b]$, $a<b$, and its indicator function $\chi_I$ we have 
\begin{equation}\label{local_resolvent}
  \chi_I ( R_\infty(z) - R_\infty(w) )\chi_I
     = (w-z) \chi_I R_\infty(z)\chi_I R_\infty(w)\chi_I
       + \frac{1}{4i} \big( \frac{1}{\sqrt{z}}-\frac{1}{\sqrt{w}} \big)                    
               \chi_I ( E_a(z,w) + E_b(z,w) )\chi_I .
\end{equation}
Recall that $w,z\neq 0$, $\im(\sqrt{z})\geq 0$, $\im(\sqrt{w})\geq 0$.
The rank one operator $E_c(z,w)$ has the kernel
\begin{equation*}
  E_c(z,w;x,y) \coloneqq e^{i\sqrt{z}|x-c|}e^{i\sqrt{w}|y-c|},\ c\in\R .
\end{equation*}
Note that \eqref{local_resolvent} becomes the usual resolvent formula in the limit $a\to-\infty$, $b\to\infty$.
The product of the resolvents can be estimated by H\"older's inequality \eqref{hoelder_inequality} yielding the
Hilbert--Schmidt norms
\begin{equation*}
  \|\chi_{I_j}\sqrt{|V|}R_\infty(z)\chi_{I_j}\|_2 \leq \frac{1}{2|\sqrt{z}|} \alpha_j^{\frac{1}{2}},\
  \|\chi_{I_j}R_\infty(w)\sqrt{|V|}\chi_{I_j}\|_2 \leq \frac{1}{2|\sqrt{w}|} \alpha_j^{\frac{1}{2}} .
\end{equation*}
Using \eqref{trace_norm_rank_one}
\begin{equation*}
  \| \chi_{I_j} \sqrt{|V|}E_j(z,w)\sqrt{|V|}\chi_{I_j} \|_1 \leq \alpha_j,\   
  \| \chi_{I_j} \sqrt{|V|}E_{j+1}(z,w)\sqrt{|V|}\chi_{I_j}) \|_1 \leq \alpha_j  .
\end{equation*}
For the kernel estimates we used \eqref{exp_properties}. Thus,
\begin{equation*}
  \| \chi_{I_j}( K_\infty(z) - K_\infty(w) )\chi_{I_j}\|_1 
     \leq |w-z| \frac{1}{4|\sqrt{z}||\sqrt{w}|} \alpha_j + \frac{1}{2} \Big| \frac{1}{\sqrt{z}}-\frac{1}{\sqrt{w}}\Big| \alpha_j .
\end{equation*}

(c)
We combine the results of (a) and (b). 
\begin{equation*}
\begin{split}\|
  K_\infty(z) - K_\infty(w) \|_1 
    & \leq \frac{1}{2} |\sqrt{z}-\sqrt{w}| \Big|\frac{1}{\sqrt{z}} 
            + \frac{1}{\sqrt{w}}\Big| \sum_{j\neq k} (\alpha_j\beta_k+\beta_j\alpha_k)^{\frac{1}{2}} 
            + \frac{1}{2}\Big|\frac{1}{\sqrt{z}} - \frac1{\sqrt{w}}\Big| \sum_{j\not=k} \alpha_j^{\frac{1}{2}} \alpha_k^{\frac{1}{2}} \\
    & \quad +\Big[|z-w| \frac{1}{4|\sqrt{z}\sqrt{w}|} +\frac{1}{2} \Big|\frac{1}{\sqrt{z}} - \frac{1}{\sqrt{w}}\Big|\Big]\sum_j\alpha_j \\
    & \leq \frac{1}{2} |\sqrt{z}-\sqrt{w}| \Big|\frac{1}{\sqrt{z}} 
            + \frac1{\sqrt{w}}\Big| \Big[\sum_{j\neq k} \big(\alpha_j^{\frac{1}{2}}\beta_k^{\frac{1}{2}}+\beta_j^{\frac{1}{2}}\alpha_k^{\frac{1}{2}}\big) 
            +\frac{1}{2} \sum_{j}\alpha_j\Big] \\
    & \quad + \Big|\frac{1}{\sqrt{z}} - \frac{1}{\sqrt{w}}\Big|\Big(\sum_j\alpha_j^{\frac{1}{2}}\Big)^2 \\
    & \leq  |\sqrt{z}-\sqrt{w}| \Big|\frac{1}{\sqrt{z}}+\frac{1}{\sqrt{w}}\Big| 
            \big[\llbracket V\rrbracket_{1,\frac{1}{2}} \llbracket X^2 V\rrbracket_{1,\frac{1}{2}} + \frac{1}{4} \|V\|_1 \big] 
              + \Big|\frac{1}{\sqrt{z}}-\frac{1}{\sqrt{w}}\Big| \llbracket V\rrbracket_{1,\frac{1}{2}}^2.
\end{split}
\end{equation*}
See \eqref{birman_solomyak} for the definition of $\llbracket\cdot\rrbracket$.
\end{proof}

The next lemma concerns the Fumi term.

\begin{lemma}\label{bso04t}
Let $V\in L^1(\R)$ satisfy $V\in\ell^{\frac{1}{2}}(L^1(\R))$ and $X^2V\in\ell^{\frac{1}{2}}(L^1(\R))$, see \eqref{birman_solomyak}.
Then, for all $z\in\C\setminus\{0\}$ the operators $K_{\infty,L}(z)$ and $K_\infty(z)$ are trace class.
Moreover, $K_{\infty,L}(z)\to K_\infty(z)$ in $B_1(\hilbert)$ uniformly on 
compact sets $D\subset\C\setminus\{0\}$ as $L\to\infty$.
\end{lemma}
\begin{proof}
The trace class property is a simple consequence of Lemmas \ref{bso02t} and \ref{bso03t}.
To prove convergence we start with
\begin{equation*}
\begin{split}
  \|K_{\infty,L}(z) - K_\infty(z)\|_1
    & \leq \|\chi_L K_\infty(z)\chi_L - \chi_L K_\infty(z)\|_1
          + \|\chi_L K_\infty(z) - K_\infty(z) \|_1 \\
    & \leq \|K_\infty(z)\chi_L^\perp \|_1 + \| \chi_L^\perp K_\infty(z) \|_1 .
\end{split}
\end{equation*}
Obviously, all operators involved are indeed trace class. We bound the trace norm
\begin{equation*}
  \|K_\infty(z)\chi_L^\perp \|_1 \leq \| (K_\infty(z) - K_\infty(w))\chi_L^\perp \|_1 + \| K_\infty(w)\chi_L^\perp\|_1 .
\end{equation*}
Take $\sqrt{w}=\pm i$ for $\im(\sqrt{z})\gtrless 0$. Using the same factorization of the resolvent 
and the Cauchy--Schwarz inequality for the trace norm as in the proof of Lemma \ref{bso02t} we infer that
\begin{equation*}
  \| K_\infty(w) \chi_L^\perp \|_1 \leq \frac{1}{2} \|V\|_1^{\frac{1}{2}} \|V\chi_L^\perp\|_1^{\frac{1}{2}} .
\end{equation*}
The proof of Lemma \ref{bso03t} can be modified to give the bound
\begin{equation*}
\begin{split}
\lefteqn{\| (K_\infty(z) - K_\infty(w))\chi_L^\perp \|_1}\\
   & \leq \frac{1}{2} |\sqrt{z}-\sqrt{w}| \Big|\frac{1}{\sqrt{z}} 
            + \frac{1}{\sqrt{w}}\Big| \sum_{\substack{\chi_{I_k}\cdot\chi_L^\perp\neq 0\\j\neq k}} (\alpha_j\beta_k+\beta_j\alpha_k)^{\frac{1}{2}} 
            + \frac{1}{2}\Big|\frac{1}{\sqrt{z}} - \frac{1}{\sqrt{w}}\Big| 
                 \sum_{\substack{\chi_{I_k}\cdot\chi_L^\perp\neq 0\\j\neq k}} \alpha_j^{\frac{1}{2}} \alpha_k^{\frac{1}{2}} \\
   & \quad + \Big[|z-w| \frac{1}{4|\sqrt{z}\sqrt{w}|} 
           +\frac{1}{2} \Big|\frac{1}{\sqrt{z}} - \frac{1}{\sqrt{w}}\Big|\Big]\sum_{j,\chi_{I_j}\cdot\chi_L^\perp\neq 0}\alpha_j \\
   & \leq |\sqrt{z}-\sqrt{w}| \Big|\frac{1}{\sqrt{z}}+\frac{1}{\sqrt{w}}\Big| 
            \Big[\frac{1}{2}\llbracket V\rrbracket_{1,\frac{1}{2}} \llbracket X^2 V\chi_L^\perp\rrbracket_{1,\frac{1}{2}} 
               + \frac{1}{2}\llbracket X^2 V\rrbracket_{1,\frac{1}{2}}\llbracket V\chi_L^\perp\rrbracket_{1,\frac{1}{2}}  
               + \frac{1}{4} \|V\chi_L^\perp\|_1 \Big]  \\
   & \quad + \Big|\frac{1}{\sqrt{z}}-\frac{1}{\sqrt{w}}\Big| \llbracket V\rrbracket_{1,\frac{1}{2}} \llbracket V\chi_L^\perp\rrbracket_{1,\frac{1}{2}}.
\end{split}
\end{equation*}
The norm $\| \chi_L^\perp K_\infty(z) \|_1$ has the same bound. This and the assumptions on $V$ yield the claimed convergence.
\end{proof}

\subsection{Wave operators\label{wo}}
We define the wave operators (cf. \eqref{Omega_K}) corresponding to \eqref{K_infinity}
\begin{equation}\label{Omega}
  \Omega_{\infty,L}(z)  \coloneqq (\id -K_{\infty,L}(z))^{-1} : \hilbert_L\to\hilbert_L,\
  \Omega_\infty(z)     \coloneqq (\id -K_\infty(z))^{-1} : \hilbert_\infty\to\hilbert_\infty
\end{equation}
and their boundary values (see \eqref{K_boundary_values})
\begin{equation}\label{Omega_boundary_values}
  \Omega_{\infty,L}^\pm(\nu)  \coloneqq (\id -K_{\infty,L}^\pm(\nu))^{-1},\
  \Omega_\infty^\pm(\nu) \coloneqq (\id -K_\infty^\pm(\nu))^{-1} .
\end{equation}
First of all, we have to ensure that the wave operators exist.

\begin{lemma}\label{wo01t}
Let $V\in L^1(\R)$. Then the following hold true.
\begin{enumerate}
\item 
The wave operator $\Omega_\infty(z)$ exists for all $z\in\C\setminus\interval[open right]{-\frac{1}{4}\|V\|_1^2}{\infty}$.
\item
If, in addition, $V\in\ell^{\frac{1}{2}}(L^1(\R))$ and $X^2V\in\ell^{\frac{1}{2}}(L^1(\R))$ then
$\Omega_\infty(z)$ exists for all $z\in\C\setminus[-\frac{1}{4}\|V\|_1^2, 0]$ and is bounded.
Equivalently, $\det(\id-K_\infty(z))\neq 0$. Moreover, for every compact $D\subset\C\setminus[-\frac{1}{4}\|V\|_1^2, 0]$
\begin{equation}\label{wo01t01}
  \inf_{z\in D} |\det(\id-K_\infty(z))| > 0 .
\end{equation}
\end{enumerate}
\end{lemma}
\begin{proof}
1.
First note that the Birman--Schwinger operator $K_\infty(z)$ is bounded for $z\neq 0$, Lemma \ref{bso01t}.
We distinguish two cases.

(i)
Let $\im(z)\neq 0$. Then $z\notin\sigma(H)\cup\sigma(H_V)$. 
By Lemma \ref{energy_omega} the inverse $\Omega_\infty(z)$ exists and is bounded.

(ii)
Let $\im(z)=0$ and $\re(z)<-\frac{1}{4}\|V\|_1^2$. From Jensen's inequality \eqref{jensen_inequality}
and \eqref{bso01t01}
\begin{equation*}
   \|K_\infty(z)\| \leq \|K_\infty(z)\|_2 \leq \frac{1}{2|\sqrt{z}|}\|V\|_1 < 1 .
\end{equation*}
A Neumann series argument then shows that $\Omega_\infty(z)$ exists and is bounded.

Both in case (i) and (ii) we have $\im(\sqrt{z})\neq 0$ whence $K_\infty(z)\in B_1(\hilbert)$ 
by Lemma \ref{bso02t}. Therefore, the perturbation determinant is well-defined and non-zero
(cf. \cite[Thm. XIII.105]{ReedSimon1978}). 

2.
Assume the Birman--Solomyak condition on $V$ whereby $K_\infty(z)\in B_1(\hilbert)$ for $z\in\interval[open]{0}{\infty}$,
Lemma \ref{bso03t}, and the perturbation determinant is well-defined.
We infer from the Jost--Pais formula (see \eqref{jost_pais01} and \eqref{jost_pais02}) that $|\det(\id-K_\infty(z))| \geq 1$.
Thus, the determinant cannot vanish and the inverse $\Omega(z)$ exists and is bounded 
(cf. \cite[Thm. XIII.105, Thm. XIII.107]{ReedSimon1978}).
By Lemma \ref{bso03t}, the map $K_\infty(z)$ depends continuously on $z$ with respect to $B_1(\hilbert)$ and so does
the determinant $\det(\id-K_\infty(z))$. This implies \eqref{wo01t01}
\end{proof}

The wave operator $\Omega_{\infty,L}(z)$ requires weaker conditions.

\begin{lemma}\label{wo01at}
Let $V\in L^1(\R)$. Then $\Omega_{\infty,L}(z)$ exists for all $z\in\C\setminus[-\frac{1}{4}\|V\|_1^2, 0]$ and is bounded.
\end{lemma}
\begin{proof}
$K_{\infty,L}(z)$ is the Birman--Schwinger operator for the potential $V\chi_L$, which has compact
support and, thus, satisfies the Birman--Solomyak conditions in Lemma \ref{wo01t}.
\end{proof}

In studying the energy difference with the aid of Lemma \ref{wo03t} we need the wave operators restricted, essentially,
to the deficiency subspace (see \eqref{deficiency_subspace}). 
By \ref{Omega_T} this leads to the T-matrix
\begin{equation}\label{Omega_matrix}
  2i\sqrt{z}T_L(z) = ( (\sqrt{|V|}\varepsilon_j(\bar z),J\Omega_{\infty,L}(z)\sqrt{|V|}\varepsilon_k(z)) )_{j,k=1,2},\
   \im(\sqrt{z})\geq 0 .
\end{equation}
For the infinite volume operators this is, in general, true only for $z\in\R$,
\begin{equation*}
  2i\sqrt{\nu}T(\nu) = ( (\sqrt{|V|}\varepsilon_j(\nu),J\Omega_\infty^+(\nu)\sqrt{|V|}\varepsilon_k(\nu)) )_{j,k=1,2}.
\end{equation*}
For the T-matrices and S-matrices see \eqref{Omega_T} and \eqref{s-matrix02}. Note that $z=k^2$ and $S(z)=\mathscr{S}(k)$,
this convention being used for all scattering data.

\begin{lemma}\label{wo02t}
Let $V\in L^1(\R)$. Then, the following hold true.
\begin{enumerate}
\item 
Define $V_-$ as $V_-(x) \coloneqq \min\{ V(x) ,0\}$. Let $z=(\sqrt{\nu}+is)^2$ with $\nu>0$ and $s\geq0$. Then,
\begin{equation} \label{wo02t01}
  \| S_L(z) \|_2 \leq  \sqrt{2}\Big\{ 1+ \frac{1}{2|\sqrt{z}|}\mathcal{V}_L(2s) \exp\Big[ \frac{1}{|\sqrt{z}|}\|V\|_1 \Big] \Big\}
                        \exp\Big[ \frac{1}{2\nu} ( \| V_-\|_1^2 + \sqrt{\nu} \| V_-\|_1 )\Big] .
\end{equation}
Let $z=(t+ib)^2$ with $b\geq \|V_-\|_1$. Then,
\begin{equation}\label{wo02t01a}
  \|T_L(z) \|_2  \leq \frac{\sqrt{2}}{2|\sqrt{z}|}\big[ \|V\|_1 + \mathcal{V}_L(2b) \big]\exp\Big[ \frac{1}{|\sqrt{z}|}\|V\|_1 \Big]
                  \exp\big[ \frac{2}{b}\|V_-\|_1 + \frac{2}{b^2}\|V_-\|_1^2\big] .
\end{equation}
\item
The map $\Gamma_\nu^+\ni z \mapsto S_L(z)$ is continuous at the Fermi energy
\begin{equation}\label{wo02t02}
  \| S_L(z) - S_L(\nu) \|_2 = \| T_L(z) - T_L(\nu) \|_2 \leq C s
\end{equation}
\item
The scattering matrix $S_L(\nu)$ converges to $S(\nu)$ as $L\to\infty$. More precisely,
\begin{equation}\label{wo02t03}
  \| S_L(\nu) - S(\nu) \|_2 = \| T_L(\nu) - T(\nu) \|_2
     \leq \Big(\frac{5}{2}\Big)^{\frac{1}{2}}\frac{1}{\sqrt{\nu}} \int_{|x|\geq L} |V(x)|\, dx \exp\big[\frac{1}{\sqrt{\nu}}\|V\|_1\big].
\end{equation}
\end{enumerate}
\end{lemma}
\begin{proof}
It is convenient to introduce
\begin{equation*}
  s_{0,L}(z) \coloneqq (e_1(\bar z),\chi_L Jf_1(z)),\ 
  s_{1,L}(z) \coloneqq (e_1(\bar z),\chi_L Jf_2(z)),\ s_{2,L}(z) \coloneqq (e_2(\bar z),\chi_LJf_1(z))
\end{equation*}
and the matrix
\begin{equation*}
  \tilde S_L(z) \coloneqq
\begin{pmatrix}
  s_{0,L}(z) & s_{1,L}(z) \\
  s_{2,L}(z) & s_{0,L}(z)
\end{pmatrix} .
\end{equation*}
The relations
\begin{equation}\label{wo02t04}
  r_{j,L}(z) = \frac{1}{2i\sqrt{z}} s_{j,L}(z) t_L(z),\ j=1,2,\
  T_L(z) =  \frac{t_L(z)}{2i\sqrt{z}}\tilde S_L(z)
\end{equation}
follow easily from \eqref{s-matrix01}.

1.
From the Faddeev--Deift--Trubowitz formula \eqref{deift_trubowitz01} we infer that
\begin{equation}\label{deift_trubowitz02}
  |t_L(z)| \leq \prod_{j=1}^n \frac{|u+i(v+\beta_j)|}{|u+i(v-\beta_j)|},\ v\geq 0 .
\end{equation}
We find uniform bounds. Firstly, we maximize with respect to $v$ thereby obtaining
\begin{equation*}
  \frac{|u+i(v+\beta_j)|^2}{|u+i(v-\beta_j)|^2}
    \leq 1 + \frac{2}{u^2} ( \beta_j^2 + \beta_j\sqrt{u^2+\beta_j^2} ) 
    \leq 1 + \frac{2}{u^2} ( 2 \beta_j^2 + \beta_j u ).
\end{equation*}
Secondly, the maximizing with respect to $u$ yields for $v\geq 2\beta_j$, $j=1,\ldots,n$,
\begin{equation*}
  \frac{|u+i(v+\beta_j)|^2}{|u+i(v-\beta_j)|^2}
     = 1 + \frac{4v\beta_j}{u^2+(v-\beta_j)^2}
     \leq 1 + \frac{4\beta_j}{v-\beta_j} + \frac{4\beta_j^2}{(v-\beta_j)^2}
     \leq 1 + \frac{8\beta_j}{v} + \frac{16\beta_j^2}{v^2} .
\end{equation*}
Therefore,
\begin{equation*}
  |t_L(z)| \leq \exp\Big[ \frac{1}{u^2} \sum_{j=1}^n ( 2 \beta_j^2 + \beta_j u ) \Big] ,\
  |t_L(z)| \leq \exp\Big[ \frac{4}{v}\sum_{j=1}^n \beta_j + \frac{8}{v^2}\sum_{j=1}^n \beta_j^2 \Big] .
\end{equation*}
The sums of the $\beta_j$ can be estimated with the aid of a Lieb--Thirring inequality 
(see \cite{Weidl1996}, \cite{HundertmarkLiebThomas1998})
\begin{equation*}
  \big[\sum_{j=1}^n \beta_j^2\big]^{\frac{1}{2}} 
    \leq \sum_{j=1}^n \beta_j 
    \leq \frac{1}{2} \int_\R \chi_L(x) |V_-(x)| \, dx \leq \frac{1}{2} \| V_- \|_1 .
\end{equation*}
The first sum could also be treated directly via an appropriate Lieb--Thirring inequality. However, that
would require an additional condition on $V$. Now,
\begin{equation*}
  |t_L(z)| \leq \exp\big[ \frac{1}{2u^2} ( \| V_-\|_1^2 + u \| V_-\|_1 )\big] ,\
  |t_L(z)| \leq \exp\big[ \frac{2}{v}\|V_-\|_1 + \frac{2}{v^2}\|V_-|_1^2\big] .
\end{equation*}
Using \eqref{stl01t02} we obtain
\begin{gather*}
  |s_{0,L}(z)| \leq \int_\R  |V(x)|\chi_L(x)\, dx \exp\big[ \frac{1}{|\sqrt{z}|}\|V\|_1 \big]
              \leq \|V\|_1 \exp\big[ \frac{1}{|\sqrt{z}|}\|V\|_1 \big], \\
  |s_{j,L}(z)|
    \leq \int_\R e^{-2x\im(\sqrt{z})} |V(x)|\chi_L(x)\, dx \exp\big[ \frac{1}{|\sqrt{z}|}\|V\|_1 \big]
    \leq \mathcal{V}_L(2\im(\sqrt{z})) \exp\big[ \frac{1}{|\sqrt{z}|}\|V\|_1 \big] ,\ j=1,2,
\end{gather*}\
cf. \eqref{V_L}. Finally, put $u=\sqrt{\nu}$ in the first estimate of $t_L(z)$. Then, the first relation in \eqref{wo02t04} 
yields \eqref{wo02t01}. For \eqref{wo02t01a} put $v=b$ in the second estimate of $t_L(z)$ 
and use the second relation in \eqref{wo02t04}.
Note that $\beta_j\leq\frac{1}{2}\|V_-\|_1$ which is a trivial consequence of the Lieb--Thirring inequality.

2.
Using \eqref{stl03t01} we obtain via the mean value theorem
\begin{equation*}
  |m_j(\sqrt{\nu}+is;x) - m_j(\sqrt{\nu};x)| 
    \leq s \frac{2}{\sqrt{\nu}} \exp\big[ \frac{2}{\sqrt{\nu}} \|V\|_1\big] \times
\begin{cases}
   \int_x^\infty (y-x)|V(y)|\, dy & \text{for}\ j=1 , \\
   \int_{-\infty}^x (x-y)|V(y)|\, dy & \text{for}\ j=2 ,
\end{cases}
\end{equation*}
and thus
\begin{align*}
  |s_{0,L}(z) - s_{0,L}(\nu)|
     & \leq s \|XV\|_1\|V\|_1 \frac{2}{\sqrt{\nu}}\exp\big[ \frac{2}{\sqrt{\nu}}\|V\|_1\big] , \\
  |s_{j,L}(z) - s_{j,L}(\nu)| 
     & \leq s \Big\{ \mathcal{V}_L'(2s)\|V\|_1 + \mathcal{V}_L(2s)\|XV\|_1\Big\} \frac{2}{\sqrt{\nu}}\exp\big[\frac{2}{\sqrt{\nu}}\|V\|_1\big] ,\
        j=1,2.
\end{align*}
One can check straightforwardly that
\begin{equation*}
  S_L(z) - S_L(\nu)
    = \frac{t_L(z)}{2i\sqrt{z}}
        \Big\{ \tilde S_L(z) - \tilde S_L(\nu) + \big[ s_{0,L}(z)-s_{0,L}(\nu) - 2i(\sqrt{z}-\sqrt{\nu})\big] T_L(\nu) \Big\}
\end{equation*}
which implies \eqref{wo02t02}.

3.
For $S(\nu)$ define $\tilde S(\nu)$ analogously to $\tilde S_L(\nu)$.
Using  \eqref{stl01t02} we obtain easily
\begin{equation*}
   |s_0(\nu) - s_{0,L}(\nu)|
     = \Big|\frac{1}{2i\sqrt{\nu}} \int_{|x|\geq L} V(x) e^{-i\sqrt{\nu}x}f_1(\nu;x)\, dx\Big|
     \leq \frac{1}{2\sqrt{\nu}} \int_{|x|\geq L} |V(x)| \, dx \exp\big[ \frac{1}{\sqrt{\nu}} \|V\|_1\big]
\end{equation*}
and likewise for $j=1,2$
\begin{equation*}
  |s_j(\nu) - s_{j,L}(\nu)|
     \leq \frac{1}{2\sqrt{\nu}} \int_{|x|\geq L} |V(x)| \, dx \exp\big[ \frac{1}{\sqrt{\nu}} \|V\|_1\big] .
\end{equation*}
The simple relation
\begin{equation*}
  S(\nu) - S_L(\nu)
    = \frac{t(\nu)}{2i\sqrt{\nu}}[ \tilde S(\nu) - \tilde S_L(\nu) + (s_0(\nu)-s_{0,L}(\nu))(S_L(\nu)-\id) ]
\end{equation*}
then yields \eqref{wo02t03}.
\end{proof}

For a finite rank operator $D$ the Fredholm determinant of $\id-D$ reduces to a usual determinant.

\begin{lemma}\label{determinant_finite_rank}
Let $D:\hilbert\to\hilbert$ be an operator of finite rank $n\in\N$ and let $\hat D$ be the corresponding 
$n\times n$ Gram matrix, i.e.
\begin{equation}\label{determinant_finite_rank01}
   D = \sum_{j=1}^n (g_j,\cdot)f_j,\ \hat D\coloneqq((g_j,f_k))_{j,k=1,\ldots,n}
\end{equation}
where $f_j,g_k\in\hilbert$. Then $\sigma(D)\setminus\{0\}=\sigma(\hat D)\setminus\{0\}$.
In particular,
\begin{equation}\label{determinant_finite_rank02}
   \det (\id-D) = \det (\id - \hat D) .
\end{equation}
\end{lemma}
\begin{proof}
Let $\lambda\neq 0$ be an eigenvalue of $D$, i.e. $D\varphi=\lambda\varphi$ for some $\varphi\neq 0$.
Explicitly,
\begin{equation*}
   \lambda\varphi = \sum_{k=1}^n (g_k,\varphi) f_k .
\end{equation*}
We conclude $c\coloneqq( (g_1,\varphi),\ldots,(g_n,\varphi))^T\neq 0$ since otherwise we had $\lambda=0$.
Now, one easily checks $\hat D c=\lambda c$. Conversely, writing out the eigenvalue equation $\lambda c = \hat D c$,
\begin{equation*}
  \lambda c_j = \sum_{k=1}^n (g_j,f_k)c_k = (g_j, \sum_{k=1}^n c_k f_k) = (g_j,\varphi),\ \varphi\coloneqq\sum_{k=1}^n c_kf_k,
\end{equation*}
shows that $\varphi\neq 0$. Now, by simple algebra $D\varphi=\lambda\varphi$.

In the determinants $\det(\id-D)$ and $\det(\id-\hat D)$ a possible eigenvalue $0$ of $D$ or $\hat D$
does not matter which proves \eqref{determinant_finite_rank02}.
\end{proof}

By now we have introduced all objects so that we can finally present the
factorization of the perturbation determinant (cf. \eqref{perturbation_determinant}).

\begin{lemma}\label{wo03t}
The following hold true.
\begin{enumerate}
\item
We have the factorization
\begin{equation}\label{wo03t00}
  \id - K_L(z) = (\id - K_{\infty,L}(z))(\id + \Omega_{\infty,L}(z)\sqrt{|V|}D_L(z)\sqrt{|V|}J) .
\end{equation}
For $\im(z)\neq 0$, the operator $\Omega_{\infty,L}(z)\sqrt{|V|}D_L(z)\sqrt{|V|}J$ has no spectral values
in $\interval[open left]{-\infty}{-1}$.
\item
Let $T_L(z)$ be the T-matrix for the potential $V\chi_L$. Then,
\begin{equation}\label{wo03t01}
  \det(\id-K_L(z))
      = \det(\id-K_{\infty,L}(z)) \det(\id+d_L(z) T_L(z)G_L(z)) .
\end{equation}
Note that $\det(\id-K_L(\bar{z})) = (\det(\id-K_L(z)))^*$, and similarly for the two other determinants on the right-hand side.

For $\im(z)\neq 0$ the $2\times 2$-matrix $d_L(z)T_L(z)G_L(z)$ does not have eigenvalues in $\interval[open left]{-\infty}{-1}$.
\item
With the S-matrix, $S_L(z)=\id+T_L(z)$, we have
\begin{equation}\label{wo03t02}
  \id+d_L(z) T_L(z)G_L(z) = (e^{2iL\sqrt{z}}S_L(z) + U(z)\sigma_x)(e^{2iL\sqrt{z}}\id + U(z)\sigma_x)^{-1} .
\end{equation}
\end{enumerate}
\end{lemma}
\begin{proof}
1.
The decomposition \eqref{fr01t01} in Lemma \ref{fr01t} implies
\begin{equation*}
  \id - K_L(z) = \id - K_{\infty,L}(z) + \sqrt{|V|}D_L(z)\sqrt{|V|}J
               = (\id - K_{\infty,L}(z)) (\id + \Omega_{\infty,L}(z)\sqrt{|V|}D_L(z)\sqrt{|V|}J) .
\end{equation*}
Since $\Omega_{\infty,L}(z)\sqrt{|V|}D_L(z)\sqrt{|V|}J$ is a rank two operator it has only eigenvalues. 
One eigenvalue may be $\varkappa=0\notin\interval[open left]{-\infty}{-1}$. If $\im(z)\neq 0$ and $\im(\varkappa)\neq 0$
the statement is true as well. Hence, the critical case is $\im(z)\neq 0$ and $\im(\varkappa)=0$.
Let $\varphi\neq 0$ such that
\begin{equation*}
  \Omega_{\infty,L}(z)\sqrt{|V|}D_L(z)\sqrt{|V|}J\varphi = \varkappa\varphi .
\end{equation*}
Via the decomposition \eqref{fr01t01} this implies
\begin{equation*}
  ( \varkappa\id + \sqrt{|V|}R_L(z)\sqrt{|V|} )J\varphi = (\varkappa+1)\sqrt{|V|}R_{\infty,L}(z)\sqrt{|V|}J\varphi .
\end{equation*}
We multiply by $J$, take scalar products
\begin{equation*}
  \varkappa (\varphi,J\varphi) + (\varphi,J\sqrt{|V|}R_L(z)\sqrt{|V|}J\varphi)
      = (\varkappa+1)(\varphi,J\sqrt{|V|}R_{\infty,L}(z)\sqrt{|V|}J\varphi) ,
\end{equation*}
and look at the imaginary part
\begin{multline*}
  \im(\varkappa) (\varphi,J\varphi) - \im(z) \|R_L(z)\sqrt{|V|}J\varphi\|^2\\
     = \im(\varkappa)\re((\varphi,J\sqrt{|V|}R_{\infty,L}(z)\sqrt{|V|}J\varphi)) 
      - (\re(\varkappa)+1)\im(z) \|R_{\infty,L}(z)\sqrt{|V|}J\varphi\|^2 .
\end{multline*}
By the assumption on $z$ and $\varkappa$ this implies
\begin{equation*}
  \| R_L(z)\sqrt{|V|}J\varphi\|^2 = (\re(\varkappa)+1) \|R_{\infty,L}(z)\sqrt{|V|}J\varphi\|^2 .
\end{equation*}
The norms do not vanish since that would imply $\sqrt{|V|}J\varphi=0$ and furthermore $\varkappa\varphi=0$ 
by the eigenvalue equation which contradicts $\varkappa\neq 0$ and $\varphi\neq 0$. We conclude
that $\re(\varkappa)+1>0$ which proves the statement on the spectral values.

2.
We take determinants and apply Lemma \ref{determinant_finite_rank} to the finite rank operator 
$\Omega_{\infty,L}(z)\sqrt{|V|}D_L(z)\sqrt{|V|}J$ thereby obtaining
\begin{equation*}
\begin{split}
  \det(\id- K_L(z))
     & = \det(\id-K_{\infty,L}(z)) \det(\id+\Omega_{\infty,L}(z)\sqrt{|V|}D_L(z)\sqrt{|V|}J)\\
     & = \det(\id-K_{\infty,L}(z)) \det(\id+\tilde\Omega_L(z)) .
\end{split}
\end{equation*}
The matrix $\tilde\Omega_L(z)$ has the entries, $j,k=1,2$,
\begin{equation*}
  (\tilde\Omega_L(z))_{jk}
     = \frac{1}{2i\sqrt{z}}d_L(z)\sum_{l=1}^2 (\sqrt{|V|}J\varepsilon_j(\bar z),\Omega_{\infty,L}(z)\sqrt{|V|}\varepsilon_l(z))g_{lk}(z)
     = \frac{1}{2i\sqrt{z}}d_L(z)(\hat\Omega_{\infty,L}(z)G_L(z))_{jk} .
\end{equation*}
This proves \eqref{wo03t01}. Apart from $0$ the matrix $\tilde\Omega_L(z)$ has the same eigenvalues
as the rank two operator in part one, see \ref{determinant_finite_rank}. 

3.
Finally, \eqref{wo03t02} is straightforward to prove.
\end{proof}

\section{Boundary condition scattering matrix \texorpdfstring{$U(z)$}{U}\label{bcs}}
We study in more detail the properties of the boundary condition scattering matrix $U(z)$, see \eqref{fr_bcs}.
At first, we derive some general properties, which might be interesting in its own right.
Then we investigate the behaviour on the Fermi parabola.

\subsection{General properties of \texorpdfstring{$U(z)$}{U}}
We consider a somewhat more general situation by allowing for $n\times n$-matrices
instead of the $2\times 2$-matrix in \eqref{fr_bcs}. Also, $\sqrt{z}$ may be replaced by
$w\in\C$. We start with showing that the inverse in \eqref{fr_bcs} exists.

\begin{lemma}\label{bcs_general01t}
The matrix $iA+w B$ is invertible for all $w\in\C$ with $\re(w)\neq 0$ and for all $w\in\C$, $\re(w)=0$,
except a finite number of points $w_1,\ldots,w_m$, $m\leq n$.
\end{lemma}
\begin{proof}
The matrix $iA+w B$ is invertible if and only if its adjoint is invertible.
Let $(iA+wB)^*\varphi=0$. We compute
\begin{equation*}
\begin{split}
  (iA+w B)(iA+w B)^*
    & = AA^* + |w|^2BB^* +i\re(w)(AB^* - BA^*) + \im(w) (AB^*+BA^*) \\
    & = (A+\im(w)B)(A^*+\im(w)B^*) + \re(w)^2 BB^* .
\end{split}
\end{equation*}
Since $\re(w)\neq 0$ we conclude
\begin{equation*}
  (A^* + \im(w) B^*)\varphi = 0,\ B^*\varphi = 0 .
\end{equation*}
Lemma \ref{abc01t} implies that $\varphi=0$. In particular, $\det(iA+wB)\neq 0$ for $\re(w)\neq 0$.
Obviously, the determinant is a polynomial in $w$ of degree at most $n$. Since it does not vanish identically 
there can be at most $n$ zeros.
\end{proof}

The conditions on $A$ and $B$ imply some nice properties of $U$.

\begin{lemma}\label{bcs_general02t}
The matrix $U(w)\coloneqq (iA-wB)^{-1}(iA+wB)$ has an eigenvalue $-1$ if $A$ has an eigenvalue $0$. Likewise,
$U(w)$ has an eigenvalue $1$ if $B$ has an eigenvalue $0$. The respective eigenvectors can be chosen independently of $w$.
Furthermore, $U(w_1)U(w_2)^*=U(w_2)^*U(w_1)$ for all $w_1,w_2\in\C\setminus\{0\}$ and $w_1+\bar w_2\neq 0$. 
In particular, $U(w)$ is normal for $\re(w)\neq 0$. If, in addition, $\im(w)=0$ then $U(w)$ is unitary.
\end{lemma}
\begin{proof}
By Lemma \ref{bcs_general01t} the matrix $iA-wB$ is invertible. We may therefore write
\begin{equation*}
  U(w) = -\id + 2i(iA-wB)^{-1}A = \id + 2w(iA-wB)^{-1}B .
\end{equation*}
This shows the statement on the eigenvalues $\pm 1$.
The eigenvectors of $A$ and $B$ can, obviously, be chosen independently of $w$.
Using \eqref{abc_condition} we compute
\begin{equation*}
\begin{split}
  w_1(U(w_1)+\id)(U(w_2)-\id)^*
   & = 4iw_1\bar w_2 (iA-w_1B)^{-1}AB^*(-iA^*-\bar w_2B^*)^{-1}\\
   & = 4iw_1\bar w_2 (iA-w_1B)^{-1}BA^*(-iA^*-\bar w_2B^*)^{-1}\\
   & = -\bar w_2(U(w_1)-\id)(U(w_2)+\id)^*\
\end{split}
\end{equation*}
which yields
\begin{equation*}
  U(w_1)U(w_2)^* - cU(w_1)+ cU(w_2)^* = \id,\ c\coloneqq\frac{w_1-\bar w_2}{w_1+\bar w_2}
\end{equation*}
and furthermore
\begin{equation*}
  ( U(w_1)+c\id )( U(w_2)^*-c\id ) = (1-c^2)\id .
\end{equation*}
For $c^2\neq 1$ we infer that $U(w_1)+c\id$ is a left-inverse and therefore a right-inverse 
as well since we are working with matrices. Thus
\begin{equation*}
  ( U(w_2)^* - c\id )( U(w_1) + c\id ) = (1-c^2)\id .
\end{equation*}
This implies that $U(w_1)$ and $U(w_2)^*$ commute if $c\neq \pm 1$ or, equivalently, $w_{1,2}\neq 0$.
$U(w)$ is unitary if $c=0$ which is equivalent to $\im(w)=0$.
\end{proof}

The matrix $U(z)$ is closely related to the spectrum of $H_L$, cf. Lemma \ref{H_L_spectrum}.

\begin{lemma}\label{bcs_special01t}
Let $\lambda\in\C\setminus\{0\}$ such that the matrix $iA-\sqrt{\lambda}B$ is invertible, which is in particular true
for $\re(\sqrt{\lambda})\neq 0$, see Lemma \ref{bcs_general01t}. Then $\lambda$ is an eigenvalue of $H_L$, cf. \eqref{H_L}, 
if and only if the matrix $e^{2iL\sqrt{\lambda}}\id + U(\lambda)\sigma_x$ is singular .
\end{lemma}
\begin{proof}
We solve the eigenvalue equation
\begin{equation*}
  -\varphi''(x) = \lambda\varphi(x), \ 
A
\begin{pmatrix}
  \varphi(L) \\ \varphi(-L)
\end{pmatrix}
-
B
\begin{pmatrix}
  -\varphi'(L) \\ \varphi'(-L)
\end{pmatrix}
= 0 .
\end{equation*}
Since $\lambda\neq 0$ the general solution is, cf. \eqref{deficiency_subspace},
\begin{equation*}
  \varphi(x) = a_1 e^{i\sqrt{\lambda}x} + a_2 e^{-i\sqrt{\lambda}x} .
\end{equation*}
The boundary conditions are satisfied if and only if there is a solution $a\coloneqq (a_1,a_2)^T\neq 0$ of the equation
\begin{equation*}
 M a = 0,\ 
 M\coloneqq  A \big[ e^{i\sqrt{\lambda}L}\id + e^{-i\sqrt{\lambda}L}\sigma_x \big] 
    - i\sqrt{\lambda}B\big[ -e^{i\sqrt{\lambda}L}\id + e^{-i\sqrt{\lambda}L}\sigma_x \big],
\end{equation*}
which is to say that the matrix $M$ is singular. Thanks to our assumption on $iA-\sqrt{\lambda}B$ we may write
\begin{equation*}
\begin{split}
   M & = e^{i\sqrt{\lambda}L}( A+i\sqrt{\lambda}B) + e^{-i\sqrt{\lambda}L}( A-i\sqrt{\lambda}B)\sigma_x \\
     & = e^{-i\sqrt{\lambda}L} (A+i\sqrt{\lambda}B)\big[ e^{2i\sqrt{\lambda}L}\id + (A+i\sqrt{\lambda}B)^{-1}(A-i\sqrt{\lambda}B)\sigma_x \big]\\
     & = -i e^{-i\sqrt{\lambda}L} (iA-\sqrt{\lambda}B)\big[ e^{2i\sqrt{\lambda}L}\id + U(\lambda) \sigma_x \big]
\end{split}
\end{equation*}
and conclude that the rightmost matrix must be singular.
\end{proof}
 
\subsection{Behaviour of \texorpdfstring{$U(z)$}{U} on the Fermi parabola}
We study how $G_L(z)$, see \eqref{fr01t02}, behaves on the Fermi parabola starting with $z$ off the spectrum.

\begin{lemma}\label{bcs_fermi00t}
There are constants $b_0>0$, $L_0>0$, and $C(b_0,L_0)\geq 0$ such that
\begin{equation}\label{bcs_fermi00t01}
  \| G_L(z) \| \leq C(b_0,L_0)\ \text{for all} \ L\geq L_0 \ \text{and}\  z=(t+ib)^2\in\Gamma_{\nu,b}^- \ \text{with}\ \ b\geq b_0 .
\end{equation}
\end{lemma}
\begin{proof}
We use a Neumann series argument. To this end, we write
\begin{equation*}
  e^{2iL\sqrt{z}}\id + U(z)\sigma_x = ( e^{2iL\sqrt{z}} \sigma_x U(z)^{-1} + \id ) U(z)\sigma_x .
\end{equation*}
Recalling the definition of $U(z)$ in \eqref{fr_bcs} one can easily see that (cf. Lemma \ref{bcs_general01t})
\begin{equation*}
  u \coloneqq \sup_{\substack{b\geq b_0\\ z\in\Gamma_{\nu,b}^-}} \|U(z)^{-1}\|<\infty
\end{equation*}
for an appropriate $b_0>0$. We note that $\sigma_x$ is unitary and choose $L_0$ such that
\begin{equation*}
  \|e^{2iL\sqrt{z}} \sigma_x U(z)^{-1}\| = e^{-2Lb} \|U(z)^{-1}\| \leq e^{-2Lb} u \leq e^{-2L_0b_0} u < 1
\end{equation*}
for $L\geq L_0$, $b\geq b_0$ thereby obtaining
\begin{equation*}
  \|G_L(z)\| = \|\sigma_x U(z)^{-1}( e^{2iL\sqrt{z}} \sigma_x U(z)^{-1} + \id )^{-1}\| \leq \frac{u}{1- e^{-2L_0b_0} u} \eqqcolon C(b_0,L_0).
\end{equation*}
This is \eqref{bcs_fermi00t01}.
\end{proof}

The behaviour near the spectrum requires a bit more effort.

\begin{lemma}\label{bcs_fermi01t}
There are constants $L_0>0$, $a>0$, and $C\geq 0$ such that
\begin{equation*}
  \| G_L(z) \| \leq C \big[ \frac{1}{Ls}\Theta(\frac{a}{L}-s) + \Theta(s-\frac{a}{L}) \big]
\end{equation*}
for all $z=(\sqrt{\nu}+is)^2\in\Gamma_\nu$ and $L\geq L_0$.
\end{lemma}
\begin{proof}
We distinguish between small and large $s$. We start with the easier case whose proof parallels that of Lemma \ref{bcs_fermi00t}. 

(i) 
Let $\frac{a}{L}\leq s$ with $a>0$ to be chosen soon. We write
\begin{equation*}
  e^{2iL\sqrt{z}}\id + U(z)\sigma_x = ( e^{2iL\sqrt{z}}\sigma_x U(z)^{-1} + \id ) U(z)\sigma_x .
\end{equation*}
Recalling the definition of $U(z)$ in \eqref{fr_bcs} one can easily see that (cf. Lemma \ref{bcs_general01t})
\begin{equation*}
  u \coloneqq \sup_{z\in\Gamma_\nu} \|U(z)^{-1}\|<\infty .
\end{equation*}
We note that $\sigma_x$ is unitary and obtain for all $z\in\Gamma_\nu$
\begin{equation*}
  \| e^{2iL\sqrt{z}} \sigma_x U(z)^{-1} \| = e^{-2Ls} \|U(z)^{-1}\| \leq e^{-a} u  < 1
\end{equation*}
if $a>0$ is sufficiently large. Hence, a Neumann series argument shows that
\begin{equation*}
  \|G_L(z)\| = \|\sigma_x U(z)^{-1}( e^{2iL\sqrt{z}} \sigma_x U(z)^{-1} + \id )^{-1}\| 
                \leq \frac{u}{1-e^{-a} u} < \infty .
\end{equation*}
(ii) 
Let $0 < s\leq \frac{a}{L}$ with $a$ as in part (i). After some simple calculations we obtain
\begin{equation*}
  e^{2iL\sqrt{z}}\id + U(z)\sigma_x
     = (1+\frac{s^2}{\nu})^{\frac{1}{2}} \tilde U(z)\sigma_x  \big[\sigma_x \tilde U(z)^* I(z) + \id \big]
\end{equation*}
with the matrices
\begin{equation*}
  \tilde U(z) \coloneqq \frac{1}{(1+\frac{s^2}{\nu})^{\frac{1}{2}}} ( U(z) + \frac{is}{\sqrt{\nu}} \id ) , \
  I(z)\coloneqq \frac{1}{(1+\frac{s^2}{\nu})^{\frac{1}{2}}} ( e^{2iL\sqrt{z}}\id - \frac{is}{\sqrt{\nu}}\sigma_x ) .
\end{equation*}
Note that $\tilde U(z)$ is unitary and $I(z)$ is normal. Thus,
\begin{equation*}
  \|\sigma_x \tilde U(z)^*I(z)\| = \|I(z)\| = \max\{ |\varkappa_1|,|\varkappa_2| \} \eqqcolon \varkappa
\end{equation*}
where the eigenvalues $\varkappa_{1,2}$ of $I(z)$ are
\begin{equation*}
  \varkappa_{1,2} = \frac{1}{(1+\frac{s^2}{\nu})^{\frac{1}{2}}} ( e^{2iL\sqrt{z}} -\sigma \frac{is}{\sqrt{\nu}} ),\ \sigma = \pm 1 .
\end{equation*}
Straightforward calculations show that
\begin{equation*}
   1 - |\varkappa_{1,2}|^2
      =  2 \frac{e^{-2Ls}}{1+\frac{s^2}{\nu}} ( \sinh(2Ls) + \frac{\sigma s}{\sqrt{\nu}}\sin(2L\sqrt\nu) ) 
      \geq 4Ls \frac{e^{-2Ls}}{1+\frac{s^2}{\nu}} ( 1 + \sigma \frac{\sin(2L\sqrt\nu)}{2L\sqrt{\nu}} )
\end{equation*}
where we used the estimate $\sinh(x)\geq x$. Moreover, using $\sin(x)\leq x$ we obtain
\begin{equation*}
  \varkappa_0 \coloneqq \inf_{L\geq L_0} ( 1 + \sigma \frac{\sin(2L\sqrt\nu)}{2L\sqrt{\nu}} ) > 0 
\end{equation*}
and thereby the lower bound
\begin{equation*}
 1 - |\varkappa_{1,2}|^2
      \geq 4Ls \frac{e^{-2Ls}}{1+\frac{s^2}{\nu}} \varkappa_0
      > 0 .
\end{equation*}
A Neumann series argument shows that $G_L(z)$ exists with
\begin{equation*}
  \| G_L(z)\| \leq (1+\frac{s^2}{\nu})^{-\frac{1}{2}} \frac{1}{1-|\varkappa|} .
\end{equation*}
Using the simple estimate
\begin{equation*}
  |\varkappa|  \leq e^{-2Ls} + \frac{s}{\sqrt{\nu}} \leq 1 + \frac{s}{\sqrt{\nu}} 
\end{equation*}
we obtain further
\begin{equation*}
  \| G_L(z)\| 
     \leq (1+\frac{s^2}{\nu})^{-\frac{1}{2}} \frac{1+|\varkappa|}{1-|\varkappa|^2}
     \leq (1+\frac{s^2}{\nu})^{\frac{1}{2}} (2+\frac{s}{\sqrt{\nu}}) \frac{1}{4Ls} e^{2Ls} \frac{1}{\varkappa_0} .
\end{equation*}
Now use $0<s\leq a/L_0$ in all factors except $1/(Ls)$ to conclude the proof.
\end{proof}

\section{Asymptotics of the energy difference \texorpdfstring{$\mathcal{E}_L(\nu)$}{}\label{asymptotics}}
We express the energy difference with the aid of Proposition \ref{energy01t}.
The results of Sections \ref{bso} and \ref{wo} ensure that the necessary conditions are satisfied.
Now, we plug the factorization \eqref{wo03t01} of the perturbation determinant
into the general formula \eqref{energy01t01} and obtain the decomposition of the energy difference  
\begin{equation}\label{asymptotics01}
  \mathcal{E}_L(\nu)= - f(\nu)\xi_L(\nu) + \mathcal{E}_{L,b}^{\text{Fumi}}(\nu) + \mathcal{E}_{L,b}^{\text{FSE}}(\nu)
  \ \text{for all}\ b>0 .
\end{equation}
Here $\xi_L$ is the spectral shift function for the finite volume system and
(see \eqref{K_infinity}, \eqref{Omega_matrix}, and \eqref{fr01t02})
\begin{align}
      \mathcal{E}_{L,b}^{\text{Fumi}}(\nu) & \coloneqq - \frac{1}{2\pi i}\int_{\Gamma_{\nu,b}} f'(z)
              \ln[\det(\id - K_{\infty,L}(z))]\, dz , \label{E_L_Fumi}\\
      \mathcal{E}_{L,b}^{\text{FSE}}(\nu)  & \coloneqq - \frac{1}{2\pi i}\int_{\Gamma_{\nu,b}} f'(z) 
              \ln[\det(\id + d_L(z)T_L(z) G_L(z))]\, dz . \label{E_L_FSE}
\end{align}
The first term is named after Fumi \cite{Fumi1955} (cf. \cite{Affleck1997})
and the correction is called finite size energy (FSE), see \cite{Affleck1997}.
It remains to show that the said terms converge and that the limits do not depend on $b$.

\subsection{Spectral shift correction\label{ssc}}
Because our main focus is on the Fumi term and the finite size energy we will look only
briefly at the spectral shift function $\xi_L(\nu)$ in \eqref{asymptotics01}. 
Since $\xi_L(\nu)\in\Z$ for $L<\infty$ it cannot converge unless it is constant. 
Nonetheless, it is known that at least for certain boundary conditions
$\xi_L(\nu)$ stays bounded as $L\to\infty$ or that even $\xi_L(\nu)=0$ for potentials $V$ being small in an appropriate sense,
see \cite[Sec. 3.5]{KuettlerOtteSpitzer2014}). Furthermore, it was shown in \cite{BorovykMakarov2012} that $\xi_L(\nu)$
is Ces\'aro-convergent to $\xi(\nu)$
\begin{equation*}
  \lim_{L\to\infty} \frac{1}{L} \int_0^L \xi_l(\nu)\, dl = \xi(\nu) .
\end{equation*}
Another way to look at it is via the micro-canonical energy difference. To this end, define
\begin{equation}\label{ssc00}
  M \coloneqq \max\{ k \mid \mu_{k,L} \leq \nu \}, \ 
  N \coloneqq \max\{ k \mid \lambda_{k,L} \leq \nu \} .
\end{equation}
Then, cf. Lemma \ref{H_L_spectrum} and \eqref{H_VL_spectrum},
\begin{equation}\label{ssc01}
  \mathcal{E}^{\text{mc}}_{N,L} \coloneqq \sum_{k=1}^N\mu_{k,L} - \sum_{k=1}^N\lambda_{k,L} ,
\end{equation}
with $f(z)=z$ for simplicity.
It is easily seen that $\mathcal{E}^{\text{mc}}_{N,L}$ is related to $\mathcal{E}_L(\nu)$ from \eqref{energy_difference} via
\begin{equation*}
  \mathcal{E}^{\text{mc}}_{N,L}
     = \mathcal{E}_L(\nu) + \sign(N-M)\sum_{k=\min\{M,N\}+1}^{\max\{M,N\}}\mu_{k,L} .
\end{equation*}
Note that starting with $\mathcal{E}^{\text{mc}}_{N,L}$ one would use the unperturbed eigenvalue $\lambda_{N,L}$ 
to determine $\nu$ and thereby $\mathcal{E}_L(\nu)$. That is why the sum is over the $\mu_{k,L}$ and
not the $\lambda_{k,L}$.
From Lifshitz's formula \eqref{lifshitz} we know that $\xi_L(\nu) = -(M-N)$. Thereby, the spectral shift correction
can be compensated
\begin{equation}\label{ssc02}
  - \nu\xi_L(\nu) + \sign(N-M)\sum_{k=\min\{M,N\}+1}^{\max\{M,N\}}\mu_{k,L} 
     = \sign(N-M)\sum_{k=\min\{M,N\}+1}^{\max\{M,N\}} (\mu_{k,L}-\nu) .
\end{equation}
It can be shown that $\mu_{N,L}-\nu=O(1/L)$ as $N\to\infty,L\to\infty$.
This was first considered in \cite{FukudaNewton1956}. Since $M-N=O(1)$ the term in \eqref{ssc02}
is $O(1/L)$ and does not contribute to the leading Fumi term but to the finite size correction
of the micro-canonical energy difference, $\mathcal{E}^{\text{mc}}_{N,L}$, as $N,L\to\infty$ and $N/(2L)\to\rho$.

\subsection{Fumi term\label{fumi}}
One would reasonably expect that in some sense the restricted Birman--Schwinger operators $K_{\infty,L}(z)$ in \eqref{E_L_Fumi} 
converge, namely to $K_\infty(z)$, and that the determinants do so as well
(cf. the Kac--Ahiezer Theorem \cite{Kac1954,Ahiezer1964}).

\begin{lemma}\label{fumi01t}
Let $V\in L^1(\R)$ satisfy $V\in\ell^{\frac{1}{2}}(L^1(\R))$ and $X^2V\in\ell^{\frac{1}{2}}(L^1(\R))$, see \eqref{birman_solomyak}.
Then, for all $2b>\|V\|_1$
\begin{equation}\label{fumi01t01}
  \lim_{L\to\infty} \mathcal{E}^{Fumi}_{L,b}(\nu) 
    = \mathcal{E}^{Fumi}_b(\nu)  
    \coloneqq - \frac{1}{2\pi i} \int_{\Gamma_{\nu,b}} f'(z) \ln\big[\det(\id-K_\infty(z))\big]\, dz .
\end{equation}
\end{lemma}
\begin{proof}
Because of $b>\|V\|_1/2$ the compact set $\Gamma_{\nu,b}$ satisfies 
\begin{equation*}
  \Gamma_{\nu,b}\subset\C\setminus[-\frac{1}{4}\|V\|_1^2,0] .
\end{equation*}
Let $z\in\Gamma_{\nu,b}$. We look at the difference
\begin{multline*}
  \ln\big[\det(\id - K_{\infty,L}(z))\big] - \ln\big[\det(\id - K_\infty(z))\big]\\
   = \ln\Big[ \frac{\det(\id-K_{\infty,L}(z))}{\det(\id-K_\infty(z))}\Big]
   = \ln\Big[ 1 - \frac{\det(\id-K_\infty(z)) - \det(\id - K_{\infty,L}(z))}{\det(\id-K_\infty(z))}\Big]
\end{multline*}
and use the inequality, see e.g. \cite[Thm. 3.4]{Simon2005},
\begin{equation*}
  |\det(\id-K_\infty(z)) - \det(\id-K_{\infty,L}(z))| 
   \leq \|K_\infty(z) - K_{\infty,L}(z)\|_1 \exp\big[ \|K_{\infty,L}(z)\|_1+\|K_\infty(z)\|_1 + 1\big] .
\end{equation*}
From Lemma \ref{wo01t} we infer that
\begin{equation*}
  \inf_{z\in\Gamma_{\nu,b}} |\det(\id-K_\infty(z))| > 0 .
\end{equation*}
Along with Lemma \ref{bso04t} this shows that
\begin{equation*}
  \ln[\det(\id-K_{\infty,L}(z))] \to \ln[\det(\id-K_\infty(z))] \ \text{as}\ L\to\infty
\end{equation*}
uniformly on $\Gamma_{\nu,b}$. Hence, the integrals converge as well.
\end{proof}

We express the Fumi term from \eqref{fumi01t01} through 
the spectral shift function $\xi$ thereby showing that it does not depend on $b$. This would also
follow from the decomposition \eqref{asymptotics01} (where the left-hand side does not depend on $b$) 
and from $\mathcal{E}^{\mathrm{FSE}}_L(\nu)\to0$ as $L\to\infty$ (proved in Section \ref{fse}).

\begin{proposition}\label{fumi02t}
Let $V\in L^1(\R)$ satisfy
$V\in\ell^{\frac{1}{2}}(L^1(\R))$ and $X^2V\in\ell^{\frac{1}{2}}(L^1(\R))$, see \eqref{birman_solomyak}.
Then, for all $2b>\|V\|_1$ we have $\mathcal{E}^{\mathrm{Fumi}}_b(\nu)=\mathcal{E}^{Fumi}(\nu)$ where
\begin{equation}
  \mathcal{E}^{\text{Fumi}}(\nu)
       = \int_{-\infty}^\nu f'(\lambda)\xi(\lambda)\, d\lambda \label{fumi02t02} .
\end{equation}
Here, the integration is actually over a bounded interval.
\end{proposition}
\begin{proof}
The conditions on $V$ ensure that the limit in \eqref{krein_xi} exists for all $\nu>0$. That is the
limiting absorption principle in Lemma \ref{bso03t}.

In order to prove formula \eqref{fumi02t02} we recall from  Lemma \ref{ssf02t}
\begin{equation*}
  \ln\det(\id-K_\infty(z)) = \int_\R \frac{\xi(\lambda)}{\lambda-z}\, d\lambda,\ \im(z)\neq 0 .
\end{equation*}
Along with
\begin{equation*}
  \frac{1}{2\pi i}\int_{\Gamma_{\nu,b}} \frac{f'(z)}{z-\lambda}\, dz =
\begin{cases}
  f'(\lambda) & \text{for}\ \lambda\ \text{inside}\ \Gamma_{\nu,b} , \\
   0          & \text{otherwise}
\end{cases}
\end{equation*}
we obtain
\begin{multline*}
  \int_{\Gamma_{\nu,b}} f'(z)\ln\det(\id-K_\infty(z))\, dz\\
     = \int_{\Gamma_{\nu,b}} f'(z) \int_\R \frac{\xi(\lambda)}{\lambda-z}\, d\lambda\, dz
     = \int_\R \xi(\lambda)\int_{\Gamma_{\nu,b}} \frac{f'(z)}{\lambda-z}\, dz \, d\lambda
     = - 2\pi i \int_{-\infty}^\nu \xi(\lambda)f'(\lambda)\, d\lambda .
\end{multline*}
Inserting this into \eqref{fumi01t01} proves \eqref{fumi02t02}.
\end{proof}

Formula \eqref{fumi02t02} is related to Kre\u\i{}n's trace formula for the function
$\lambda\mapsto \Theta(\nu-\lambda)f(\lambda)$.
Note that our proof uses Cauchy's integral formula which needs a holomorphic $f$.
A more general situation was studied in \cite{FrankPushnitski2015}, see also \cite{KohmotoKomaNakamura2012}
and \cite{peller2016}.

We illustrate the result in Proposition \ref{fumi02t}. 
The micro-canonical energy difference \eqref{ssc01} can be bounded via the variational principle 
(see also \cite{Gasymov1963,Gasymov1964} for a direct proof)
\begin{equation}\label{fumi01}
  \sum_{j=1}^N (\psi_{j,L}, V\psi_{j,L}) 
   \leq \sum_{j=1}^N (\mu_{j,L}-\lambda_{j,L})
   \leq \sum_{j=1}^N (\varphi_{j,L},V\varphi_{j,L}) ,
\end{equation}
Here, $\varphi_{j,L}$ and $\psi_{j,L}$ are the eigenvectors to the eigenvalues $\lambda_{j,L}$ and 
$\mu_{j,L}$ of the unperturbed and perturbed Hamilton operator, respectively, cf. Lemma \ref{H_L_spectrum}, \eqref{H_VL_spectrum}. 
If $V$ has a definite sign then \eqref{fumi01} provides either an effective upper or lower bound. Let, e.g.,
$V\geq 0$ then the rightmost inequality is non-trivial. We write out the right-hand side
\begin{equation*}
  \sum_{j=1}^N (\varphi_{j,L},V\varphi_{j,L}) = \int_{-L}^L V(x) \sum_{j=1}^N \varphi_{j,L}(x)^2 \, dx
\end{equation*}
and determine the thermodynamic limit of the spectral function in the special case of Dirichlet boundary conditions. 
Then,
\begin{equation*}
  \frac{1}{L} \Big[ 
    \sum_{\substack{ j\leq N \\ \text{odd}}} \cos^2(\frac{j\pi x}{2L})
  + \sum_{\substack{ j\leq N \\ \text{even}}} \sin^2(\frac{j\pi x}{2L})
              \Big]
   \to \frac{1}{\pi} \int_0^{\sqrt{\nu}} \cos^2(xt)\, dt
  + \frac{1}{\pi} \int_0^{\sqrt{\nu}} \sin^2(xt)\, dt
  = \frac{1}{\pi} \sqrt{\nu} .
\end{equation*}
We conclude
\begin{equation}\label{fumi02}
  \limsup_{N,L\to\infty} \sum_{j=1}^N (\mu_{j,L}-\lambda_{j,L})
   \leq \frac{1}{\pi} \sqrt{\nu} \int_\R V(x)\, dx ,
\end{equation}
which is in line with the result derived by Fumi \cite[Eq. 6]{Fumi1955} in dimension three using first order perturbation theory, 
see also \cite[Eq. 14]{Affleck1997} and \cite[Eq. 26]{FrankLewinLiebSeiringer2011}.
Interestingly, one would obtain \eqref{fumi02} by inserting the leading term in the high energy asymptotics 
of the spectral shift function (see e.g. \cite[Thm. 2]{AsselDimassi2008}) into \eqref{fumi02t02}
for the Fumi term.

\subsection{Finite size energy\label{fse}}
From a simple change of coordinates it can be seen that the
finite size correction $\mathcal{E}_L^{\text{FSE}}(\nu)$ in \eqref{E_L_FSE} is of the order $1/L$.
This will be made precise in Lemma \ref{fse01t}. To begin with,
we determine how $\mathcal{V}_L(s)$ in \eqref{V_L} behaves asymptotically for large $L$. 
The condition $V\in L^1(\R)$ makes $\mathcal{V}_L$ a tad smaller than the exponential function.

\begin{lemma}\label{con01t}
Let $V\in L^1(\R)$ satisfy $X^nV\in L^1(\R)$ with $n\in\N_0$. 
Then, for $s\geq 0$
\begin{equation}\label{con01t01}
  0 \leq e^{-Ls} \mathcal{V}_L^{(n)}(s) \leq \|X^nV\|_1 .
\end{equation}
If $X^{n+\alpha} V\in L^1(\R)$ with $\alpha\geq 0$ then for $s>0$
\begin{equation}\label{con01t02}
  \lim_{L\to\infty} \big[ L^\alpha e^{-Ls} \mathcal{V}_L^{(n)}(s) \big] =0 .
\end{equation}
Moreover,
\begin{equation}\label{con01t03}
  \lim_{L\to\infty} \big[ L^\alpha \int_0^b s^\alpha e^{-Ls} \mathcal{V}_L^{(n)}(s)\, ds \big] = 0 .
\end{equation}
\end{lemma}
\begin{proof}
The bound \eqref{con01t01} follows immediately from definition \eqref{V_L}.
Pick any $0<\delta<1$. For $s\geq 0$
\begin{equation}\label{con01t04}
\begin{split}
 L^\alpha e^{-Ls} \mathcal{V}_L^{(n)}(s) 
   & = L^\alpha e^{-Ls} \int_{0\leq |x|\leq\delta L}|x|^n |V(x)| e^{|x|s}\, dx 
          + L^\alpha\int_{\delta L\leq |x| \leq L} |x|^{n+\alpha} |V(x)| \frac{1}{|x|^\alpha} e^{-(L-|x|)s}\, dx\\
   & \leq L^\alpha e^{-(1-\delta)Ls} \|X^nV\|_1 + \frac{1}{\delta^\alpha} \int_{\delta L\leq |x|} |x|^{n+\alpha} |V(x)|\, dx .
\end{split}
\end{equation}
This implies \eqref{con01t02}. Finally, we integrate \eqref{con01t04}
\begin{equation*}
\begin{split}
  L^\alpha\int_0^b s^\alpha e^{-Ls}\mathcal{V}_L^{(n)}(s)\, ds
    & \leq \int_0^b L^\alpha s^\alpha e^{-(1-\delta)Ls}\, ds \|X^nV\|_1 
          + \frac{1}{\delta^\alpha} \int_0^b s^\alpha \int_{\delta L\leq |x|} |x|^{n+\alpha} |V(x)|\, dx\, ds\\
    & \leq \frac{1}{L}\int_0^\infty s^\alpha e^{-(1-\delta)s}\, ds \|X^nV\|_1 
               + \frac{1}{\delta^\alpha}\frac{b^{\alpha+1}}{\alpha+1} \int_{\delta L\leq |x|} |x|^{n+\alpha} |V(x)|\, dx
\end{split}
\end{equation*}
and obtain \eqref{con01t03}.
\end{proof}

In order to perform the limit in the FSE term we need the inequality for $n\times n$-matrices
\begin{equation}\label{det_difference_matrix}
  |\det(M_1) - \det(M_2)| \leq n \|M_1-M_2\|_2 \max\{\|M_1\|_2^{n-1},\|M_2\|_2^{n-1}\},
\end{equation}
which follows via multilinearity and Hadamard's inequality.

\begin{lemma}\label{fse01t}
Let $V\in L^1(\R)$ such that $X^nV\in L^1(\R)$ for $n=1,2,3$. We assume that $e^{2iL\sqrt{\nu}}\to e^{i\pi\eta}$,
$\eta\in\interval[open right]{-1}{1}$, as $L\to\infty$. 
Then, for all $b>0$ large enough $L\mathcal{E}^{\mathrm{FSE}}_{L,b}(\nu)\to\mathcal{E}^{\mathrm{FSE}}(\nu)$ as $L\to\infty$. Here,
\begin{equation}\label{fse01t01}
  \mathcal{E}^{\mathrm{FSE}}(\nu)
     = - \frac{1}{\pi} \sqrt{\nu} f'(\nu) \int_0^\infty \Big\{\ln\det(\id + d(s) T(\nu) G(s)) + \ln\det(\id + d(s) T(\nu) G(s))^* \Big\}\, ds
\end{equation}
with the T-matrix from \eqref{Omega_T} and with
\begin{equation}\label{fse01t02}
  G(s) \coloneqq  (e^{i\pi\eta}e^{-2s}\id + U(\nu)\sigma_x)^{-1},\ d(s)\coloneqq e^{i\pi\eta}e^{-2s},\ s\geq 0.
\end{equation}
\end{lemma}
\begin{proof}
We treat the two parts $\Gamma_{\nu,b}^-$ and $\Gamma_{\nu,b}^+$ of the Fermi parabola separately.

(a) 
On $\Gamma_{\nu,b}^-$ we use \eqref{det_difference_matrix} with $M_2=\id$ and obtain
\begin{equation*}
  |\det(\id+d_L(z) T_L(z)G_L(z)) - 1|
     \leq 2|d_L(z)| \|T_L(z)\|_2 \|G_L(z)\|_2
     \leq C e^{-2Lb} \mathcal{V}_L(2b) 
\end{equation*}
where we used \ref{wo02t01a} and Lemma \ref{bcs_fermi00t}. Now, Lemma \ref{con01t} shows that
$\det(\id+d_L(z) T_L(z)G_L(z))\to 1$ as $L\to\infty$ uniformly in $z\in\Gamma_{\nu,b}^-$. Hence, the integral
vanishes in this limit. 

(b) 
On $\Gamma_{\nu,b}^+$ we replace $S_L(z)$ and $U(z)$ by $S_L(\nu)$ and $U(\nu)$, respectively, which yields
the difference
\begin{multline*}
  \ln\Big[\frac{\det(e^{2iL\sqrt{z}}S_L(z)\sigma_x U(z)^{-1} + \id)}{\det(e^{2iL\sqrt{z}}\sigma_x U(z)^{-1} + \id)}\Big]
   - \ln\Big[\frac{\det(e^{2iL\sqrt{z}}S_L(\nu)\sigma_x U(\nu)^{-1} + \id)}{\det(e^{2iL\sqrt{z}}\sigma_x U(\nu)^{-1} + \id)}\Big]\\
   = \ln\Big[\frac{\det(e^{2iL\sqrt{z}}S_L(z)\sigma_x U(z)^{-1} + \id)}{\det(e^{2iL\sqrt{z}}S_L(\nu)\sigma_x U(\nu)^{-1} + \id)}\Big]
       - \ln\Big[\frac{\det(e^{2iL\sqrt{z}}\sigma_x U(z)^{-1} + \id)}{\det(e^{2iL\sqrt{z}}\sigma_x U(\nu)^{-1} + \id)}\Big].
\end{multline*}
We rewrite the first term
\begin{equation*}
  \frac{\det(e^{2iL\sqrt{z}}S_L(z)\sigma_x U(z)^{-1} + \id)}{\det(e^{2iL\sqrt{z}}S_L(\nu)\sigma_x U(\nu)^{-1} + \id)}
     = 1 + \frac{\det(e^{2iL\sqrt{z}}S_L(z)\sigma_x U(z)^{-1} + \id) 
                  - \det(e^{2iL\sqrt{z}}S_L(\nu)\sigma_x U(\nu) + \id)}{\det(e^{2iL\sqrt{z}}S_L(\nu)\sigma_x U(\nu)^{-1} + \id)}
\end{equation*}
and estimate the numerator via \eqref{det_difference_matrix}
\begin{equation*}
\begin{split}
  \lefteqn{| \det(e^{2iL\sqrt{z}}S_L(z)\sigma_x U(z)^{-1} + \id) - \det(e^{2iL\sqrt{z}}S_L(\nu)\sigma_x U(\nu)^{-1} + \id) |}\\
    & \leq 2 e^{-2Ls}\|S_L(z)\sigma_x U(z)^{-1} - S_L(\nu)\sigma_x U(\nu)^{-1}\|_2 \\
    & \quad \times \max\{ \|e^{2iL\sqrt{z}}S_L(z)\sigma_x U(z)^{-1}+\id\|_2, \|e^{2iL\sqrt{z}}S_L(\nu)\sigma_x U(\nu)^{-1}+\id\|_2 \} .
\end{split}
\end{equation*}
Now with Lemma \ref{wo02t} (see also Lemma \ref{bcs_fermi01t})
\begin{equation*}
  \|S_L(z)\sigma_x U(z)^{-1} - S_L(\nu)\sigma_x U(\nu)^{-1}\|_2
      \leq \|S_L(z)-S_L(\nu)\|_2 \|U(z)^{-1}\| + \|U(z)^{-1}-U(\nu)^{-1}\|_2 
      \leq C s .
\end{equation*}
The second term can be treated likewise
\begin{equation*}
  \frac{\det(e^{2iL\sqrt{z}}\sigma_x U(z)^{-1} + \id)}{\det(e^{2iL\sqrt{z}}\sigma_x U(\nu)^{-1} + \id)}
     = 1 + \frac{\det(e^{2iL\sqrt{z}}\sigma_x U(z)^{-1} + \id) 
                  - \det(e^{2iL\sqrt{z}}\sigma_x U(\nu)^{-1} + \id)}{\det(e^{2iL\sqrt{z}}\sigma_x U(\nu)^{-1} + \id)} .
\end{equation*}
The numerator can be estimated
\begin{equation*}
\begin{split}
  \lefteqn{|\det(e^{2iL\sqrt{z}}\sigma_x U(z)^{-1} + \id) - \det(e^{2iL\sqrt{z}}\sigma_x U(\nu)^{-1} + \id)|}\\
    & \leq 2 e^{-2Ls}\|U(z)^{-1}-U(\nu)^{-1}\|_2 
         \max\{\|e^{2iL\sqrt{z}}\sigma_x U(z)^{-1} + \id\|_2 ,\|e^{2iL\sqrt{z}}\sigma_x U(\nu)^{-1} + \id\|_2\} \\
    & \leq C e^{-2Ls} s .
\end{split}
\end{equation*}
All in all, Lemma \ref{con01t} proves the statement.
\end{proof}

Next, we cast the integral in \eqref{fse01t01} into a more compact form.

\begin{lemma}\label{fse02t}
The finite size energy \eqref{fse01t01} can be written as (see \eqref{Omega_T}, \eqref{fse01t02})
\begin{equation}\label{fse02t01}
   \mathcal{E}^{\text{FSE}}(\nu)
     = -\frac{\sqrt{\nu}}{2\pi}f'(\nu) \int_0^\infty 
          \ln\big\{\det\big[ \id + ( \cosh(s) + \re(e^{i\eta}\sigma_xU(\nu)^*) )^{-1}\re(e^{i\eta}\sigma_xU(\nu)^*T(\nu)) \big]\big\} \, ds .
\end{equation}
\end{lemma}
\begin{proof}
The integrand in \eqref{fse01t01} is
\begin{gather*}
  \ln\det(\id + d(s) T(\nu) G(s)) + \ln\det(\id + d(s) T(\nu) G(s))^*  
     = \ln\det(F(s)),\\
  F(s) \coloneqq (\id+d(s)T(\nu)G(s))^*(\id+d(s)T(\nu)G(s)) .
\end{gather*}
Note that $T(\nu)^*T(\nu)=-T(\nu)-T(\nu)^*$. Then,
\begin{equation*}
\begin{split}
  F(s) 
   & = \id + \bar d(s) G(s)^*T(\nu)^* + d(s) T(\nu)G(s) + |d(s)|^2 G(s)^*T(\nu)^*T(\nu)G(s) \\
   & = \id + \bar d(s) G(s)^*T(\nu)^*(\id-d(s)G(\nu)) + d(s)(\id-\bar d(s)G(s)^*T(\nu)G(s) \\
   & = \id + G(s)^*\big[ \bar d(s)T(\nu)^*U(\nu)\sigma_x + d(s) \sigma_xU(\nu)^*T(\nu) \big]G(s)
\end{split}
\end{equation*}
where we used $\id-d(s)G(s)= U(\nu)\sigma_xG(s)$. In the determinant we commute the factors
\begin{equation*}
  \det(F(s))
    = \det\big[ \id + G(s)G(s)^*( \bar d(s)T(\nu)^*U(\nu)\sigma_x + d(s) \sigma_xU(\nu)^*T(\nu) ) \big] .
\end{equation*}
Recalling the definition of $d(s)$ we obtain
\begin{equation*}
  G(s)G(s)^* 
   = e^{2s} \big[ (e^{-2s} + e^{2s})\id + e^{-i\eta}U(\nu)\sigma_x + e^{i\eta}\sigma_xU(\nu)^* \big]^{-1}
\end{equation*}
and thereby
\begin{equation*}
  \det(F(s))
   =  \det\big[ \id + ( \cosh(2s) + \re(e^{i\eta}\sigma_xU(\nu)^*) )^{-1}\re(e^{i\eta}\sigma_xU(\nu)^*T(\nu)) \big] .
\end{equation*}
Finally, substitute $s\to s/2$ in the integral to prove \eqref{fse02t01}.
\end{proof}

Surprisingly, the integral in \eqref{fse02t01} can be evaluated.
It is instructive to consider a more general integral.
The method is motivated by that used for integrals of rational functions over the half-line. 
It involves a very simple version of a Riemann--Hilbert problem.

\begin{lemma}\label{integral01t}
Let $H_1$ and $H_2$ be self-adjoint matrices with $H_1>0$ and $H_1+H_2>0$. Let $f(t)\coloneqq\cosh(t)-1$. Then,
\begin{equation}\label{integral01t01}
  \int_0^\infty \ln\big[\det( \id + (f(t)\id+H_1)^{-1}H_2 ) \big] \, dt =
     \frac{1}{2} \tr\big[ \arcosh^2(H_1+H_2-\id) - \arcosh^2(H_1-\id) \big].
\end{equation}
\end{lemma}
\begin{proof}
We start with the more general integral
\begin{equation*}
  I \coloneqq \int_0^\infty \ln\big[\det(\id+(f(t)\id+H_1)^{-1}H_2)\big]\, dt
\end{equation*}
where $f:\R^+\to\R$ is a differentiable function with 
\begin{equation*}
   f'(t)> 0\ \text{for}\ t>0,\
   f(0)=0,\ f(\infty)=\infty,\ \lim_{x\to\infty} f^{-1}(x)/x =0 .
\end{equation*}
We substitute
\begin{equation*}
  t = f^{-1}(x),\ dt = \frac{d}{dx}(f^{-1}(x))\, dx
\end{equation*}
and integrate by parts
\begin{equation*}
\begin{split}
  I  & = \int_0^\infty \ln\big[\det( \id + (x\id+H_1)^{-1}H_2)\big] \frac{d}{dx}( f^{-1}(x) )\, dx\\
     & = - \int_0^\infty \tr\big[ (x\id+H_1+H_2)^{-1} - (x\id+H_1)^{-1} \big] f^{-1}(x)\, dx\\
     & = \int_0^\infty r(x) g(x)\, dx
\end{split}
\end{equation*}
where
\begin{equation*}
   r(x) \coloneqq - \tr\big[ (x\id+H_1+H_2)^{-1} - (x\id+H_1)^{-1} \big],\ g(x)\coloneqq f^{-1}(x).
\end{equation*}
Note that $r$ is a rational function whose poles are the eigenvalues of $-H_1$ and $-H_1-H_2$ which we assumed to be negative. 
Thus, $r$ has no poles on $\R^+$. Integrals of this type can be evaluated if one finds a
holomorphic function $h:\C\setminus\interval[open right]{0}{\infty}\to\C$ satisfying on the cut
\begin{equation*}
  h(x+i0) - h(x-i0) = g(x),\ x\in[0,\infty[ .
\end{equation*}
Then, the keyhole integration contour
\begin{equation*}
  \Gamma_{\varepsilon,R,r} \coloneqq \Gamma_\varepsilon^+ \cup \Gamma_r \cup \Gamma_\varepsilon^- \cup \Gamma_R ,
\end{equation*}
see Fig. \ref{f_keyhole},
\begin{figure}[bth]
\includegraphics[width=0.75\textwidth]{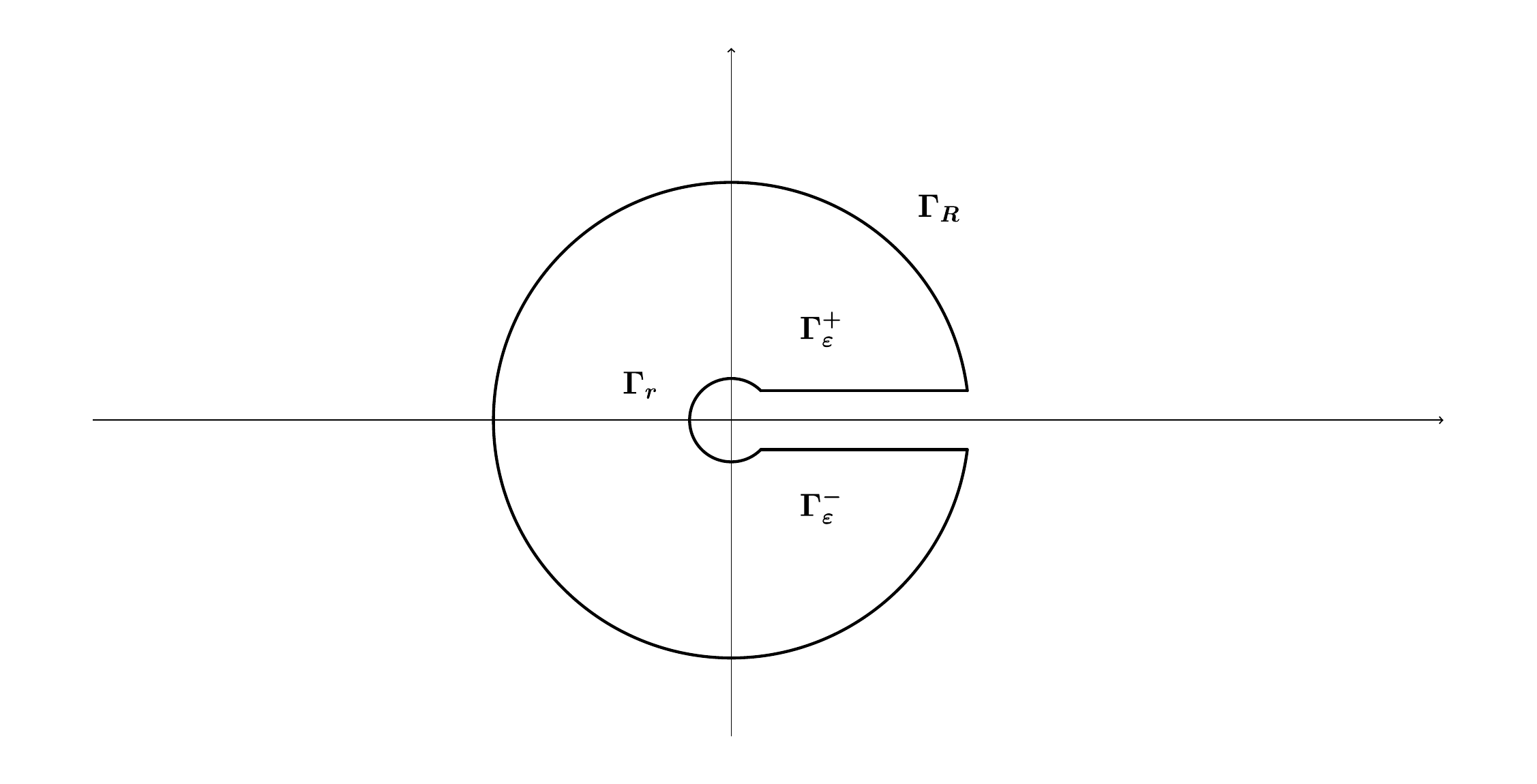}
\caption{The keyhole integration contour in the complex plane.}
\label{f_keyhole}
\end{figure}
along with the residue theorem can be used to show that
\begin{equation*}
  I = 2\pi i \sum_{z\in\C\setminus\interval[open right]{0}{\infty}} \res(r(z)h(z)) .
\end{equation*}
In our case,
\begin{equation*}
  f(t) = \cosh(t) - 1,\ 
  g(x)\coloneqq f^{-1}(x) = \arcosh(x+1) .
\end{equation*}
Now, we are left with finding the function $h$ which would generally
require to solve a Riemann--Hilbert problem. Here we can use
the area function itself. The principal branch is given by the formula
\begin{equation*}
  \arcosh(z) = \ln( \pm(z^2-1)^{\frac{1}{2}} + z),\ z\in\C\setminus\interval[open]{-\infty}{1},\ \re(z)\gtrless 0 .
\end{equation*}
It is holomorphic on $\C$ except for $\interval[open left]{-\infty}{1}$ where it satisfies
\begin{equation*}
  \lim_{\varepsilon\to+0} \arcosh(x\pm i\varepsilon) = 
\begin{cases}
  \pm\ln( i(1-x^2)^{\frac{1}{2}} + x )        & -1 < x \leq 1 ,\\
  \pm \pi i + \ln( (x^2-1)^{\frac{1}{2}} - x) & x\leq -1 .
\end{cases}
\end{equation*}
This motivates us to put
\begin{equation*}
  h(z) \coloneqq -\frac{1}{4\pi i} \arcosh^2(-z-1) .
\end{equation*}
The outer minus sign is necessary since the minus sign in the argument swaps the upper and the lower halfplane.
Obviously, $h$ is holomorphic at least on $\C\setminus\interval[open right]{-2}{\infty}$. But
\begin{equation*}
  h(x+i0) - h(x-i0) = 0,\ -2 \leq x < 0 ,
\end{equation*}
shows that it extends to $\C\setminus\interval[open right]{0}{\infty}$. Furthermore,
\begin{equation*}
  h(x+i0) - h(x-i0) 
    = \ln\big( ((x+1)^2-1)^{\frac{1}{2}} + x +1 \big)
    = \arcosh(x+1),\ x\geq 0 .
\end{equation*}
Hence,
\begin{equation*}
\begin{split}
  I & = 2\pi i \sum_{z\in\C\setminus[0,\infty[} \res\Big\{ \tr\big[ (z\id+H_1+H_2)^{-1} - (z\id+H_1)^{-1} \big]
        \frac{1}{4\pi i}\arcosh^2(-z-1)\Big\} \\
    & = \frac{1}{2}  \tr\big[ \arcosh^2(H_1+H_2-\id) - \arcosh^2(H_1-\id) \big] .
\end{split}
\end{equation*}
This proves \eqref{integral01t01}.
\end{proof}

Finally, we can compute the finite size energy.

\begin{proposition}\label{fse03t}
The finite size energy is given by (cf. \eqref{fr_bcs} and \eqref{s-matrix02})
\begin{equation}\label{fse03t01}
   \mathcal{E}^{\text{FSE}}(\nu) = \frac{\sqrt{\nu}}{4\pi}f'(\nu)
     \tr\big[ \arccos^2(\re(e^{i\eta}\sigma_xU(\nu)^*S(\nu))) - \arccos^2(\re(e^{i\eta}\sigma_xU(\nu)^*)) \big] .
\end{equation}
\end{proposition}
\begin{proof}
We define
\begin{equation*}
  H_1 \coloneqq \id + \re(e^{i\eta}\sigma_xU(\nu)^*),\ H_2\coloneqq \re(e^{i\eta}\sigma_xU(\nu)^*T(\nu)) .
\end{equation*}
Note that $0\leq H_1$ and $0\leq H_1+H_2$. Inserting $H_{1,2}$ into Lemma \ref{integral01t} we obtain
\begin{equation*}
\begin{split}
  \lefteqn{\int_0^\infty\ln\det\big[ \id + ( \cosh(s)\id + \re(e^{i\eta}\sigma_xU(\nu)^*))^{-1}\re(e^{i\eta}\sigma_xU(\nu)^*T(\nu)) \big]\, ds}\\
    & = \frac{1}{2}\tr\big[ \arcosh^2(\re(e^{i\eta}\sigma_xU(\nu)^*)+\re(e^{i\eta}\sigma_xU(\nu)^*T(\nu)))
                           - \arcosh^2(\re(e^{i\eta}\sigma_xU(\nu)^*)) \big]\\
    & = \frac{1}{2}\tr\big[ \arcosh^2(\re(e^{i\eta}\sigma_xU(\nu)^*S(\nu)))
                           - \arcosh^2(\re(e^{i\eta}\sigma_xU(\nu)^*)) \big] .
\end{split}
\end{equation*}
Since $U(\nu)$ and $S(\nu)$ are unitary the eigenvalues of
$\re(e^{i\eta}\sigma_xU(\nu)^*)$ and $\re(e^{i\eta}\sigma_xU(\nu)^*S(\nu))$
lie in the interval $[-1,1]$ whereby $\arcosh(z)=\pm i\arccos(z)$. This shows \eqref{fse03t01}.
\end{proof}

\section{Asymptotics on the half-line\label{hl}}
The study of the asymptotics of the energy difference on the half-line can be pursued along
the same line of ideas as on the whole line. But we do not know how to extend Gebert's methods for the half-line to the whole line.

So, let $\hilbert_\infty\coloneqq L^2(\R^+)$, $\dom(H_\infty)\subset\hilbert$ be dense,
and $H_\infty:\dom(H_\infty)\to\hilbert_\infty$, $H_\infty \coloneqq -\frac{d^2}{dx^2}$
the free Schr\"odinger operator defined on the half-line with boundary condition at the origin
\begin{equation}\label{hl01}
  a\varphi(0) = b\varphi'(0),\ a\bar b=b\bar a,\ |a|^2+|b|^2=1.
\end{equation}
We compute the resolvent $R_\infty(z)\coloneqq (z\id-H_\infty)^{-1}$. The function
\begin{equation}\label{hl02}
  \varepsilon_0(z;x) \coloneqq \frac{ia-\sqrt{z}b}{ia+\sqrt{z}b}e^{i\sqrt{z}x} - e^{-i\sqrt{z}x}
\end{equation}
is a formal solution, i.e. it is not in $\hilbert_\infty$,
to the differential equation $H_\infty \varepsilon_0(z)=z\varepsilon_0(z)$
and satisfies the boundary condition \eqref{hl01} at the origin.
Note the property
\begin{equation}\label{hl03}
  \overline{\varepsilon_0(z)} = - \frac{ia+\sqrt{\bar z}b}{ia-\sqrt{\bar z}b} \varepsilon_0(\bar z) .
\end{equation}
Then, the Green function of $R_\infty(z)$ is
\begin{equation}\label{hl04}
  R_\infty(z;x,y) = \frac{i}{2\sqrt{z}}
\begin{cases}
  \varepsilon_0(z;x)e^{i\sqrt{z}y} & \text{for}\ x\leq y , \\
  e^{i\sqrt{z}x}\varepsilon_0(z;y) & \text{for}\ x\geq y .
\end{cases}
\end{equation}
We restrict the operator $H_\infty$ to $\hilbert_L\coloneqq L^2(\Lambda_L)$, $\Lambda_L\coloneqq [0,L]$, thereby obtaining the maximal operator
$\tilde H:\dom(\tilde H)\to\hilbert_L$, $\tilde H\coloneqq -\frac{d^2}{dx^2}$.
Integration by parts shows that
\begin{equation}\label{hl05}
  (\varphi,\tilde H\psi) - (\tilde H\varphi,\psi)
  = (\Gamma_1\varphi,\Gamma_2\psi) - (\Gamma_2\varphi,\Gamma_1\psi)
    \ \text{with}\ 
  \Gamma_1\varphi = \varphi(L),\ \Gamma_2\varphi = -\varphi'(L) .
\end{equation}
The deficiency subspace is (cf. \eqref{abc_deficiency_subspace}, \eqref{hl02})
\begin{equation}\label{hl06}
  \mathcal{N}_z = \mathspan\{ \varepsilon_0(z) \} .
\end{equation}
The boundary condition at $x=L$ is parametrized by the $1\times 1$-matrices $A$ and $B$.
We define the self-adjoint restriction $H_L\coloneqq\tilde H|_{\dom(H_L)}$ via
\begin{equation}\label{hl07}
  \dom(H_L) \coloneqq \{ \varphi\in\dom(\tilde H) \mid (A\Gamma_1-B\Gamma_2)\varphi =0 \}
\end{equation}
with resolvent
\begin{equation}\label{hl08}
  R_L(z)\coloneqq (z\id- H_L)^{-1}:\hilbert_L\to\hilbert_L,\
  z\in\C\setminus\sigma(H_L) .
\end{equation}
The restrictions of $\Gamma_{1,2}$ to $\mathcal{N}_z$ are just numbers or, more precisely, multiplication operators by thoses numbers
\begin{equation}\label{hl09}
  \Gamma_1 \varepsilon_0(z) = \frac{ia-\sqrt{z}b}{ia+\sqrt{z}b} e^{i\sqrt{z}L} - e^{-i\sqrt{z}L},\
  \Gamma_2 \varepsilon_0(z) = -i\sqrt{z} \big[ \frac{ia-\sqrt{z}b}{ia+\sqrt{z}b} e^{i\sqrt{z}L} + e^{-i\sqrt{z}L} \big] .
\end{equation}
Furthermore, using \eqref{hl03} one obtains
\begin{equation}\label{hl10}
  \Gamma_1R_{\infty,L}(z)\varphi
    = \frac{1}{2i\sqrt{z}}\frac{ia-\sqrt{z}b}{ia+\sqrt{z}b} e^{i\sqrt{z}L} ( \varepsilon_0(\bar z),\varphi)
  ,\
  \Gamma_2R_{\infty,L}(z)\varphi = -\frac{1}{2}\frac{ia-\sqrt{z}b}{ia+\sqrt{z}b} e^{i\sqrt{z}L} (\varepsilon_0(\bar z),\varphi) .
\end{equation}
Now (cf. \eqref{fr01t02})
\begin{equation}\label{hl11}
  G_L(z) = - \frac{i}{2\sqrt{z}}e^{2i\sqrt{z}L}
    \big[ e^{2i\sqrt{z}L} - \frac{ia+\sqrt{z}b}{ia-\sqrt{z}b}\frac{iA+\sqrt{z}B}{iA-\sqrt{z}B} \big]^{-1} .
\end{equation}
Note that on the half-line the relation between the Fermi energy $\nu$ and the system length $L$ is
\begin{equation*}
  L \sqrt{\nu} = \pi( N+\eta ), \ 0\leq \eta < 1 .
\end{equation*}
The relevant quantity for the finite size energy is (cf. \eqref{fse01t02})
\begin{equation}\label{hl12}
  W(z) : = - e^{-2i\pi\eta}\frac{ia+\sqrt{z}b}{ia-\sqrt{z}b}\frac{iA+\sqrt{z}B}{iA-\sqrt{z}B} ,
\end{equation}
which appears in the analogue of Proposition \ref{fse03t}
\begin{equation}\label{hl13}
   \mathcal{E}^{\text{FSE}}(\nu) = \frac{\sqrt{\nu}}{4\pi}f'(\nu)
     \big[ \arccos^2(\re(W(\nu)^*S(\nu))) - \arccos^2(\re(W(\nu))) \big] .
\end{equation}
We compare \eqref{hl13} with Gebert's result \cite[(1.4), (2.12)]{Gebert2015}
\begin{equation*}
  \mathcal{E}^\text{mc}_{N,L}(\nu)
    = \int_{-\infty}^\nu \xi(\lambda)\, d\lambda 
      + \frac{\sqrt{\nu}\pi}{L} E^\text{mc}_\eta(\nu) + o(\frac{1}{L}),\ 
     E^\text{mc}_\eta(\nu) = (1-2\eta)\xi(\nu) + \xi(\nu)^2 .
\end{equation*} 
To this end, we write
\begin{equation*}
  W(\nu) = e^{i\pi} e^{-2\pi i\eta}  e^{2\pi i\vartheta} ,\ S(\nu) = e^{-2\pi i\xi(\nu)} .
\end{equation*}
The latter is the Birman--Kre\u\i{}n formula (cf. Lemma \ref{ssf03t}). Then,
\begin{equation*}
  \re(W(\nu)^*S(\nu)) = \cos(2\pi( -\frac{1}{2} + \eta  - \vartheta - \xi(\nu) )),\
  \re(W(\nu)) = \cos(2\pi( \frac{1}{2} - \eta  + \vartheta))
\end{equation*}
and furthermore
\begin{multline*}
  \arccos^2(\re(W(\nu)^*S(\nu))) - \arccos^2(\re(W(\nu)) ) \\
    = 4\pi^2 \big[ ( -\frac{1}{2} + \eta  - \vartheta - \xi(\nu) )^2 - (\frac{1}{2} - \eta + \vartheta) \big] 
    = \xi(\nu)^2 + (1-2\eta+2\vartheta) \xi(\nu) .
\end{multline*}
For Dirichlet boundary conditions at $x=0$ and $x=L$ we have $\vartheta=0$ which yields exactly the coefficient
as in \cite{Gebert2015}.

\appendix
\section{Scattering theory\label{scattering}}
General results from scattering theory show that the operator $\hat\Omega^\pm(\nu)$ in \eqref{Omega_matrix} 
is related to the scattering matrix (see e.g. \cite[ Thm. 4]{BirmanEntina1967}).
Here, we give an elementary derivation.
To this end, we provide some scattering theoretic background for the one-dimensional
Schr\"odinger equation
\begin{equation}\label{st01}
  -\psi''(x) + V(x)\psi(x) = k^2 \psi(x),\ \im(k) \geq 0,
\end{equation}
both on the line and on the half-line.
We write $k^2$ instead of $z$ to stay closer to the usual notation in scattering theory.
For scattering on the line see, e.g., \cite[Sec. 2, \S 3, in particular pp. 145, 146]{DeiftTrubowitz1979}
and on the half-line see \cite[Ch. 4]{Yafaev2010}.

\subsection{Scattering on the whole line\label{stl}}
In a scattering experiment, a plane wave coming from, say, $+\infty$ interacts
with a potential whereby one part is reflected back to $+\infty$ and thus superposes
the incoming wave. Another part moves toward $-\infty$.
This scenario is described by the scattering solution $u_+$ of \eqref{st01}
\begin{equation}\label{stl01}
  u_+(x) \sim
\begin{cases}
  \mathscr{t}_1(k) e^{-ikx}          & \text{for}\ x\to-\infty ,\\
  e^{-ikx} + \mathscr{r_1}(k) e^{ikx} & \text{for}\ x\to\infty ,
\end{cases}
\end{equation}
where $\mathscr{t}_1(k)$ and $\mathscr{r}_1(k)$ are the transmission and reflection coefficient,
respectively. Analogously, $u_-$ describes scattering from $-\infty$
\begin{equation}\label{stl02}
  u_-(x) \sim
\begin{cases}
  \mathscr{t}_2(k) e^{ikx}           & \text{for}\ x\to\infty ,\\
  e^{ikx} + \mathscr{r}_2(k) e^{-ikx} & \text{for}\ x\to-\infty .
\end{cases}
\end{equation}
It can be shown that $\mathscr{t}_1(k)=\mathscr{t}_2(k)\eqqcolon \mathscr{t}(k)$, which is reasonable on physical grounds since the wave
moves through the entire potential. In order to describe the asymptotics in \eqref{stl01} and \eqref{stl02}
we use the so-called Jost solutions $\psi_\pm(k;\cdot)$ of \eqref{st01}
\begin{equation}\label{stl03}
  \psi_+(k;x) \sim e^{ikx},\ \text{for}\ x\to\infty,\
  \psi_-(k;x) \sim e^{-ikx},\ \text{for}\ x\to-\infty .
\end{equation}
Their existence and properties can be obtained via the Lippmann--Schwinger equation
\begin{equation}\label{lippmann_schwinger}
\begin{aligned}
  \psi_+(k;x) & = e^{ikx} + \frac{1}{k}\int_x^\infty \sin(k(y-x)) V(y) \psi_+(k;y)\, dy,\ x\in\R ,\\
  \psi_-(k;x) & = e^{-ikx} + \frac{1}{k}\int_{-\infty}^x \sin(k(x-y)) V(y) \psi_-(k;y)\, dy,\ x\in\R .
\end{aligned}
\end{equation}
In order to study their analytical properties it is more convenient to consider the functions
\begin{equation}\label{stl03_modified}
  m_+(k;x) \coloneqq e^{-ikx}\psi_+(k;x)\ \text{and}\ m_-(k;x)\coloneqq e^{ikx}\psi_-(k;x) . 
\end{equation}
They satisfy the equations
\begin{equation}\label{lippmann_schwinger_modified}
\begin{aligned}
  m_+(k;x) & = 1 + \int_x^\infty D_k(y-x)V(y)m_+(k;y)\, dy,\ x\in\R, \\
  m_-(k;x) & = 1 + \int_{-\infty}^x D_k(x-y)V(y)m_-(k;y)\, dy,\ x\in\R,
\end{aligned}
\end{equation}
with the kernel function
\begin{equation*}
  D_k(x) \coloneqq \int_0^x e^{2iky}\, dy = \frac{1}{2ik} ( e^{2ikx} - 1 ) .
\end{equation*}
We will need a version of Gronwall's lemma.

\begin{lemma}\label{volterra} 
Let $r:\R\to\C$ be bounded. For $k:\R\times\R\to\C$ define 
\begin{equation*}
  \tilde k(x,y) \coloneqq 
\begin{cases}
   \sup_{x\leq t\leq y} |k(t,y)| & \text{for}\ x\leq y , \\
   \tilde k(x,y) = 0          &  \text{otherwise} ,
\end{cases}
\end{equation*}
and assume that $\tilde{k}(x,\cdot)\in L^1(\R)$ for all $x\in\R$. Then, the Volterra equation
\begin{equation*}
  u(x) = r(x) + \int_x^\infty k(x,y) u(y)\, dy
\end{equation*}
has a unique solution $u$ satisfying
\begin{equation*}
  |u(x)-r(x)| \leq \int_x^\infty |r(y_1)| \tilde k(x,y_1)\exp\Big[ \int_x^{y_1} \tilde k(x,y_2)\, dy_2 \Big]\, dy_1.
\end{equation*}
Moreover, if $|r(x)|\leq\tilde r(x)$ where $\tilde r\in C^1(\R)$ is bounded 
and $\lim_{x\to\infty}\tilde r(x)\eqqcolon \tilde r(\infty)$ exists then
\begin{equation*}
  |u(x)-r(x)| \leq \tilde r(\infty)\exp\big[\int_x^\infty \tilde k(x,y)\, dy \big] - \tilde r(x)
               -\int_x^\infty \tilde r'(y_1) \exp\big[ \int_x^{y_1} \tilde k(x,y_2)\, dy_2\big] \, dy_1 .
\end{equation*}
\end{lemma}
\begin{proof}
We only go through the major steps. Unique solvability follows via successive iteration. 
We define $u_0\coloneqq r$ and
\begin{equation*}
  u_{n+1}(x) \coloneqq \int_x^\infty k(x,y) u_n(y)\, dy,\ n\in\N_0 .
\end{equation*}
The solution can then be written as
\begin{equation*}
  u(x) - r(x) = \sum_{n=1}^\infty u_n(x) .
\end{equation*}
In order to ensure uniform convergence we show the estimate
\begin{equation*}
  |u_n(x)| \leq \frac{1}{(n-1)!} \int_x^\infty |r(y_1)|\tilde k(x,y_1) \big[ \int_x^{y_1} \tilde k(x,y_2)\, dy_2\big]^{n-1}\, dy_1,\ n\geq 1.
\end{equation*}
This is obviously true for $n=1$. Now, for $n+1$
\begin{equation*}
\begin{split}
  |u_{n+1}(x)| 
    & \leq \int_x^\infty \tilde k(x,y) |u_n(y)| \, dy \\
    & \leq \frac{1}{(n-1)!}\int_x^\infty \tilde k(x,y_1) \int_{y_1}^\infty |r(y_2)| \tilde k(y_1,y_2)
              \big[\int_{y_1}^{y_2} \tilde k(y_1,y_3)\, dy_3\big]^{n-1}\, dy_2\, dy_1 .
\end{split}
\end{equation*}
Since $x\leq y_1$ we have $\tilde k(y_1,y)\leq \tilde k(x,y)$ and thus
\begin{equation*}
\begin{split}
  |u_{n+1}(x)| 
     & \leq \frac{1}{(n-1)!}\int_x^\infty \tilde k(x,y_1) \int_{y_1}^\infty |r(y_2)| \tilde k(x,y_2)
              \big[\int_{y_1}^{y_2} \tilde k(x,y_3)\, dy_3\big]^{n-1}\, dy_2\, dy_1\\
     & = \frac{1}{(n-1)!} \int_x^\infty |r(y_2)|\tilde k(x,y_2) 
           \int_x^{y_2} \tilde k(x,y_1)\big[ \int_{y_1}^{y_2} \tilde k(x,y_3)\, dy_3\big]^{n-1}\, dy_1\, dy_2\\
     & = - \frac{1}{n!} \int_x^\infty |r(y_2)|\tilde k(x,y_2)
           \int_x^{y_2}\frac{\partial}{\partial y_1} \big[ \int_{y_1}^{y_2}\tilde k(x,y_3)\, dy_3\big]^n \, dy_1\, dy_2\\
     & = \frac{1}{n!} \int_x^\infty |r(y_2)|\tilde k(x,y_2)\big[ \int_x^{y_2}\tilde k(x,y_3)\, dy_3\big]^n\, dy_2 .
\end{split}
\end{equation*}
That proves the bound. Now,
\begin{equation*}
\begin{split}
  |u(x)-r(x)|
    & \leq \sum_{n=1}^\infty |u_n(x)|\\
    & \leq \sum_{n=1}^\infty \frac{1}{(n-1)!} \int_x^\infty r(y_1)\tilde k(x,y_1)\big[ \int_x^{y_1}\tilde k(x,y_2)\, dy_2\big]^{n-1}\, dy_1\\
    &  =   \int_x^\infty r(y_1) \tilde k(x,y_1)\exp\big[\int_x^{y_1}\tilde k(x,y_2)\, dy_2\big]\, dy_1
\end{split}
\end{equation*}
which shows the estimate. If we estimate further $r\leq\tilde r$ the second bound follows simply via an integration by parts.
\end{proof}

We note some fundamental properties of the Jost solutions starting with the $x$-space properties.

\begin{lemma}\label{stl01t}
Let $V\in L^1(\R)$. Then, for all $\im(k)\geq 0$, $k\neq 0$ the Lippmann--Schwinger equations \eqref{lippmann_schwinger_modified} 
have a unique solution that solves \eqref{st01} and has the asymptotics \eqref{stl03}. We have the estimate
\begin{equation}\label{stl01t01}
  |m_+(k;x) - 1 | \leq \exp\big[ \frac{1}{|k|}\int_x^\infty |V(y)|\, dy\big] - 1,\ x\in\R .
\end{equation}
In particular,
\begin{equation}\label{stl01t02}
  |m_+(k;x)| \leq \exp\big[ \frac{1}{|k|} \|V\|_1\big],\ x\in\R .
\end{equation}
\end{lemma}
\begin{proof}
Cf. \cite[2. Lemma 1, (i)]{DeiftTrubowitz1979}.
We use Lemma \ref{volterra}. The estimate
\begin{equation*}
  | D_k(x) | \leq \frac{1}{|k|},\ \im(k)\geq 0,\ k\neq 0,
\end{equation*}
immediately implies \eqref{stl01t01}.
\end{proof}

The restriction $k\neq 0$ in Lemma \ref{stl01t} can be removed if the potential $V$ falls off fast enough.
The bound \eqref{stl01t01} is replaced by, actually, two bounds: one reflecting the correct asymptotic behaviour
at $x=+\infty$ and the other one at $x=-\infty$.

\begin{lemma}\label{stl02t}
Let $V\in L^1(\R)$ satisfy $XV\in L^1(\R)$.
Then, for all $\im(k)\geq 0$ the Lippmann--Schwinger equations \eqref{lippmann_schwinger_modified} 
have a unique solution that solves \eqref{st01} and has the asymptotics \eqref{stl03}. We have the bounds
\begin{align}\label{stl02t01}
  |m_+(k;x) - 1 | & \leq \exp\big[ \int_x^\infty (y-x)|V(y)|\, dy \big] - 1 , \\
\label{stl02t02}
       |m_+(k;x)| & \leq 2 (1+|x|) e^{\|V\|_1 + 2\|XV\|_1} .
\end{align}
\end{lemma}
\begin{proof}
In Lemma \ref{volterra} we use the estimate
\begin{equation*}
  | D_k(x) | \leq |x| ,\ \im(k)\geq 0,
\end{equation*}
which immediately yields \eqref{stl02t01}.
Whereas \eqref{stl02t01} displays the correct behaviour for $x\to+\infty$ the bound behaves like $e^{|x|}$ for $x\to-\infty$.
Therefore, in a second step, we refine our bound. To begin with,
\begin{equation*}
  |m_+(k;x)| \leq \exp\big[ \int_x^\infty y|V(y)|\, dy \big] , \ x\geq 0 .
\end{equation*}
Furthermore,
\begin{equation*}
  |m_+(k;x)|
    \leq 1 + \int_x^\infty y |V(y)| |m_+(k;y)|\, dy - x\int_x^\infty |V(y)||m_+(k;y)|\, dy .
\end{equation*}
The first integral can be bounded uniformly in $x\in\R$
\begin{equation*}
\begin{split}
  \int_x^\infty y |V(y) |m_+(k;y)| \, dy
    & \leq \int_0^\infty y |V(y)| |m_+(k;y)| \, dy\\
    & \leq \int_0^\infty y_1 |V(y_1)| \exp\big[ \int_{y_1}^\infty y_2 |V(y_2)|\, dy_2\big] \, dy_1\\
    & = \exp\big[ \int_0^\infty y |V(y)|\, dy\big] - 1\\
    & \leq e^{\|XV\|_1} .
\end{split}
\end{equation*}
Thereby,
\begin{equation*}
\begin{split}
  \frac{|m_+(k;x)|}{1+|x|} 
     & \leq \frac{1}{1+|x|}\big[ 1 + e^{\|XV\|_1} \big] - \frac{x}{1+|x|} \int_x^\infty (1+|y|)|V(y)| \frac{|m_+(k;y)|}{1+|y|}\, dy \\
     & \leq 2 e^{\|XV\|_1} + \int_x^\infty (1+|y|)|V(y)| \frac{|m_+(k;y)|}{1+|y|}\, dy .
\end{split}
\end{equation*}
Lemma \ref{volterra} implies
\begin{equation*}
  \frac{|m_+(k;x)|}{1+|x|} 
    \leq 2 e^{\|XV\|_1} \exp\big[\int_x^\infty (1+|y|) |V(y)|\, dy\big]
    \leq 2 e^{\|V\|_1 + 2\|XV\|_1} ,
\end{equation*}
which yields \eqref{stl02t02}.
\end{proof}

We study the properties of the Jost solutions as functions of $k$.

\begin{lemma}\label{stl03t}
For each $x\in\R$, $m_+(\cdot,x)$ is analytic for $\im(k) > 0$ and continuous for $\im(k)\geq 0$, $k\neq 0$.
If, in addition, $XV\in L^1(\R)$ then $m(\cdot,x)$ is continuous for all $\im(k)\geq 0$. Furthermore,
\begin{equation}\label{stl03t01}
  |\dot m_+(k;x)| \leq \frac{2}{|k|} \exp\Big[\frac{2}{|k|} \|V\|_1\Big]\,\int_x^\infty (y-x)|V(y)|\, dy 
     ,\ \im(k)\geq 0,\ k\neq 0.
\end{equation}
\end{lemma}
\begin{proof}
See \cite[2. Lemma 1, (v), p.130]{DeiftTrubowitz1979}.
We differentiate \eqref{lippmann_schwinger_modified} by $k$ and obtain
\begin{equation*}
  \dot m_+(k;x) = r(x) + \int_x^\infty D_k(y-x)V(y)\dot m_+(k;y)\, dy,\ 
     r(x)\coloneqq \int_x^\infty \dot D_k(y-x)V(y)m_k(k;y)\, dy .
\end{equation*}
We integrate by parts
\begin{equation*}
  \dot D_k(x) = 2i \int_0^x e^{2iky} y\, dy
              = \frac{1}{k}\big[ e^{2iky}y\big]_0^x - \frac{1}{k}\int_0^x e^{2iky} \, dy
\end{equation*}
and obtain the estimate
\begin{equation*}
  |\dot D_k(x)| \leq \frac{1}{|k|} |x| + \frac{1}{|k|} |x| = \frac{2}{|k|}|x| .
\end{equation*}
Furthermore, using \eqref{stl01t01} we obtain
\begin{equation*}
  |r(x)| \leq \frac{2}{|k|} \int_x^\infty y_1|V(y_1)|\exp\big[ \frac{1}{|k|}\int_{y_1}^\infty |V(y_2)|\, dy_2\big]\, dy_1 \eqqcolon \tilde r(x) .
\end{equation*}
Note that $\tilde r(\infty)=0$. Using $e^x-1 \leq xe^x$ for $x\geq0$ we obtain
\begin{equation*}
\begin{split}
  -\tilde r'(x) 
    & = \frac{2}{|k|} \int_x^\infty |V(y_1)| \exp\big[\frac{1}{|k|}\int_{y_1}^\infty |V(y_2)|\, dy_2\big] \, dy_1\\
    & = 2 \big\{ \exp\big[ \frac{1}{|k|}\int_x^\infty |V(y)|\, dy\big] - 1 \big\}\\
    & \leq \frac{2}{|k|}\int_x^\infty |V(y)|\, dy \,\cdot\,\exp\big[ \frac{1}{|k|}\int_x^\infty |V(y)|\, dy\big] .
\end{split}
\end{equation*}
Since $\tilde k(x,y) \leq \frac{1}{|k|}|V(y)|$ we infer from Lemma \ref{volterra} that
\begin{equation*}
\begin{split}
  |\dot m_+(k;x)|
     & \leq \frac{2}{|k|} \int_x^\infty \int_{y_1}^\infty |V(y_2)|\, dy_2 \exp\big[\frac{1}{|k|}\int_{y_1}^\infty |V(y_3)|\, dy_3\big]
                   \exp\big[\frac{1}{|k|}\int_x^{y_1} |V(y_4)|\, dy_4\big] \, dy_1\\
     & \leq \frac{2}{|k|} \int_x^\infty \int_{y_1}^\infty |V(y_2)|\, dy_2\, dy_1 \exp\big[\frac{2}{|k|}\int_x^\infty |V(y)|\, dy\big] .
\end{split}
\end{equation*}
This implies \eqref{stl03t01}.
\end{proof}                           

The prescribed asymptotics in \eqref{stl01} and \eqref{stl02} imply that
\begin{equation}\label{stl04}
  \mathscr{t}(k)\psi_-(k;x) = \psi_+(-k;x) + \mathscr{r}_1(k) \psi_+(k;x) ,\
  \mathscr{t}(k)\psi_+(k;x) = \psi_-(-k;x) + \mathscr{r}_2(k) \psi_-(k;x) .
\end{equation}
A first consequence, which will be needed below, is a relation between the transmission coefficient and
a Wronski determinant
\begin{equation}\label{wronski01}
\begin{vmatrix}
  \psi_-(k,\cdot)  & \psi_+(k,\cdot) \\
  \psi_-'(k,\cdot) & \psi_+'(k,\cdot)
\end{vmatrix}
  = \frac{2ik}{\mathscr{t}(k)} .
\end{equation}
By computing the Wronski determinants of $u_\pm$ with $\varepsilon_{1,2}$, the solutions \eqref{deficiency_subspace}
of the unperturbed differential equation, one can further show that (cf. \cite[pp. 145,146]{DeiftTrubowitz1979})
\begin{equation}\label{s-matrix01}
\begin{aligned}
  \frac{1}{\mathscr{t}(k)}   & = 1 - \frac{1}{2ik} (e_1(k), Jf_1(k))
                & \frac{\mathscr{r}_1(k)}{\mathscr{t}(k)} & = \frac{1}{2ik} (e_1(k), Jf_2(k))\\
  \frac{\mathscr{r}_2(k)}{\mathscr{t}(k)} & = \frac{1}{2ik} ( e_2(k), Jf_1(k))
                & \frac{1}{\mathscr{t}(k)}   & = 1 - \frac{1}{2ik} ( e_2(k), Jf_2(k))
\end{aligned}
\end{equation}
where
\begin{gather*}
  e_1(k;x)\coloneqq\sqrt{|V(x)|}e^{ikx},\ e_2(k;x)\coloneqq\sqrt{|V(x)|}e^{-ikx},\\
  f_1(k;x)\coloneqq\sqrt{|V(x)|}\psi_+(k;x),\ f_2(k;x)\coloneqq\sqrt{|V(x)|}\psi_-(k;x) .
\end{gather*}
The $f_{1,2}(k)$ satisfy the symmetrized Lippmann--Schwinger equation
\begin{equation}\label{lippmann_schwinger_sym}
  f_1(k) =  e_1(k) + \frac{1}{k} \sqrt{|V|}G_+(k)\sqrt{|V|}J f_1(k),\
  f_2(k) =  e_2(k) + \frac{1}{k} \sqrt{|V|}G_-(k)\sqrt{|V|}J f_2(k) .
\end{equation}
The operators $G_\pm(k)$ are given through their kernels
\begin{equation*}
   G_+(k;x,y) \coloneqq \Theta(y-x) \sin(k(y-x)),\
   G_-(k;x,y) \coloneqq \Theta(x-y) \sin(k(x-y)) .
\end{equation*}
They differ from the resolvent by a rank one operator
\begin{equation*}
\begin{split}
   \sqrt{|V|}R_\infty^+(k^2)\sqrt{|V|}J 
      = - \frac{i}{2k} (Je_1(k),\cdot)e_1(k) + \frac{1}{k} \sqrt{|V|}G_+(k)\sqrt{|V|}J \\
      = - \frac{i}{2k} (Je_2(k),\cdot)e_2(k) + \frac{1}{k} \sqrt{|V|}G_-(k)\sqrt{|V|}J .
\end{split}
\end{equation*}
Thereby, the Lippmann--Schwinger equation can be rewritten as (cf. \eqref{Omega}, \eqref{Omega_boundary_values})
\begin{equation*}
   f_1(k) = ( 1 + \frac{i}{2k} (e_1(k), Jf_1(k)) ) \Omega_\infty^+(k^2)  e_1(k),\
   f_2(k) = ( 1 + \frac{i}{2k} (e_2(k), Jf_2(k)) ) \Omega_\infty^+(k^2)  e_2(k) .
\end{equation*}
Taking scalar products we can express $\mathscr{t}(k)$, $\mathscr{r}_{1,2}(k)$ through the matrix $\hat\Omega_{\infty}^+(k)$ from \eqref{Omega_matrix}.
Recall that $k=\sqrt{\nu}$. With the unitary scattering matrix
\begin{equation}\label{s-matrix02}
  \mathscr{S}(k) \coloneqq
\begin{pmatrix}
   \mathscr{t}(k)   & \mathscr{r}_1(k) \\
   \mathscr{r}_2(k) & \mathscr{t}(k)
\end{pmatrix}
  ,\ 
  |\mathscr{t}(k)|^2 + |\mathscr{r}_1(k)|^2 = 1 = |\mathscr{t}(k)|^2 + |\mathscr{r}_2(k)|^2,\ 
  \bar{\mathscr{t}}(k)\mathscr{r}_1(k) + \bar{\mathscr{r}}_2(k)\mathscr{t}(k) = 0 ,
\end{equation}
we finally obtain
\begin{equation}\label{Omega_T}
  \hat\Omega_\infty^+(k^2) = 2ik\mathscr{T}(k),\ \mathscr{T}(k)\coloneqq  \mathscr{S}(k)-\id,\ k\in\R ,
\end{equation}
where $\mathscr{T}(k)$ is the so-called T-matrix.
The transmission coefficient $\mathscr{t}(k)$ is, essentially, determined by its values for $k\in\R$.

\begin{lemma}[Faddeev--Deift--Trubowitz formula]\label{deift_trubowitz} 
Let $V\in L^1(\R)$ satisfy $X^2V\in L^1(\R)$ and 
let $-\beta_j^2$, $\beta_j>0$, $j=1,\ldots,n$, be the negative eigenvalues of $H_V$. Then,
\begin{equation}\label{deift_trubowitz01}
  \mathscr{t}(k) = \exp\Big[ \frac{1}{\pi i}\int_\R \frac{\ln|\mathscr{t}(u)|}{u-k}\, du \Big]
          \prod_{j=1}^n \frac{k+i\beta_j}{k-i\beta_j} ,\ \im(k)>0 .
\end{equation}
\end{lemma}
\begin{proof}
Cf. \cite[p. 323]{Faddeev1964} and \cite[p. 154]{DeiftTrubowitz1979}. 
Schwarz's integral formula for the half-plane expresses a holomorphic function through its real part.
Apply that to the function $k\mapsto\ln(\mathscr{t}(k))$.
\end{proof}

Moreover, the transmission coefficient can be expressed by the 
perturbation determinant, see \cite{JostPais1951}.

\begin{lemma}[Jost--Pais formula]\label{jost_pais}
Assume that the Birman--Schwinger operator $K(z)\in B_1(\hilbert)$. Then,
\begin{equation}\label{jost_pais01}
  \frac{1}{\mathscr{t}(k)} = \det(\id - K(z)),\ z=k^2,\ k>0 .
\end{equation}
In particular, for $z>0$
\begin{equation}\label{jost_pais02}
  |\det(\id-K(z))| \geq 1 .
\end{equation}
\end{lemma}
\begin{proof}
Following \cite[App. A]{Newton1980} we introduce a coupling parameter and study the function
\begin{equation*}
  D(\alpha) \coloneqq \det(\id - \alpha K),\ 0\leq \alpha\leq 1 .
\end{equation*}
Its logarithmic derivative is (cf. \cite[(1.7.10)]{Yafaev1992})
\begin{equation*}
  \frac{d}{d\alpha}\ln(D(\alpha)) = - \tr[ (\id - \alpha K)^{-1}K ] .
\end{equation*}
Since $K$ is a bounded operator a Neumann series argument shows that the inverse 
exists for $0\leq\alpha\leq\alpha_0$ with some $\alpha_0>0$.
We rewrite the operator
\begin{equation*}
  (\id -\alpha K)^{-1}K 
    = (\id - \alpha \sqrt{|V|}R_\infty\sqrt{|V|}J)^{-1} \sqrt{|V|}R_\infty\sqrt{|V|}J
    = \sqrt{|V|}( k^2\id - H-\alpha V)^{-1}\sqrt{|V|}J .
\end{equation*}
Note that the inverse exists even though it is not a bounded operator. Thus,
\begin{equation*}
  \frac{d}{d\alpha} \ln(D(\alpha)) = - \tr[ \sqrt{|V|}(k^2\id - H - \alpha V)^{-1}\sqrt{|V|}J ] .
\end{equation*}
This trace can be expressed through the Jost solutions. To this end, we need the Green function (cf. \eqref{wronski01})
\begin{equation*}
 ( k^2\id - H - \alpha V)^{-1}(x,y) = \frac{\mathscr{t}_\alpha}{2ik}
\begin{cases}
  \psi_+(x)\psi_-(y) & \text{for}\ x\geq y , \\
  \psi_-(x)\psi_+(y) & \text{for}\ x < y .
\end{cases}
\end{equation*}
Here, $\mathscr{t}_\alpha$ is the transmission coefficient corresponding to the potential $\alpha V$ (instead of $V$ as above). Thus,
\begin{equation*}
  \frac{d}{d\alpha} \ln(D(\alpha))
     = - \frac{\mathscr{t}_\alpha}{2ik} \int_\R \sqrt{|V(x)|}\psi_+(x)\psi_-(x)\sqrt{|V(x)|}J(x)\, dx
     = - \frac{\mathscr{t}_\alpha}{2ik} \int_\R f_1(x)f_2(x) J(x)\, dx .
\end{equation*}
On the other hand, we have
\begin{equation*}
  \frac{1}{\mathscr{t}_\alpha} = 1 - \frac{\alpha}{2ik} (e_1,Jf_1),\ f_{1,2} = e_{1,2} + \alpha\sqrt{|V|}G_\pm \sqrt{|V|}J f_{1,2}
\end{equation*}
and consequently
\begin{equation*}
  - \frac{\dot{\mathscr{t}}_\alpha}{\mathscr{t}_\alpha^2} = - \frac{1}{2ik}(e_1,J(f_1+\alpha\dot f_1))
\end{equation*}
where the dot denotes differentiation with respect to $\alpha$. Differentiating the Lippmann--Schwinger equation
one obtains after some simple calculations
\begin{equation*}
  f_1 + \alpha \dot f_1 = (\id - \alpha\sqrt{|V|}G_+ \sqrt{|V|}J)^{-1} f_1 .
\end{equation*}
Therefore,
\begin{equation*}
  \frac{\dot{\mathscr{t}}_\alpha}{\mathscr{t}_\alpha^2}
    = \frac{1}{2ki} ( e_1, J(\id-\alpha\sqrt{|V|}G_+\sqrt{|V|}J)^{-1}f_1)
    = \frac{1}{2ik} ((\id - \alpha\sqrt{|V|}G_-\sqrt{|V|}J)^{-1}e_1,J f_1)
    = \frac{1}{2ik} ( \bar f_2,Jf_1)
\end{equation*}
where we used $e_1 = \bar e_2$. We conclude that for $0\leq\alpha\leq \alpha_0$,
\begin{equation*}
  \frac{d}{d\alpha}\ln(D(\alpha)) = - \frac{d}{d\alpha}\ln(\mathscr{t}_\alpha) .
\end{equation*}
For $\alpha=0$ both quantities have the same value. Furthermore, since $D(\alpha_0)=\frac{1}{\mathscr{t}_{\alpha_0}}\neq 0$
we can extend the result to $0\leq\alpha\leq 1$ which proves \eqref{jost_pais01}.
The estimate \eqref{jost_pais02} follows immediatetly from $|\mathscr{t(k)}|\leq 1$ for $k>0$.
\end{proof}

\subsection{Scattering on the half-line\label{sthl}}
On the half-line there cannot be a transmission cofficient since the wave is entirely reflected at the origin. 
Hence, conditions \eqref{stl01} and \eqref{stl02} are to be replaced by (cf. \eqref{hl01})
\begin{equation}\label{sthl01}
  u(x) \sim e^{-ikx} - \frac{ia-kb}{ia+kb} S(k)e^{ikx}\ \text{as}\ x\to\infty,\ a u(0) - bu'(0) =0 ,
\end{equation}
which reads in terms of the Jost solutions
\begin{equation}\label{sthl02}
  u(x) = \psi(-k;x) - \frac{ia-kb}{ia+kb} S(k)\psi(k;x) .
\end{equation}
Conditions \eqref{sthl01} and \eqref{sthl02} have been chosen so that $S(k)=1$ when $V=0$ (cf. \eqref{hl02}).
The Lippmann--Schwinger equation is the same only this time it is restricted to the positive axis
\begin{equation}\label{sthl03}
  \psi(k;x) = e^{ikx} + \frac{1}{k}\int_x^\infty \sin(k(y-x))V(y) \psi(k;y)\, dy,\ x\geq 0 .
\end{equation}
Note that these $\psi$ are the restrictions of the Jost solutions $\psi_+$ defined on the line.
As in Section \ref{stl} it is convenient to work with the symmetrized Lippmann--Schwinger equation
\begin{equation}\label{sthl04}
  f_j = e_j + \frac{1}{k}\sqrt{|V|}G_+\sqrt{|V|} f_j,\ j=1,2,
\end{equation}
with
\begin{equation*}
  f_1(x) \coloneqq \sqrt{|V(x)|}\psi(k;x),\ f_2(x) \coloneqq \sqrt{|V(x)|}\psi(-k;x) .
\end{equation*}
Note that $f_2$ is different from that on the whole line.
Using the boundary conditions one can derive the scattering matrix from \eqref{sthl02}
\begin{equation*}
  S(k) = c \frac{a\psi(-k;0) - b\psi'(-k;0)}{a\psi(k;0)-b\psi'(k;0)},\ 
    c \coloneqq \frac{ia + kb}{ia-kb} = \frac{i\bar a + k\bar b}{i\bar a-k\bar b},\
    |c|^2 = 1 .
\end{equation*}
The values at $0$ can be obtained via the Lippmann--Schwinger equation \eqref{sthl03} and \eqref{sthl04}
\begin{equation*}
  \psi(k;0) = 1 + \frac{1}{k}(e_s,Jf_1),\ \psi'(k;0) = ik - (e_c,Jf_1),\
  \psi(-k;0) = 1 + \frac{1}{k}(e_s,Jf_2),\ \psi'(-k;0) = -ik - (e_c,Jf_2)
\end{equation*}
where
\begin{equation*}
  e_s(x) \coloneqq \sqrt{|V(x)|}\sin(kx),\ e_c(x)\coloneqq\sqrt{|V(x)|}\cos(kx) .
\end{equation*}
We rewrite the scattering matrix
\begin{equation*}
  S(k) = c \frac{ a+ikb + \frac{1}{k}(\bar a e_s+k\bar b e_c,Jf_2) }{ a-ikb + \frac{1}{k}(\bar a e_s+k\bar b e_c,Jf_1)}
       = \frac{i-\frac{1}{2k}(e_0,Jf_2)}{i-\frac{1}{2kc}(e_0,Jf_1)}
\end{equation*}
where (see \eqref{hl02})
\begin{equation*}
  e_0\coloneqq \sqrt{|V|}\varepsilon_0,\ \bar a e_s + k\bar b e_c = - \frac{1}{2}(i\bar a+k\bar b)e_0 .
\end{equation*}
The T-matrix is then
\begin{equation*}
  T(k) = S(k) - 1 = \frac{1}{2ki} \frac{(e_0,Jf_0)}{1 + \frac{i}{2kc}(e_0,Jf_1)},\
  f_0 \coloneqq \frac{1}{c}f_1 - f_2 .
\end{equation*}
By linearity, $f_0$ satisfies the Lippmann--Schwinger equation with $e_0$ instead of $e_{1,2}$.

A look at the kernel functions shows that the resolvent is a rank one perturbation of the Lippmann--Schwinger operator
which yields for the Birman--Schwinger operator
\begin{equation*}
  \sqrt{|V|}R_\infty\sqrt{|V|}J 
     = \frac{i}{2k} (J\bar e_0,\cdot) e_1 + \frac{1}{k}\sqrt{|V|}G_+\sqrt{|V|}J
     = - \frac{i}{2ku} (Je_0,\cdot) e_1 + \frac{1}{k}\sqrt{|V|}G_+\sqrt{|V|}J .
\end{equation*}
Thereby, the Lippmann--Schwinger equation for $f_0$ gives
\begin{equation*}
  f_0 = \Omega_\infty e_0 + \frac{i}{2kc} (e_0,Jf_0)\Omega_\infty e_1 .
\end{equation*}
Taking scalar products one obtains
\begin{equation*}
  (1-\frac{i}{2kc}(e_0,J\Omega_\infty e_1))(e_0,Jf_0) = (e_0,J\Omega_\infty e_0) .
\end{equation*}
Likewise for $f_1$
\begin{equation*}
  f_1 = (1 + \frac{i}{2kc}(e_0,Jf_1)) \Omega_\infty e_1 .
\end{equation*}
Once again, we take scalar products and rearrange the terms 
\begin{equation*}
  (1+\frac{i}{2kc}(e_0,Jf_1))(1-\frac{i}{2kc}(e_0,J\Omega_\infty e_1) = 1 .
\end{equation*}
We use the formulae for $f_0$ and $f_1$ to obtain
\begin{equation*}
  T(k) = \frac{1}{2ki} \frac{(e_0,J\Omega_\infty e_0)}{(1-\frac{i}{2kc}(e_0,J\Omega_\infty e_1))(1 + \frac{i}{2kc}(e_0,Jf_1))}
       = \frac{1}{2ki} (e_0,J\Omega_\infty e_0) .
\end{equation*}
Finally, we have established the relation between the wave operator \eqref{Omega_matrix} and the T-matrix
\begin{equation*}
  T(k) = \frac{1}{2ki} \hat\Omega_\infty(\nu) .
\end{equation*}
Recall that $k=\sqrt{\nu}$.

\subsection{Spectral shift function \texorpdfstring{$\xi$}{}\label{ssf}}
We collect some properties of the spectral shift function 
(see e.g. \cite[Ch.~8]{Yafaev1992}, \cite[Ch.~0~\S~9]{Yafaev2010}).
As our definition we use
\begin{equation}\label{krein_xi}
  \xi(\nu) \coloneqq \frac{1}{\pi} \lim_{y\to +0} \im\ln[\Delta(\nu+iy)] ,
\end{equation}
which is also known as Kre\u\i{}n's formula. If $H$ and $H_V$ are semi-bounded from below and
have no essential spectrum then the spectral shift function can be expressed by Lifshitz's formula
\begin{equation}\label{lifshitz}
  \xi(\nu) = -\tr(\Pi_\nu - P_\nu),\ \nu\in\R .
\end{equation}
Note that in general the difference of the spectral projections is not trace class.

\begin{lemma}\label{ssf01t}
Let $R_V(z)-R(z)\in B_1(\hilbert)$ for $z\notin\sigma(H_V)\cup\sigma(H)$. Then the following hold true.
\begin{enumerate}
\item $\xi\in L^1(\R; (1+|\lambda|)^{-2})$.
\item Assume in addition that there are constants $C\geq 0$ and $\alpha>0$ such that
for some $\re(z)$ and all sufficiently large $\im(z)>0$
\begin{equation*}
  \| R_V(z) - R(z) \|_1 \leq \frac{C}{\im(z)^\alpha} .
\end{equation*}
Then $\xi\in L^1(\R; (1+|\lambda|)^{-\tau-1})$ for all $\tau > \max\{-1, 1-\alpha\}$.
\end{enumerate}
\end{lemma}
\begin{proof}
1. 
See \cite[Thm. 4.1]{SinhaMohapatra1994}.
2. 
See \cite[(8.8.3)]{Yafaev1992}.
\end{proof}

Under certain conditions \eqref{krein_xi} can be inverted in that the perturbation determinant
can be expressed via the spectral shift function.

\begin{lemma}\label{ssf02t}
Assume that $\xi\in L^1(\R;(1+|\lambda|)^{-1})$. Then,
\begin{equation*}
  \ln[\Delta(z)] = \int_{-\infty}^\infty \frac{\xi(\lambda)}{\lambda-z}\, d\lambda,\ \im(z)\neq 0,
\end{equation*}
and $\xi$ is given via \eqref{krein_xi}.
\end{lemma}
\begin{proof}
See \cite[(0.9.38)]{Yafaev2010} and \cite[(0.9.39)]{Yafaev2010}.
\end{proof}

The spectral shift function is related to the scattering matrix (cf. Section \ref{scattering}).

\begin{lemma}[Birman--Kre\u\i{}n formula]\label{ssf03t}
Let $R_V(z)-R(z)\in B_1(\hilbert)$ for $z\notin\sigma(H_V)\cup\sigma(H)$. Then,
\begin{equation}\label{ssf03t01}
  \det(S(\lambda)) = e^{-2\pi i\xi(\lambda)}
\end{equation}
where $S(\lambda)$ is the scattering matrix at energy $\lambda\in\R$ for the operators $H_V$ and $H$.
\end{lemma}
\begin{proof}
See \cite[Thm. 0.9.4]{Yafaev2010}.
\end{proof}

%
%
\section*{Acknowledgement}
It is a pleasure to thank M. Gebert, F. Gesztesy, H.-K. Kn\"orr,
P. M\"uller, K. Pankrashkin, and G. Raikov for discussions leading to this work
and A. Pushnitski for making us aware of his paper \cite{FrankPushnitski2015}.

\bibliographystyle{abbrv}
\bibliography{/home/otte/Bibliothek/BibTeX/aoc}
\end{document}